\documentclass[11pt,a4paper]{article}

\usepackage{fontspec}
\usepackage[titletoc,title,header]{appendix}

\usepackage{amsmath}
\usepackage{amssymb}
\usepackage{amsthm}
\usepackage{mathtools}
\usepackage{xparse}    

\usepackage[margin=1in]{geometry}

\usepackage{graphicx}
\usepackage{xcolor}

\usepackage{hyperref}
\hypersetup{
  colorlinks=true,
  linkcolor=blue,
  citecolor=blue,
  urlcolor=blue,
  pdfauthor={Yoshitsugu Sekine},
  pdftitle={Mean-Field Bose--Einstein Condensation and Condensate Ideals in the Resolvent Algebra}
}

\usepackage[]{natbib}
\theoremstyle{plain}
\newtheorem{thm}{Theorem}[section]
\AtEndEnvironment{thm}{\qed}
\newtheorem{lem}[thm]{Lemma}
\AtEndEnvironment{lem}{\qed}
\newtheorem{prop}[thm]{Proposition}
\AtEndEnvironment{prop}{\qed}

\AtEndEnvironment{cor}{\qed}

\AtEndEnvironment{conj}{\qed}

\AtEndEnvironment{fact}{\qed}

\theoremstyle{definition}
\newtheorem{defn}[thm]{Definition}
\AtEndEnvironment{defn}{\qed}

\AtEndEnvironment{ex}{\qed}

\theoremstyle{remark}
\newtheorem{rem}[thm]{Remark}
\AtEndEnvironment{rem}{\qed}

\makeatletter
\providecommand*{\dashv}{\mathrel{\mathpalette\@Dashv\vDash}}
\newcommand*{\@dashv}[2]{\reflectbox{$\m@th#1#2$}}
\makeatother
\newcommand{\absvol}[1]{\abs{#1}} 
\newcommand{\abs}[1]{\left| #1 \right|} 
\NewDocumentCommand{\weaknorm}{O{\dbk} m}{#1{#2}} 
\newcommand{\norm}[1]{\left\Vert #1 \right\Vert} 
\newcommand{\twonorm}[1]{\norm{#1}_2}
\newcommand{\bkt}[2]{\left\langle #1,\,#2 \right\rangle} 
\newcommand{\rbkt}[2]{\left( #1,\,#2 \right)} 
\newcommand{\isomto}{\mathrel{\rightarrowtail\kern-1.9ex\twoheadrightarrow}} 
\newcommand{\cbk}[1]{\left\{ #1 \right\}} 
\newcommand{\dbk}[1]{\left\langle #1 \right\rangle} 
\newcommand{\rbk}[1]{\left( #1 \right)} 
\newcommand{\sqbk}[1]{\left[ #1 \right]} 
\newcommand{\funcond}[3]{\fun{#1}{#2 \middle| #3}} 
\newcommand{\fun}[2]{#1 \rbk{#2}} 
\newcommand{\sqfun}[2]{#1 \sqbk{#2}} 
\newcommand{\closedinterval}[2]{\sqbk{#1,\,#2}} 
\newcommand{\openinterval}[2]{\rbk{#1,\,#2}} 
\newcommand{\commutator}[2]{\sqbk{#1,\,#2}} 
\newcommand{\ket}[1]{\left| #1 \right\rangle} 

\DeclareMathOperator{\sgn}{sgn} 
\NewDocumentCommand{\imunit}{O{\mathsf{i}}}{#1} 
\NewDocumentCommand{\placeholder}{O{\bullet}}{#1} 
\NewDocumentCommand{\trace}{O{\operatorname{Tr}}}{#1} 
\newcommand{\Ker}{\operatorname{Ker}} 
\newcommand{\bigmiddleslash}[2]{\left. #1 \middle/ #2 \right.} 
\newcommand{\clos}[1]{\overline{#1}} 
\newcommand{\cmpconj}[1]{\overline{#1}} 
\newcommand{\diracdelta}{\delta} 
\newcommand{\eqcsq}[1]{\sqbk{#1}} 
\newcommand{\fml}[2]{\cbk{#1}_{#2}} 
\newcommand{\invrbk}[1]{\rbk{#1}^{-1}} 
\newcommand{\kroneckerdelta}{\delta} 
\newcommand{\laplacian}{\Delta} 
\newcommand{\napiernum}{\mathsf{e}} 
\newcommand{\onehalf}{\frac{1}{2}} 
\newcommand{\oneoverfour}{\frac{1}{4}} 
\newcommand{\opchern}[1]{\operatorname{ch}} 
\newcommand{\opideal}{\operatorname{Ideal}} 
\newcommand{\opimag}{\operatorname{Im}} 
\newcommand{\opreal}{\operatorname{Re}} 
\newcommand{\setSymbolDownLeft}[2]{{\vphantom{#2}}_{#1}{#2}} 
\newcommand{\setSymbolUpLeft}[2]{{\vphantom{#2}}^{#1}{#2}} 
\DeclareMathOperator{\eqalgisom}{\cong} 
\NewDocumentCommand{\agvariety}{O{\mathcal}}{#1} 
\NewDocumentCommand{\cmdrel}{O{\omega}}{#1} 
\NewDocumentCommand{\dfsp}{O{A}}{#1} 
\NewDocumentCommand{\eqcpointed}{O{\eqcsq} m}{#1{#2}_{\ast}} 
\NewDocumentCommand{\fnheaviside}{O{H}}{#1} 
\NewDocumentCommand{\fthol}{O{\mathcal{O}}}{#1} 
\NewDocumentCommand{\ftmero}{O{\mathcal{M}}}{#1} 
\NewDocumentCommand{\grcentralizer}{O{Z}}{#1} 
\NewDocumentCommand{\grmetform}{O{2} m m}{\grmet[#1] \! \rbkt{#2}{#3}} 
\NewDocumentCommand{\grmet}{O{2}}{\setSymbolDownLeft{#1}{g}} 
\NewDocumentCommand{\grnormalizer}{O{N}}{#1} 
\NewDocumentCommand{\gropasym}{O{A}}{#1} 
\NewDocumentCommand{\gropsym}{O{S}}{#1} 
\NewDocumentCommand{\grpermorderedpair}{O{\mathcal{P}}}{#1} 
\NewDocumentCommand{\grsym}{O{\mathfrak{S}} m}{#1_{#2}} 
\NewDocumentCommand{\gtbase}{O{\mathcal}}{#1} 
\NewDocumentCommand{\gtfilter}{O{\mathcal}}{#1} 
\NewDocumentCommand{\gtfmlclosed}{O{\mathcal}}{#1} 
\NewDocumentCommand{\gtfmlopen}{O{\mathcal}}{#1} 
\NewDocumentCommand{\gtopenball}{O{U}}{#1} 
\NewDocumentCommand{\gtopencover}{O{\mathcal}}{#1} 
\NewDocumentCommand{\gtopennbh}{O{\mathcal}}{#1} 
\NewDocumentCommand{\gtpreopencover}{O{\mathcal}}{#1} 
\NewDocumentCommand{\gtsubbase}{O{\mathcal}}{#1} 
\NewDocumentCommand{\gtvicinity}{O{\mathcal}}{#1} 
\NewDocumentCommand{\lasp}{O{\mathcal}}{#1} 
\NewDocumentCommand{\latprightrbk}{O{\top} m}{\rbk{#2}^{#1}} 
\NewDocumentCommand{\latpright}{O{\top} m}{#2^{#1}} 
\NewDocumentCommand{\latp}{O{t} m}{\setSymbolUpLeft{#1}{#2}} 
\NewDocumentCommand{\lpdistribution}{O{\mu} m}{#2_{\ast,#1}} 
\NewDocumentCommand{\lpmollifier}{O{\rho}}{#1} 
\NewDocumentCommand{\lpofpositive}{O{\chi}}{#1} 
\NewDocumentCommand{\manliederiv}{O{L}}{#1} 
\NewDocumentCommand{\mansmoothnbh}{O{\mathcal}}{#1} 
\NewDocumentCommand{\mblfmldsysgenerated}{O{d} m}{\fun{#1}{#2}} 
\NewDocumentCommand{\mblfmlgenerated}{O{\sigma} m}{\fun{#1}{#2}} 
\NewDocumentCommand{\oacorrfn}{O{\Gamma}}{#1} 
\NewDocumentCommand{\oagnsvector}{O{\Omega}}{#1} 
\NewDocumentCommand{\oaideal}{O{\mathcal}}{#1} 
\NewDocumentCommand{\oanumberoperator}{O{A}}{#1} 
\NewDocumentCommand{\oaposcone}{O{\mathcal{P}}}{#1} 
\NewDocumentCommand{\oapressure}{O{P}}{#1} 
\NewDocumentCommand{\oarepn}{O{\pi}}{#1} 
\NewDocumentCommand{\oaspnormalstate}{O{N}}{#1} 
\NewDocumentCommand{\oasppurestate}{O{P}}{#1} 
\NewDocumentCommand{\oaspstate}{O{E}}{#1} 
\NewDocumentCommand{\oastatevector}{O{\Omega}}{#1} 
\NewDocumentCommand{\oastate}{O{\omega}}{#1} 
\NewDocumentCommand{\opdilation}{O{\delta}}{#1} 
\NewDocumentCommand{\opdmat}{O{\rho}}{#1} 
\NewDocumentCommand{\opfockan}{O{a}}{#1} 
\NewDocumentCommand{\opfockcran}{O{a}}{#1^{\#}} 
\NewDocumentCommand{\opfockcrdagger}{O{a}}{#1^{\dagger}} 
\NewDocumentCommand{\opfockcr}{O{a}}{#1^{\ast}} 
\NewDocumentCommand{\opfocknumber}{O{N}}{#1} 
\NewDocumentCommand{\opfocksegalconj}{O{\pi}}{#1} 
\NewDocumentCommand{\opfocksegal}{O{\phi}}{#1} 
\NewDocumentCommand{\opspecmeas}{O{E}}{#1} 
\NewDocumentCommand{\opspec}{O{} m}{\fun{\sigma_{#1}}{#2}} 
\NewDocumentCommand{\optransl}{O{\tau}}{#1} 
\NewDocumentCommand{\physaction}{O{\mathcal{A}}}{#1} 
\NewDocumentCommand{\physcharge}{O{e}}{\mathrm{#1}} 
\NewDocumentCommand{\physcplconst}{O{\mathsf{g}}}{#1} 
\NewDocumentCommand{\physelectrostaticcapasity}{O{\mathrm{Cap}}}{#1} 
\NewDocumentCommand{\physenergy}{O{E}}{#1} 
\NewDocumentCommand{\physgse}{O{E}}{#1_{0}} 
\NewDocumentCommand{\physham}{O{H}}{#1} 
\NewDocumentCommand{\physlagdensity}{O{\mathcal{L}}}{#1} 
\NewDocumentCommand{\physlag}{O{L}}{#1} 
\NewDocumentCommand{\physliouvilean}{O{L}}{#1} 
\NewDocumentCommand{\physmass}{O{m}}{#1} 
\NewDocumentCommand{\prbbrownmv}{O{B}}{#1} 
\NewDocumentCommand{\prbcharfun}{O{\chi}}{#1} 
\NewDocumentCommand{\prbdist}{O{\mathcal{P}}}{#1} 
\NewDocumentCommand{\prbgaussianmeasure}{O{\msrcal{N}}}{#1} 
\NewDocumentCommand{\prbnormaldist}{O{N}}{#1} 
\NewDocumentCommand{\prbpoissonprocess}{O{N}}{#1} 
\NewDocumentCommand{\prbprocess}{O{X}}{#1} 
\NewDocumentCommand{\prbqspace}{O{\mathcal{Q}}}{#1} 
\NewDocumentCommand{\prbspsample}{O{\Omega}}{#1} 
\NewDocumentCommand{\psh}{O{\mathfrak}}{#1} 
\NewDocumentCommand{\qtquantumchannel}{O{\mathcal{L}}}{#1} 
\NewDocumentCommand{\repn}{O{\pi}}{#1} 
\NewDocumentCommand{\schattencls}{O{\mathbb{K}}}{#1} 
\NewDocumentCommand{\setfmlcylinder}{O{\mathcal{C}}}{#1} 
\NewDocumentCommand{\setfml}{O{\mathcal}}{#1} 
\NewDocumentCommand{\setindex}{O{\mathcal} m}{#1{#2}} 
\NewDocumentCommand{\setlattice}{O{\Gamma}}{#1} 
\NewDocumentCommand{\setspecial}{O{\mathcal} m}{#1{#2}} 
\NewDocumentCommand{\shdiffform}{O{\sheaf{A}}}{#1} 
\NewDocumentCommand{\sheaf}{O{\mathfrak}}{#1} 
\NewDocumentCommand{\smchemicalpotential}{O{\mu}}{#1} 
\NewDocumentCommand{\smenergydensity}{O{\varrho}}{#1} 
\NewDocumentCommand{\smfluctuationwithdmat}{O{\beta} m}{\smuncertaintywithdmat[#1]{#2}^2} 
\NewDocumentCommand{\sminvtemperature}{O{\beta}}{#1} 
\NewDocumentCommand{\smlocaldensityoperator}{O{\rho}}{#1} 
\NewDocumentCommand{\smmicrocanonicalstate}{O{\beta} m}{\physmean{#2}_{#1}} 
\NewDocumentCommand{\smnumberdensity}{O{\rho}}{#1} 
\NewDocumentCommand{\smooth}{O{\mathcal{E}}}{#1} 
\NewDocumentCommand{\smparticlenumber}{O{N}}{#1} 
\NewDocumentCommand{\smpressure}{O{p}}{#1} 
\NewDocumentCommand{\smspecificfreeenergy}{O{\bar{f}}}{#1} 
\NewDocumentCommand{\smthermalvac}{O{\beta}}{\Omega_{#1}} 
\NewDocumentCommand{\smuncertaintywithdmat}{O{\beta} m}{\rbk{\triangle #2}_{#1}} 
\NewDocumentCommand{\sphilb}{O{\mathcal}}{#1} 
\NewDocumentCommand{\splowerhalf}{O{\mathbb{H}}}{#1_{\txtneg}} 
\NewDocumentCommand{\spupperhalf}{O{\mathbb{H}}}{#1_{\txtnonneg}} 
\NewDocumentCommand{\topmetric}{O{d}}{#1} 
\NewDocumentCommand{\vaoutnormal}{O{\widehat}}{#1} 
\newcommand{\category}[1]{\mathop{\mathsf{#1}}} 

\newcommand{\catpresheaf}[1]{\category{PSh}} 
\newcommand{\faadj}[1]{#1^{\ast}} 
\newcommand{\faftr}[1]{\widehat{#1}} 
\newcommand{\fldcmp}{\fld{C}} 
\newcommand{\fldreal}{\fld{R}} 
\newcommand{\fld}[1]{\mathbb{#1}} 
\newcommand{\fndef}[1]{\boldsymbol{1}_{#1}} 
\newcommand{\fnexp}[1]{\fun{\exp}{#1}} 
\newcommand{\fnrestr}[2]{\left. #1 \right|_{#2}} 
\newcommand{\gtclos}[1]{\overline{#1}} 
\newcommand{\lp}{L} 
\newcommand{\mblfmlborel}{\mblfml{B}} 
\newcommand{\mblfml}[1]{\mathcal{#1}} 
\newcommand{\monnat}{\mathbb{N}} 
\newcommand{\msrbb}[1]{\mathbb{#1}} 
\newcommand{\msrcal}[1]{\mathcal{#1}} 

\newcommand{\oacenter}{Z} 
\newcommand{\oacstar}{C^{\ast}} 
\newcommand{\oadoublecommutant}[1]{#1^{\prime \prime}} 
\newcommand{\oaresolventalgebra}{\oa{R}} 
\newcommand{\oaresolvent}{R} 
\newcommand{\oa}[1]{\mathcal{#1}} 
\newcommand{\opdmsr}[1]{\mathop{d #1}} 
\newcommand{\opfnresolvent}[1]{\invrbk{#1}} 
\newcommand{\opfocksndqntdiff}{d \Gamma} 
\newcommand{\opfocksndqnt}{\Gamma} 
\newcommand{\opfockweyl}{W} 
\newcommand{\opformdomain}{\mathop{Q}} 
\newcommand{\opform}[1]{\mathsf{#1}} 
\newcommand{\opspecint}[1]{\mathcal{E}} 
\newcommand{\physmean}[1]{\dbk{#1}} 
\newcommand{\prbcov}{\mathrm{Cov}} 
\newcommand{\prbexp}{\mathbb{E}} 
\newcommand{\ringratint}{\mathbb{Z}} 

\newcommand{\seq}[2]{\if\relax\detokenize{#1}\relax \rbk{#1} \else \rbk{#1}_{#2} \fi} 
\newcommand{\setisomorphism}[1]{\operatorname{Iso}} 
\newcommand{\setone}[1]{\cbk{#1}}
\newcommand{\setquot}[2]{\bigmiddleslash{#1}{#2}} 
\newcommand{\set}[2]{\left\{#1 \, \middle| \, #2\right\}}
\newcommand{\smfreeenergy}{F} 
\newcommand{\smpartitionfunc}{Z} 
\newcommand{\spfock}{\mathcal{F}} 
\newcommand{\splinspan}{\operatorname{span}} 
\newcommand{\txtbec}{\mathrm{BEC}} 
\newcommand{\txtbsn}{\mathrm{b}} 
\newcommand{\txtcritical}{\mathrm{c}} 
\newcommand{\txteff}{\mathrm{eff}} 
\newcommand{\txtfock}{\mathrm{F}} 
\newcommand{\txtfr}{\mathrm{fr}} 
\newcommand{\txtloc}{\mathrm{loc}} 
\newcommand{\txtloop}{\mathrm{loop}} 
\newcommand{\txtmeanfield}{\mathrm{mf}} 
\newcommand{\txtneg}{\mathrm{-}} 
\newcommand{\txtnonneg}{\mathrm{+}} 
\newcommand{\txtnonzero}{\mathrm{nz}} 
\newcommand{\txtparticle}{\mathrm{p}} 
\newcommand{\txtphys}{\mathrm{phys}} 
\newcommand{\txtprob}{\mathrm{Prob}} 
\newcommand{\txtreal}{\mathrm{real}} 
\newcommand{\txtregular}{\mathrm{reg}} 
\newcommand{\txtsym}{\mathrm{s}} 

\title{Mean-Field Bose--Einstein Condensation and Condensate Ideals in the Resolvent Algebra}

\author{%
Yoshitsugu Sekine\\{\small\texttt{4429sekine@gmail.com}}%
}

\date{\today}

\begin{document}

\maketitle

\begin{abstract}
This paper studies the imperfect Bose gas after the Kac density law and the mean-field Euler equations have selected a condensed density with positive zero-mode excess. In this BEC regime the selected chemical potential cancels the mean-field shift, so the selected one-particle Hamiltonian is exactly the free one. The resulting zero-mode covariance defines a mean-field BEC ideal in the resolvent algebra, while the nonregular quotient and the direct-integral center record distinct representation-theoretic data. Occupation-number and Brownian-loop formulations recover the same density selection, excess density, ODLRO data, local tests, and the separation between finite-density BEC and Buchholz's stricter infinite-occupation proper-condensate criterion.

\noindent\textbf{Keywords:} resolvent algebra, mean-field Bose gas, BEC, Brownian loops, proper condensates
\end{abstract}

\setcounter{tocdepth}{3}
\tableofcontents

\section{Introduction}\label{introduction}

This paper studies the imperfect Bose gas as a test case for what the resolvent algebra records once Bose--Einstein condensation is already supplied by the thermodynamic model. The finite-volume Gibbs states give an occupation-number probability law. The Kac density law selects the limiting density \(\bar{\smnumberdensity}_{\txtbsn}\), and the Euler equations identify the condensed alternative by the positive zero-mode excess \(\smnumberdensity_{\txtbsn,0}(\sminvtemperature)>0\). The aim is not to derive this phase transition from the resolvent algebra alone, but to locate precisely the algebraic data corresponding to that BEC input.

In the condensed case the selected chemical potential satisfies \(\smchemicalpotential_{\mathrm{sel}}
=\lambda\bar{\smnumberdensity}_{\txtbsn}\). Consequently the selected one-particle Hamiltonian \(\physham[h]_{\bar{\smnumberdensity}_{\txtbsn}}\) reduces exactly to the free one-particle Hamiltonian; see \eqref{expedition0012694}. This reduction is the basis for the component KMS states used in the resolvent-algebra analysis.

The resolvent-algebra formulation starts from the zero-mode point-mass covariance \eqref{expedition0012440}. It defines the mean-field BEC ideal \eqref{expedition0019398}, identifies that ideal as the kernel of the nonregular BEC representation in Proposition \ref{expedition0012642}, and separates this quotient picture from the represented center of the direct-integral state in Proposition \ref{expedition0012651}. Thus finite-density BEC, quotient regularity, and represented central phase parameters are distinct pieces of structure.

The same model is then rewritten in two probabilistic languages. The occupation-number formulation proves the same density selection, zero-mode excess, direct-integral component law, ODLRO covariance, local condensate tests, and density-fluctuation remainder without using the ideal structure. The Brownian-loop formulation gives the corresponding winding-number picture: total winding collapses to the selected density, macroscopic winding is the zero-mode excess, and open paths recover the usual one-particle density-matrix criterion.

Theorem \ref{expedition0019672} summarizes these equivalences. Sections \ref{expedition0012743}, \ref{expedition0019615}, and \ref{expedition0019600} prove the resolvent-algebra, occupation-number, and Brownian-loop formulations, respectively. The comparison also separates finite-density BEC from Buchholz's stricter proper-condensate criterion: the latter requires local occupation to escape to infinity in a primary component or in a limiting family. The same strategy will be applied later to the further concrete examples in \cite{AndreVerbeure001}; after those cases are analyzed, the ideal structures extracted from them can be compared systematically.

\section{Main Results}\label{expedition0012133}

The mean-field Bose gas is studied through its finite-volume Gibbs states, Kac density selection, and the resulting zero-mode singular covariance. The common resolvent-algebra, finite-volume, and local notation is fixed first. The main conclusions are then summarized for three equivalent viewpoints: the resolvent-algebra ideal and represented center, the occupation-number probability law, and the Brownian-loop formulation.

\subsection{Space Setting and Basic Operators}\label{space-setting-and-basic-operators}

In this paper the spatial dimension is mainly \(d\geq3\), unless a model-specific restriction is stated explicitly. The one-particle Hilbert space is \[\sphilb{H}_{\txtparticle}
=
\fun{\lp^{2}}{\fldreal^{d}},
\quad
\sphilb{H}_{\txtparticle,\txtreal}
=
\fun{\lp_{\txtreal}^{2}}{\fldreal^{d}}
=
\fun{\lp^{2}}{\fldreal^{d};\fldreal}.\] On this space the basic particle Hamiltonian has the form \begin{equation}\label{expedition0019396}
\begin{aligned}
\physham[h]_{\txtparticle}
=
\physham[h]_{\txtparticle,V}
=
\physham[h]_{\txtparticle,0} + V,
\quad
\physham[h]_{\txtparticle,0}
=
-\onehalf\laplacian.
\end{aligned}
\end{equation}

The real symplectic space is denoted by \((X,\sigma)\), and in the Bose particle applications below \(X\) is the realification \(\sphilb{H}_{\txtparticle,\txtreal}\). The symplectic form is \[\sigma(f,g)=\opimag\bkt{f}{g}_{\sphilb{H}_{\txtparticle}},\] and the real Hilbert inner product is \(\opreal\bkt{f}{g}_{\sphilb{H}_{\txtparticle}}\).

For inverse temperature \(\sminvtemperature>0\) and a chemical potential in the usual free trace-class domain, the free Bose one-particle covariance is recorded by the operator \(K_{\sminvtemperature,\smchemicalpotential}
=
\coth
\frac{\sminvtemperature
\rbk{h_{\txtparticle}-\smchemicalpotential}}
{2}\). The corresponding nonzero-mode quasi-bilinear form is denoted by \(\opform{q}_{\txtbsn,\txtnonzero,\sminvtemperature,\smchemicalpotential}\). The zero-mode singular form and the regular nonzero-mode form are kept separate, because the condensate ideals below are generated exactly by resolvents whose test functions have positive value under the singular form. The bosonic Fock space over \(\sphilb{H}_{\txtparticle}\) is \[\spfock_{\txtbsn}
=
\fun{\spfock_{\txtbsn}}{\sphilb{H}_{\txtparticle}}
=
\bigoplus_{n=0}^{\infty}
\bigotimes_{\txtsym}^{n}\sphilb{H}_{\txtparticle}.\] For \(f\in\sphilb{H}_{\txtparticle}\), let \(\opfockcr_{\txtfock}(f)\) and \(\opfockan_{\txtfock}(f)\) be the creation and annihilation operators, and set \[\opfocksegal_{\txtfock}(f)
=
\frac{1}{\sqrt{2}}\rbk{\opfockcr_{\txtfock}(f)+\opfockan_{\txtfock}(f)},
\quad
\opfockweyl_{\txtfock}(f)
=
\napiernum^{\imunit\opfocksegal_{\txtfock}(f)}.\] For a non-expansive one-particle operator \(T\), the second quantization of the second kind is \(\fun{\opfocksndqnt_{\txtbsn}}{T}\). The second quantization of the first kind is denoted by \(\fun{\opfocksndqntdiff_{\txtbsn}}
{\physham[h]_{\txtparticle}}\), and the free Bose Hamiltonian is \begin{equation}\label{expedition0019395}
\begin{aligned}
\physham_{\txtbsn,\txtfr}
=
\fun{\opfocksndqntdiff_{\txtbsn}}{\physham[h]_{\txtparticle}},
\quad
\physham_{\txtbsn,\txtfr,0}
=
\fun{\opfocksndqntdiff_{\txtbsn}}{\physham[h]_{\txtparticle,0}}.
\end{aligned}
\end{equation} With this chemical-potential shift, \[\physham_{\txtbsn,\txtfr,\smchemicalpotential}
=
\fun{\opfocksndqntdiff_{\txtbsn}}
{\physham[h]_{\txtparticle} - \smchemicalpotential},
\quad
\physham_{\txtbsn,\txtfr,0,\smchemicalpotential}
=
\fun{\opfocksndqntdiff_{\txtbsn}}
{\physham[h]_{\txtparticle,0} - \smchemicalpotential}.\]

\subsection{Resolvent Algebra}\label{expedition0012083}

Following Buchholz--Grundling \cite{BuchholzGrundling2}, let \((X,\sigma)\) be a real symplectic space. Let \(\oaresolventalgebra_{0}\) be the universal unital \(\ast\)-algebra generated by \(\set{\oaresolvent(z,f)}
{z \in \fldcmp \setminus \imunit \fldreal, f\in X}\) subject to the resolvent relations \begin{align}
\oaresolvent(z,0)
&=
-\frac{\imunit}{z}, \\
\faadj{\oaresolvent(z,f)}
&=
\oaresolvent(-\cmpconj{z},f), \\
\nu\oaresolvent(\nu z,\nu f)
&=
\oaresolvent(z,f),
\quad
\nu\in\fldreal\setminus\setone{0}, \\
\oaresolvent(z,f)-\oaresolvent(w,f)
&=
\imunit(w-z)
\oaresolvent(z,f)\oaresolvent(w,f), \\
\commutator{\oaresolvent(z,f)}{\oaresolvent(w,g)}
&=
\imunit
\sigma(f,g)
\oaresolvent(z,f)
\oaresolvent(w,g)^2
\oaresolvent(z,f), \label{expedition0012052} \\
\oaresolvent(z,f)\oaresolvent(w,g)
&=
\oaresolvent(z+w,f+g)
\rbk{\oaresolvent(z,f)
+\oaresolvent(w,g)
+\imunit\sigma(f,g)\oaresolvent(z,f)^2\oaresolvent(w,g)},
\end{align} where the last relation is read for \(z+w\ne0\). The \(\oacstar\)-completion with the Buchholz--Grundling norm is the resolvent algebra and is denoted by \(\oaresolventalgebra(X,\sigma)\). For the generators one has \[\norm{\oaresolvent(z,f)}=\frac{1}{\abs{z}}.\]

By \cite{BuchholzVanNuland1}, the linear span of the resolvent generators is norm dense in the resolvent algebra. The dense linear subspace of generators used below is therefore \[\operatorname{span}_{\fldcmp}
\set{\oaresolvent(z,f)}
{z\in\fldcmp\setminus\fldreal, f\in X}.\] The real Buchholz--Grundling parameter and the complex resolvent parameter are related by \[\oaresolvent(z,f)=\opfnresolvent{\imunit z - \Phi(f)},
\quad
\lambda\in\fldreal\setminus\setone{0},
\quad
\oaresolvent(z,f)=\invrbk{z-\Phi(f)},
\quad
z\in\fldcmp\setminus\fldreal.\] Thus \(\oaresolvent(\lambda,f)\) is the special case with complex parameter \(z=\imunit\lambda\). If the represented field is shifted as \(\Phi(f)\mapsto \Phi(f)+c\) with \(c\in\fldreal\), then the complex-parameter notation gives \(\oaresolvent(z,f)\mapsto \oaresolvent(z-c,f)\). In the Buchholz--Grundling real-parameter notation this is the same statement after the embedding \(\lambda\mapsto\imunit\lambda\). When the state formulas below write an imaginary parameter shift, the functional appearing there is the corresponding resolvent-parameter shift functional; the sign convention is fixed by the scalar limit computed in Proposition \ref{expedition0012652}. When the first variable becomes visually long, the same generator may be written as \(\oaresolvent(\lambda;f)\).

For a representation \(\oarepn\) of \(\oaresolventalgebra(X,\sigma)\) and a subspace \(S\subset X\), the representation is regular on \(S\) if \[\Ker\oarepn\rbk{\oaresolvent(1,f)}=\setone{0}
\quad
\rbk{f\in S}.\] A state is regular on \(S\) when its GNS representation is regular on \(S\).

\begin{prop}[Finite-dimensional subalgebras and faithful regular representations]\label{expedition0011838}
For a symplectic space $(X,\sigma)$ and a non-degenerate finite-dimensional subspace $S\subset X$, the finite-dimensional resolvent algebra $\oaresolventalgebra(S,\sigma)$ embeds isometrically into $\oaresolventalgebra(X,\sigma)$.
The full resolvent algebra is the inductive limit of these finite-dimensional subalgebras, and every regular representation of $\oaresolventalgebra(X,\sigma)$ is faithful.
In particular, the center of the full resolvent algebra is trivial.
\end{prop}

\subsection{Settings for the Bose Field at Finite Temperature}\label{expedition0011278}

For inverse temperature \(\sminvtemperature>0\), the zero-mode covariance form of the homogeneous Bose particles is \begin{equation}\label{expedition0012440}
\begin{aligned}
\opform{q}_{\txtbsn,0,\sminvtemperature}(f)
=
2 (2 \pi)^d \smnumberdensity_{\txtbsn,0}(\sminvtemperature)
\abs{\faftr{f}(0)}^{2},
\quad
\opformdomain(\opform{q}_{\txtbsn,0,\sminvtemperature})
=
\fun{\lp^{1}}{\fldreal^{d}}
\cap
\fun{\lp^{2}}{\fldreal^{d}}.
\end{aligned}
\end{equation} The number \(\smnumberdensity_{\txtbsn,0}(\sminvtemperature)\) is the zero mode density, and defined by \begin{equation}\label{expedition0020735}
\begin{aligned}
\smnumberdensity_{\txtbsn,0}(\sminvtemperature)
=
\bar{\smnumberdensity}_{\txtbsn}
+\smnumberdensity_{\txtbsn,\txtcritical}(\sminvtemperature),
\end{aligned}
\end{equation} where \(\bar{\smnumberdensity}_{\txtbsn}\) is a fixed total density and \(\smnumberdensity_{\txtbsn,\txtcritical}(\sminvtemperature)\) is the free critical density defined by \begin{equation}\label{expedition0019397}
\begin{aligned}
\smnumberdensity_{\txtbsn,\txtcritical}(\sminvtemperature)
=
\frac{1}{(2\pi)^d}
\int_{\fldreal^d}
\frac{1}
{\napiernum^{\sminvtemperature\physham[h]_{\txtparticle,0}(k)}-1}
\opdmsr{k}
\end{aligned}
\end{equation} for the free laplacian \(\physham[h]_{\txtparticle,0}\) \eqref{expedition0019396}. The convention is the real Segal-field covariance convention: \(\opform{q}(f)
=
\opform{q}(f,f)\) is the covariance of the self-adjoint Segal field associated with the realification of the complex one-particle vector \(f\). Hermitian normal-ordered two-point functions differ by the usual polarization and creation-annihilation normalization factors. The covariance form \eqref{expedition0012440} gives the zero-mode condition used in the BEC ideals. The notation \(\opform{q}_{\txtbsn,\txtbec,\sminvtemperature}
=
\opform{q}_{\txtbsn,0,\sminvtemperature}
+\opform{q}_{\txtbsn,\txtnonzero,\sminvtemperature}\) is used when the regular and singular components are treated together. The corresponding zero-mode domain and physical one-particle test space are \[\sphilb{D}_{\txtbsn,0,\sminvtemperature}
=
\opformdomain(\opform{q}_{\txtbsn,0,\sminvtemperature})
\cap
\sphilb{H}_{\txtparticle,\sminvtemperature},
\quad
\sphilb{D}_{\txtbsn,\txtphys,\sminvtemperature}
=
\sphilb{D}_{\txtbsn,0,\sminvtemperature}.\]

\subsection{Bounded System Setting}\label{expedition0012145}

The bounded system is the finite-volume regularization of the full space \(\fldreal^{d}\). For \(L>0\), put \[I_L
=
\closedinterval{-\frac{L}{2}}{\frac{L}{2}},
\quad
I_L^d
=
\closedinterval{-\frac{L}{2}}{\frac{L}{2}}^{d},
\quad
V_L
=
\absvol{I_L^d}
=
L^d.\] Periodic boundary conditions are used for the Bose one-particle Hamiltonian. The momentum lattice is \[\setlattice_L
=
\frac{2\pi}{L}\ringratint,
\quad
\setlattice_L^d
=
\rbk{\frac{2\pi}{L}\ringratint}^{d}.\] The finite-volume one-particle and Fock spaces are \[\sphilb{H}_{\txtparticle,L}
=
\fun{\lp^{2}}{I_L^{d}},
\quad
\spfock_{\txtbsn,L}
=
\fun{\spfock_{\txtbsn}}{\sphilb{H}_{\txtparticle,L}}.\]

Let \(P_L:\sphilb{H}_{\txtparticle}\to\sphilb{H}_{\txtparticle,L}\) be the finite-volume projection. Whenever a test function \(f\) or a source is first defined on \(\fldreal^d\), the finite-volume object is obtained by replacing it with \(P_L f\) or the corresponding projected source. The finite-volume boson resolvent algebra is \begin{equation}\label{expedition0012761}
\oaresolventalgebra_L
=
\oaresolventalgebra(P_L\sphilb{H}_{\txtparticle,\txtreal},\sigma).
\end{equation} The finite-volume one-particle Hamiltonian is \begin{equation}\label{expedition0019756}
\physham[h]_{\txtparticle,0,L}
=
P_L \physham[h]_{\txtparticle,0} P_L .
\end{equation} The free Bose Hamiltonian and number operator are \begin{equation}\label{expedition0019757}
\physham_{\txtbsn,\txtfr,0,L}
=
\fun{\opfocksndqntdiff_{\txtbsn}}{\physham[h]_{\txtparticle,0,L}},
\quad
\opfocknumber_{\txtbsn,L}
=
\fun{\opfocksndqntdiff_{\txtbsn}}{P_L}.
\end{equation}

For a fixed total density \(\bar{\smnumberdensity}_{\txtbsn}>0\), define \(y_L>1\) by \[\frac{1}{L^d}
\sum_{k\in\setlattice_L^d}
\frac{1}{y_L\napiernum^{\sminvtemperature \physham[h]_{\txtparticle,0,L}(k)}-1}
=
\bar{\smnumberdensity}_{\txtbsn}.\]

\subsection{Quasilocal Structure}\label{quasilocal-structure}

The periodic finite-volume algebras \(\oaresolventalgebra_L\) in \eqref{expedition0012761} are approximating algebras. They are not treated as an increasing family by themselves, because the periodic spaces \(\lp^2(I_L^d)\) do not carry canonical inclusions for different values of \(L\). The thermodynamic local algebra is instead built from bounded regions. For a bounded open region \(\mathcal{O}\subset\fldreal^d\), set \[
\oaresolventalgebra(\mathcal{O})
=
\fun{\oaresolventalgebra}{\lp^2(\mathcal{O};\fldreal),\sigma}.
\] If \(\mathcal{O}_1\subset\mathcal{O}_2\), the embedding \(\oaresolventalgebra(\mathcal{O}_1)\to\oaresolventalgebra(\mathcal{O}_2)\) is defined by extension of test functions by zero. Thus \[
\oaresolventalgebra_{\txtloc}
=
\bigcup_{\mathcal{O}\Subset\fldreal^d}\oaresolventalgebra(\mathcal{O}),
\quad
\oaresolventalgebra
=
\gtclos{\oaresolventalgebra_{\txtloc}}
=
\oaresolventalgebra(\sphilb{D}_{\txtbsn,\sminvtemperature,\smchemicalpotential},\sigma).
\] For the finite-volume Gibbs approximants, a local test function is first embedded in the full one-particle space and then projected to the periodic box by \(P_L\). This is the relation between the local thermodynamic algebra and the finite-volume regularization. For \(z\in\fldcmp\setminus\fldreal\) and a real test function \(f\) in a local one-particle space, the corresponding complex resolvent notation is \(\oaresolvent(z,f)\).

\subsection{Main Theorem}\label{main-theorem}

The main theorem collects the BEC information after the Kac-selected density and the Euler equations have fixed the zero-mode excess.

\begin{thm}[Mean-field BEC, condensate ideal, and equivalent probabilistic descriptions]\label{expedition0019672}
For the finite-volume mean-field Bose gas \eqref{expedition0012600} with
$\lambda>0$ and $\smchemicalpotential\in\fldreal$, Theorem
\ref{expedition0012647} selects the limiting density
$\bar{\smnumberdensity}_{\txtbsn}$.
Assume the condensed alternative \eqref{expedition0012606}, equivalently
$\smnumberdensity_{\txtbsn,0}(\sminvtemperature)>0$.
Then the following assertions hold.
\begin{enumerate}
\item
The selected chemical potential satisfies
$\smchemicalpotential_{\mathrm{sel}}
=\lambda\bar{\smnumberdensity}_{\txtbsn}$.
Hence the selected one-particle Hamiltonian
$\physham[h]_{\bar{\smnumberdensity}_{\txtbsn}}$ reduces exactly to the free
one-particle Hamiltonian in \eqref{expedition0012694}, and Proposition
\ref{expedition0012653} gives the KMS property for the corresponding
component dynamics.
\item
The zero-mode point-mass covariance \eqref{expedition0012440} gives the
mean-field BEC ideal \eqref{expedition0019398}.
Proposition \ref{expedition0012642} identifies the regular quotient obtained
from the nonregular BEC representation, while Proposition
\ref{expedition0012651} identifies the represented center of the
direct-integral state \eqref{expedition0012609}.
\item
The occupation-number probability formulation gives the same density collapse
in Proposition \ref{expedition0019755}, the same zero-mode excess in
Proposition \ref{expedition0019803}, and the corresponding direct-integral
component law, zero-mode ODLRO covariance, local condensate tests, and
density-fluctuation remainder.
\item
The Brownian-loop formulation gives the same density selection as total-winding
collapse in Proposition \ref{expedition0019610}, identifies the zero-mode
excess with macroscopic winding in Proposition \ref{expedition0019611}, gives
the loop form of zero-mode ODLRO in Proposition \ref{expedition0019612}, and
recovers the standard one-particle density-matrix BEC criterion in Proposition
\ref{expedition0019671}.
\end{enumerate}
At fixed finite condensate density, the ideal \eqref{expedition0019398} and the
zero-mode ODLRO statement do not imply Buchholz's primary-state
number-resolvent criterion \eqref{expedition0019101}; see Proposition
\ref{expedition0012895}.
The strict proper-condensate conclusion is obtained only from families
satisfying
$\smnumberdensity_{\txtbsn,0}^{(\sigma)}(\sminvtemperature)
\absvol{\mathcal{O}}\to\infty$, as in Theorem \ref{expedition0012897}.
\end{thm}

\section{Resolvent Algebraic Approach}\label{expedition0012743}

The imperfect Bose gas of \cite[Section 4.4]{AndreVerbeure001} gives a concrete mean-field model in which the Kac density law and the Euler equations determine the zero-mode excess density. The resolvent-algebra construction, in the notation of \cite{YoshitsuguSekine004}, then records the corresponding zero-mode point-mass covariance by the ideal \eqref{expedition0019398}. A nonregular representation gives the regular nonzero-mode quotient, whereas direct integral of the component states gives the represented central variable. The GNS role of the ideal is fixed below by the nonregular representation and the direct-integral representation.

\subsection{Finite-Volume Mean-Field Model}\label{finite-volume-mean-field-model}

The finite-volume mean-field Hamiltonian uses the one-particle and free Fock Hamiltonians defined in \eqref{expedition0019756} and \eqref{expedition0019757}. For \(\lambda>0\) and for a chemical-potential value for which the finite-volume Gibbs trace is defined, define \begin{equation}\label{expedition0012600}
\physham_{\txtbsn,\txtmeanfield,\smchemicalpotential,L}
=
\physham_{\txtbsn,\txtfr,0,L}
-\smchemicalpotential \opfocknumber_{\txtbsn,L}
+\frac{\lambda}{2V_L}\opfocknumber_{\txtbsn,L}^{2},
\end{equation} where \(\physham_{\txtbsn,\txtfr,0,L}\) is the finite-volume free Hamiltonian defined in \eqref{expedition0019757}. This is the finite-volume mean-field Bose gas Hamiltonian. Since the operator sum \(\frac{\lambda}{2V_L}\opfocknumber_{\txtbsn,L}^{2}
-\smchemicalpotential \opfocknumber_{\txtbsn,L}\) is bounded below for any \(\smchemicalpotential \in \fldreal\), the Hamiltonian \eqref{expedition0012600} is bounded below on \(\spfock_{\txtbsn,L}\). Its Friedrichs extension is denoted by the same symbol, and this operator is treated as self-adjoint. Furthermore the free Hamiltonian and the number operator strongly commute and are simultaneously diagonalizable. Clearly its heat operator \(\napiernum^{-\sminvtemperature \physham_{\txtbsn,\txtmeanfield,\smchemicalpotential,L}}\) is a trace class operator for any \(\sminvtemperature > 0\). This formula is the Hamiltonian \cite[Eq. (4.51)]{AndreVerbeure001} written in the notation of the bounded Bose field used here. The same parameter \(\smchemicalpotential\) is used in the finite-volume Gibbs functional \eqref{expedition0012601}, in the density law \eqref{expedition0012728}, and in the Euler equations \eqref{expedition0012604}--\eqref{expedition0012607}. The Euler--Lagrange multiplier after density selection is denoted separately by \(\smchemicalpotential_{\mathrm{sel}}\).

\begin{prop}[Finite-volume super-stability of the mean-field Hamiltonian]\label{expedition0012664}
Assume $\lambda>0$ and $\smchemicalpotential\in\fldreal$.
The finite-volume Hamiltonian \eqref{expedition0012600} satisfies the following estimate:
\begin{equation}\label{expedition0012733}
\physham_{\txtbsn,\txtmeanfield,\smchemicalpotential,L}
\geq
-\frac{V_L\smchemicalpotential^{2}}{2\lambda}.
\end{equation}
In particular the energy is bounded from below in finite volume by the
quadratic density term, which is the super-stability mechanism of Verbeure's
imperfect Bose gas \cite[Eq. (4.51) and Section 4.4]{AndreVerbeure001}.
\end{prop}

\begin{proof}
The free Hamiltonian $\physham_{\txtbsn,\txtfr,0,L}$ is clearly nonnegative.
Completing the square gives
$$\frac{\lambda}{2V_L}\opfocknumber_{\txtbsn,L}^{2}
-\smchemicalpotential\opfocknumber_{\txtbsn,L}
=
\frac{\lambda}{2V_L}
\rbk{\opfocknumber_{\txtbsn,L}
-\frac{V_L\smchemicalpotential}{\lambda}}^{2}
-\frac{V_L\smchemicalpotential^{2}}{2\lambda}.$$
Combining this identity with \eqref{expedition0012600} gives
\eqref{expedition0012733}.
\end{proof}

\begin{rem}[Chemical-potential convention]\label{expedition0012900}
The symbol $\smchemicalpotential\in\fldreal$ denotes the finite-volume
chemical-potential parameter in \eqref{expedition0012600}.
After density selection, the Euler--Lagrange multiplier is
$\smchemicalpotential_{\mathrm{sel}}$, and the effective nonzero-mode chemical
potential is
$$
\smchemicalpotential_{\mathrm{eff}}
=
\smchemicalpotential_{\mathrm{sel}}
-\lambda\bar{\smnumberdensity}_{\txtbsn}.
$$
In the condensed case \eqref{expedition0012606},
$\smchemicalpotential_{\mathrm{eff}}=0$, equivalently
$\smchemicalpotential_{\mathrm{sel}}
=\lambda\bar{\smnumberdensity}_{\txtbsn}$.
No separate reference chemical potential is introduced in the finite-volume
mean-field construction.
\end{rem}

The finite-volume mean-field Gibbs functional is represented on \(\oaresolventalgebra_L\) by \begin{equation}\label{expedition0012601}
\fun{\oastate[\psi_{\txtmeanfield,\sminvtemperature,\smchemicalpotential,L}]}{A}
=
\frac{\sqfun{\trace_{\spfock_{\txtbsn,L}}}
{\napiernum^{-\sminvtemperature \physham_{\txtbsn,\txtmeanfield,\smchemicalpotential,L}} A}}
{\sqfun{\trace_{\spfock_{\txtbsn,L}}}
{\napiernum^{-\sminvtemperature \physham_{\txtbsn,\txtmeanfield,\smchemicalpotential,L}}}},
\quad
A\in\oaresolventalgebra_L.
\end{equation} The finite-volume mean-field Heisenberg action is represented on \(\oaresolventalgebra_L\) by \begin{equation}\label{expedition0012602}
\fun{\alpha_{\txtmeanfield,\smchemicalpotential,L,t}}{A}
=
\napiernum^{\imunit t\physham_{\txtbsn,\txtmeanfield,\smchemicalpotential,L}}
A
\napiernum^{-\imunit t\physham_{\txtbsn,\txtmeanfield,\smchemicalpotential,L}},
\quad
A\in\oaresolventalgebra_L.
\end{equation}

\begin{prop}[Locally normal finite-volume action on creation and annihilation operators]
Work in the regular Fock representation of
$\oaresolventalgebra_L$ on
$\spfock_{\txtbsn,L}$.
For $f
\in P_L \sphilb{H}_{\txtparticle}$, the locally normal finite-volume dynamics induced by
\eqref{expedition0012602} acts on the unbounded field algebra generated by
$\opfockcr_{\txtfock}(f)$ and $\opfockan_{\txtfock}(f)$ as follows:
\begin{equation}\label{expedition0012727}
\begin{aligned}
\napiernum^{\imunit t\physham_{\txtbsn,\txtmeanfield,\smchemicalpotential,L}}
\opfockcr_{\txtfock}(f)
\napiernum^{-\imunit t\physham_{\txtbsn,\txtmeanfield,\smchemicalpotential,L}}
&=
\fun{\opfockcr_{\txtfock}}
{\napiernum^{\imunit t(\physham[h]_{\txtparticle,0,L}-\smchemicalpotential)} f}
\napiernum^{\imunit \frac{\lambda t}{2V_L}(2\opfocknumber_{\txtbsn,L}+1)},
\\ 
\napiernum^{\imunit t\physham_{\txtbsn,\txtmeanfield,\smchemicalpotential,L}}
\opfockan_{\txtfock}(f)
\napiernum^{-\imunit t\physham_{\txtbsn,\txtmeanfield,\smchemicalpotential,L}}
&=
\fun{\opfockan_{\txtfock}}
{\napiernum^{\imunit t(\physham[h]_{\txtparticle,0,L}-\smchemicalpotential)} f}
\napiernum^{-\imunit \frac{\lambda t}{2V_L}(2\opfocknumber_{\txtbsn,L}-1)}.
\end{aligned}
\end{equation}
\end{prop}

\begin{proof}
The operators
$\physham_{\txtbsn,\txtfr,0,L}
-\smchemicalpotential\opfocknumber_{\txtbsn,L}$
and
$\opfocknumber_{\txtbsn,L}$ strongly commute.
The free part gives the standard second-quantized relations
$$\napiernum^{\imunit t(\physham_{\txtbsn,\txtfr,0,L}
-\smchemicalpotential\opfocknumber_{\txtbsn,L})}
\opfockcr_{\txtfock}(f)
\napiernum^{-\imunit t(\physham_{\txtbsn,\txtfr,0,L}
-\smchemicalpotential\opfocknumber_{\txtbsn,L})}
=
\fun{\opfockcr_{\txtfock}}
{\napiernum^{\imunit t(\physham[h]_{\txtparticle,0,L}-\smchemicalpotential)} f},$$
and the analogous formula for $\opfockan_{\txtfock}(f)$.
For the mean-field part, we obtain
$$\begin{aligned}
\napiernum^{\imunit s\opfocknumber_{\txtbsn,L}^2}
\opfockcr_{\txtfock}(f)
\napiernum^{-\imunit s\opfocknumber_{\txtbsn,L}^2}
&=
\opfockcr_{\txtfock}(f)\napiernum^{\imunit s(2\opfocknumber_{\txtbsn,L}+1)},
\\ 
\napiernum^{\imunit s\opfocknumber_{\txtbsn,L}^2}
\opfockan_{\txtfock}(f)
\napiernum^{-\imunit s\opfocknumber_{\txtbsn,L}^2}
&=
\opfockan_{\txtfock}(f)\napiernum^{-\imunit s(2\opfocknumber_{\txtbsn,L}-1)}
\end{aligned}$$
on the finite-particle space.
Substituting $s=\frac{\lambda t}{2V_L}$ gives \eqref{expedition0012727}.
\end{proof}

\begin{prop}[Finite-volume KMS identity for the approximating functional]\label{expedition0012663}
For fixed $A,B\in\oaresolventalgebra_{\txtloc}$ and all sufficiently large $L$, the functional
$\oastate[\psi_{\txtmeanfield,\sminvtemperature,\smchemicalpotential,L}]$
on $\oaresolventalgebra_L$ satisfies the $\sminvtemperature$-KMS identity with respect to
$\alpha_{\txtmeanfield,\smchemicalpotential,L,t}$.
\end{prop}

\begin{proof}
For fixed $A,B\in\oaresolventalgebra_{\txtloc}$, choose $L$ so large that both elements are represented in $\oaresolventalgebra_L$.
Formula \eqref{expedition0012601} then gives a normal finite-volume functional, and the standard Gibbs KMS identity holds for the Heisenberg evolution \eqref{expedition0012602}.
\end{proof}

\begin{prop}[Local weak-$\ast$ cluster states from the finite-volume Gibbs data]\label{expedition0012674}
The set of finite-volume functionals
$\fml{\oastate[\psi_{\txtmeanfield,\sminvtemperature,\smchemicalpotential,L}]}
{L > 0}$
has local weak-$\ast$ cluster points.
The resulting local cluster functional is a state on
$\oaresolventalgebra_{\txtloc}$ and is denoted by
$\oastate[\psi_{\txtmeanfield,\txtloc,\sminvtemperature,\smchemicalpotential}]$
when the chosen subnet is fixed.
\end{prop}

\begin{proof}
This is due to the Hahn--Banach theorem.
\end{proof}

No later argument in this section depends on a canonical thermodynamic Gibbs state on the full abstract algebra; only the local restriction and the selected-density construction after Theorem \ref{expedition0012647} are used.

\begin{rem}[No naive quasilocal limit of the finite-volume dynamics]\label{expedition0018016}
The finite-volume dynamics
$\alpha_{\txtmeanfield,\smchemicalpotential,L,t}$
does not form a compatible inductive system on the local resolvent algebras.
Indeed, \eqref{expedition0012727} shows that the evolution of a fixed local creation or annihilation operator contains the full-volume number operator
$\opfocknumber_{\txtbsn,L}$
through the exponential factor
$\napiernum^{\pm\imunit\lambda t(2\opfocknumber_{\txtbsn,L}\pm1)/(2V_L)}$.
If the same local observable is embedded into a larger box, this factor is changed by the particles in the added volume.
Thus there is no volume-independent local automorphism obtained by the prescription
$A\mapsto\alpha_{\txtmeanfield,\smchemicalpotential,L,t}(A)$
for all sufficiently large $L$.

Proposition \ref{expedition0012663} is therefore only a finite-volume KMS statement, and Proposition \ref{expedition0012674} is only a state-extension and compactness statement.
The mean-field automorphism group for the equilibrium state is obtained only
after the density has been selected.
Once the Kac-measure collapse in Theorem \ref{expedition0012647} has selected
$\bar{\smnumberdensity}_{\txtbsn}$,
the factor
$\opfocknumber_{\txtbsn,L}/V_L$
in local equilibrium correlation functions is replaced by this scalar selected density.
Consequently the one-particle multiplier in the selected-density mean-field
automorphism group is
$$\napiernum^{\imunit t(\physham[h]_{\txtparticle,0}(k)-\smchemicalpotential_{\mathrm{sel}}+\lambda\bar{\smnumberdensity}_{\txtbsn})},$$
which is exactly the multiplier used in the selected-density mean-field
automorphism group
\eqref{expedition0012675}.
The symbol for the mean-field automorphism group is used only after this
identification, namely as
$\alpha_{\txtmeanfield,\bar{\smnumberdensity}_{\txtbsn},t}$
in \eqref{expedition0012675}.
\end{rem}

\subsection{Density Decomposition and Kac Measure}\label{density-decomposition-and-kac-measure}

In \cite[Eq. (4.50) and Section 4.4]{AndreVerbeure001}, the Kac-function is introduced at the thermodynamic-limit level as the kernel which decomposes a grand-canonical equilibrium state into canonical equilibrium states of fixed density. For the mean-field Bose gas of \cite[Eq. (4.51)]{AndreVerbeure001} it is stated to be a delta-function. The finite-volume object introduced below is not a separate definition printed in \cite{AndreVerbeure001}. It is the finite-volume density law whose weak limit realizes the Kac-function statement.

Let the local number density be \(\smnumberdensity_{L}
= \frac{\opfocknumber_{\txtbsn,L}}{V_L}\) and define \(K_{\sminvtemperature,\lambda,\smchemicalpotential,L}\) to be the probability measure on \(\fldreal_{\geq0}\) obtained as the spectral distribution of \(\smnumberdensity_L\) in the finite-volume grand-canonical Gibbs representation of the Hamiltonian \(\physham_{\txtbsn,\txtmeanfield,\smchemicalpotential,L}\). Let \begin{equation}\label{expedition0019646}
\begin{aligned}
\mathcal{O}_L
=
\set{\mathsf{n} = \seq{\mathsf{n}_k}{k\in\setlattice_L^d}}
{\mathsf{n}_k \in \monnat,
\sum_{k\in\setlattice_L^d}
\mathsf{n}_k < \infty}.
\end{aligned}
\end{equation}

\begin{prop}[Finite-volume density law as a spectral distribution]\label{expedition0012679}
For every Borel set $B\subset\fldreal_{\geq0}$,
$$K_{\sminvtemperature,\lambda,\smchemicalpotential,L}(B)
=
\fun{\oastate[\psi_{\txtmeanfield,\sminvtemperature,\smchemicalpotential,L}]}
{\fndef{B}(\smnumberdensity_L)}
=
\frac{\sqfun{\trace_{\spfock_{\txtbsn,L}}}
{\napiernum^{-\sminvtemperature\physham_{\txtbsn,\txtmeanfield,\smchemicalpotential,L}}
\fndef{B}(\smnumberdensity_L)}}
{\sqfun{\trace_{\spfock_{\txtbsn,L}}}
{\napiernum^{-\sminvtemperature\physham_{\txtbsn,\txtmeanfield,\smchemicalpotential,L}}}}.$$
Equivalently,
for the Dirac measure $\diracdelta_x$,
it holds that
\begin{equation}\label{expedition0012728}
\begin{aligned}
K_{\sminvtemperature,\lambda,\smchemicalpotential,L}
&=
\frac{1}{\smpartitionfunc_{\sminvtemperature,\lambda,\smchemicalpotential,L}}
\sum_{\mathsf{n}\in\mathcal{O}_L}
\fnexp{-\sminvtemperature
\rbk{\sum_{k\in\setlattice_L^d}
\rbk{\physham[h]_{\txtparticle,0,L}(k)-\smchemicalpotential}\mathsf{n}_k
+\frac{\lambda}{2V_L}
\rbk{\sum_{k\in\setlattice_L^d}\mathsf{n}_k}^{2}}}
\delta_{\frac{1}{V_L}
\sum_{k\in\setlattice_L^d}\mathsf{n}_k},
\\ 
\smpartitionfunc_{\sminvtemperature,\lambda,\smchemicalpotential,L}
&=
\sum_{\mathsf{n}\in\mathcal{O}_L}
\fnexp{-\sminvtemperature
\rbk{\sum_{k\in\setlattice_L^d}
\rbk{\physham[h]_{\txtparticle,0,L}(k)-\smchemicalpotential}\mathsf{n}_k
+\frac{\lambda}{2V_L}
\rbk{\sum_{k\in\setlattice_L^d}\mathsf{n}_k}^{2}}}.
\end{aligned}
\end{equation}
\end{prop}

\begin{proof}
The finite-volume Fock space $\spfock_{\txtbsn,L}$ has the occupation-number orthonormal basis
$\fml{\ket{\mathsf{n}}}{\mathsf{n}\in\mathcal{O}_L}$.
On this basis we obtain
$$\opfocknumber_{\txtbsn,L}\ket{\mathsf{n}}
=
\rbk{\sum_{k\in\setlattice_L^d}\mathsf{n}_k}\ket{\mathsf{n}},
\quad
\smnumberdensity_L\ket{\mathsf{n}}
=
\frac{1}{V_L}
\rbk{\sum_{k\in\setlattice_L^d}\mathsf{n}_k}\ket{\mathsf{n}}.$$
The finite-volume operator
$\physham_{\txtbsn,\txtfr,0,L}
-\smchemicalpotential\opfocknumber_{\txtbsn,L}$
and the number operator commute and are diagonal in the same basis:
$$\rbk{\physham_{\txtbsn,\txtfr,0,L}
-\smchemicalpotential\opfocknumber_{\txtbsn,L}}\ket{\mathsf{n}}
=
\rbk{\sum_{k\in\setlattice_L^d}
\rbk{\physham[h]_{\txtparticle,0,L}(k)-\smchemicalpotential}\mathsf{n}_k}
\ket{\mathsf{n}}.$$
Therefore it holds that
$$\physham_{\txtbsn,\txtmeanfield,\smchemicalpotential,L}\ket{\mathsf{n}}
=
\rbk{\sum_{k\in\setlattice_L^d}
\rbk{\physham[h]_{\txtparticle,0,L}(k)-\smchemicalpotential}\mathsf{n}_k
+\frac{\lambda}{2V_L}
\rbk{\sum_{k\in\setlattice_L^d}\mathsf{n}_k}^{2}}
\ket{\mathsf{n}}.$$
The spectral projection $\fndef{B}(\smnumberdensity_L)$ is also diagonal and acts by multiplication with
$\fndef{B}\rbk{\frac{1}{V_L}\sum_{k\in\setlattice_L^d}\mathsf{n}_k}$.
For the fixed Borel set $B$, put
$$\mathcal{O}_L(B)
=
\set{\mathsf{n}\in\mathcal{O}_L}
{\frac{1}{V_L}\sum_{k\in\setlattice_L^d}\mathsf{n}_k\in B}.$$
Then the trace cut by this spectral projection is the subseries
$$\begin{aligned}
&\sqfun{\trace_{\spfock_{\txtbsn,L}}}
{\napiernum^{-\sminvtemperature\physham_{\txtbsn,\txtmeanfield,\smchemicalpotential,L}}
\fndef{B}(\smnumberdensity_L)}
\\ 
&=
\sum_{\mathsf{n}\in\mathcal{O}_L(B)}
\fnexp{-\sminvtemperature
\rbk{\sum_{k\in\setlattice_L^d}
\rbk{\physham[h]_{\txtparticle,0,L}(k)-\smchemicalpotential}\mathsf{n}_k
+\frac{\lambda}{2V_L}
\rbk{\sum_{k\in\setlattice_L^d}\mathsf{n}_k}^{2}}}
\\ 
&=
\sum_{\mathsf{n}\in\mathcal{O}_L}
\fnexp{-\sminvtemperature
\rbk{\sum_{k\in\setlattice_L^d}
\rbk{\physham[h]_{\txtparticle,0,L}(k)-\smchemicalpotential}\mathsf{n}_k
+\frac{\lambda}{2V_L}
\rbk{\sum_{k\in\setlattice_L^d}\mathsf{n}_k}^{2}}}
\fndef{B}\rbk{\frac{1}{V_L}
\sum_{k\in\setlattice_L^d}\mathsf{n}_k}.
\end{aligned}$$
Since the full Gibbs trace is finite, this positive subseries is well-defined for every Borel set $B$.
The last expression is exactly the value on $B$ of the weighted sum of point masses in \eqref{expedition0012728}.
Taking the trace in the occupation-number basis gives the formula for
$K_{\sminvtemperature,\lambda,\smchemicalpotential,L}(B)$.
Since the denominator is the same trace with $B=\fldreal_{\geq0}$, the total mass is one.
Hence this trace formula is exactly the spectral distribution of $\smnumberdensity_L$ in the finite-volume grand-canonical Gibbs state.
\end{proof}

The corresponding free density law obtained by deleting the quadratic mean-field weight from \eqref{expedition0012728} is explicitly \begin{equation}\label{expedition0018007}
\begin{aligned}
K_{\sminvtemperature,0,\smchemicalpotential,L}
&=
\frac{1}{\smpartitionfunc_{\sminvtemperature,0,\smchemicalpotential,L}}
\sum_{\mathsf{n}\in\mathcal{O}_L}
\fnexp{-\sminvtemperature
\sum_{k\in\setlattice_L^d}
\rbk{\physham[h]_{\txtparticle,0,L}(k)-\smchemicalpotential}\mathsf{n}_k}
\delta_{\frac{1}{V_L}
\sum_{k\in\setlattice_L^d}
\mathsf{n}_k},
\\ 
\smpartitionfunc_{\sminvtemperature,0,\smchemicalpotential,L}
&=
\sum_{\mathsf{n}\in\mathcal{O}_L}
\fnexp{-\sminvtemperature
\sum_{k\in\setlattice_L^d}
\rbk{\physham[h]_{\txtparticle,0,L}(k)-\smchemicalpotential}\mathsf{n}_k}.
\end{aligned}
\end{equation}

\begin{prop}[Exact finite-volume density weighting]\label{expedition0012680}
The mean-field density law is obtained from the free density law by
\begin{equation}\label{expedition0012603}
\opdmsr{K_{\sminvtemperature,\lambda,\smchemicalpotential,L}(r)}
=
\frac{\fnexp{-\frac{\sminvtemperature \lambda V_L r^2}{2}}}
{\int_{\fldreal_{\geq0}}
\fnexp{-\frac{\sminvtemperature \lambda V_L s^2}{2}}
\opdmsr{K_{\sminvtemperature,0,\smchemicalpotential,L}(s)}}
\opdmsr{K_{\sminvtemperature,0,\smchemicalpotential,L}(r)}.
\end{equation}
\end{prop}

\begin{proof}
The operators
$\physham_{\txtbsn,\txtfr,0,L}
-\smchemicalpotential\opfocknumber_{\txtbsn,L}$
and $\opfocknumber_{\txtbsn,L}$ commute,
hence they are simultaneously diagonal in the occupation number basis.
On the spectral subspace where $\frac{\opfocknumber_{\txtbsn,L}}{V_L}
= r$, the mean-field perturbation contributes the scalar factor
$\fnexp{-\frac{\sminvtemperature \lambda V_L r^2}{2}}$.
Taking the trace over all occupation configurations with the same value of $r$ gives \eqref{expedition0012603}.
\end{proof}

\begin{prop}[Large deviation principle for the free finite-volume density]\label{expedition0012645}
The free density laws $K_{\sminvtemperature,0,\smchemicalpotential,L}$ satisfy a large-deviation principle on $\fldreal_{\geq0}$
with speed $V_L$ and good rate function given by
$$\begin{aligned}
I_{\txtfr,\sminvtemperature,\smchemicalpotential}(r)
&=
\sup_{q<-\sminvtemperature\smchemicalpotential}
\rbk{q r
-\Lambda_{\txtfr,\sminvtemperature,\smchemicalpotential}(q)},
\\ 
\Lambda_{\txtfr,\sminvtemperature,\smchemicalpotential}(q)
&=
\lim_{L\to\infty}
\frac{1}{V_L}
\log
\int_{\fldreal_{\geq0}}
\napiernum^{qV_L r}
\opdmsr{K_{\sminvtemperature,0,\smchemicalpotential,L}(r)}
=
\lim_{L\to\infty}
\frac{1}{V_L}
\log
\frac{\smpartitionfunc_{\sminvtemperature,0,\smchemicalpotential+\frac{q}{\sminvtemperature},L}}
{\smpartitionfunc_{\sminvtemperature,0,\smchemicalpotential,L}},
\\ 
&=
\frac{1}{(2\pi)^d}
\int_{\fldreal^d}
\fun{\log}
{\frac{1-\napiernum^{-\sminvtemperature(\physham[h]_{\txtparticle,0}(k)-\smchemicalpotential)}}
{1-\napiernum^{-\sminvtemperature(\physham[h]_{\txtparticle,0}(k)-\smchemicalpotential)+q}}}
\opdmsr{k},
\quad
q < - \sminvtemperature \smchemicalpotential.
\end{aligned}$$
More explicitly, for every closed set $F\subset\fldreal_{\geq0}$,
the upper inequality holds:
$$\limsup_{L\to\infty}
\frac{1}{V_L}
\log
\fun{K_{\sminvtemperature,0,\smchemicalpotential,L}}{F}
\leq
-\inf_{r\in F}
I_{\txtfr,\sminvtemperature,\smchemicalpotential}(r).$$
On the other hand,
for every open set $G\subset\fldreal_{\geq0}$,
the lower inequality follows:
$$\liminf_{L\to\infty}
\frac{1}{V_L}
\log
\fun{K_{\sminvtemperature,0,\smchemicalpotential,L}}{G}
\geq
-\inf_{r\in G}
I_{\txtfr,\sminvtemperature,\smchemicalpotential}(r).$$
\end{prop}

The function \(\Lambda_{\txtfr,\sminvtemperature,\smchemicalpotential}\) in Proposition \ref{expedition0012645} is the difference of free Bose pressures. This is proved in the pressure calculation \eqref{expedition0018014} in Section \ref{expedition0018010}.

\begin{proof}
For the probabilistic discussion, see Section \ref{expedition0018010}.
\end{proof}

\begin{thm}[Density-law realization of the mean-field Kac-function collapse]\label{expedition0012647}
Let $I_{\txtfr,\sminvtemperature,\smchemicalpotential}$ be the good rate function
from Proposition \ref{expedition0012645}.
Then the function
$$I_{\txtmeanfield,\sminvtemperature,\smchemicalpotential,\lambda}(r)
=
I_{\txtfr,\sminvtemperature,\smchemicalpotential}(r)
+\frac{\sminvtemperature\lambda r^2}{2}
-\inf_{s\geq0}
\rbk{I_{\txtfr,\sminvtemperature,\smchemicalpotential}(s)
+\frac{\sminvtemperature\lambda s^2}{2}}$$
has a unique minimizer,
and denote it by $\bar{\smnumberdensity}_{\txtbsn}$.
Furthermore $K_{\sminvtemperature,\lambda,\smchemicalpotential,L}$
weakly converges to the delta measure $\delta_{\bar{\smnumberdensity}_{\txtbsn}}$.
\end{thm}

In the mean-field Bose gas of \cite[Section 4.4]{AndreVerbeure001}, this is the statement that the Kac function is a delta function and that the canonical and grand-canonical equilibrium states are equivalent in the thermodynamic limit.

\begin{proof}
See Section \ref{expedition0018011}.
\end{proof}

From this point on, \(\bar{\smnumberdensity}_{\txtbsn}\) denotes the density selected by the Kac collapse at the parameter \(\smchemicalpotential\). The Euler--Lagrange multiplier in the density-reduced variational problem is denoted by \(\smchemicalpotential_{\mathrm{sel}}\). No statement below identifies \(\smchemicalpotential\) with \(\smchemicalpotential_{\mathrm{sel}}\). The equality \(\smchemicalpotential_{\mathrm{sel}}=\lambda\bar{\smnumberdensity}_{\txtbsn}\) is only the complementary condition in the condensed case \eqref{expedition0012606}. In other words, the Kac collapse selects only the scalar total density. It does not by itself imply BEC, determine a zero-mode Dirac point mass in the covariance, or define a resolvent ideal. The zero-mode covariance form \eqref{expedition0012440} becomes nonzero only after the condensed case \eqref{expedition0012606} has been selected and the nonzero-mode density has been compared with its saturated value.

\subsection{Equilibrium Equations}\label{equilibrium-equations}

Two chemical-potential symbols are used in the mean-field section. The parameter \(\smchemicalpotential\) is the finite-volume parameter in the density law. The parameter \(\smchemicalpotential_{\mathrm{sel}}\) is the Euler--Lagrange multiplier after the density has been selected. With this convention the selected nonzero-mode energy is \(\physham[h]_{\txtparticle,0}(k)-\smchemicalpotential_{\mathrm{sel}}+\lambda\bar{\smnumberdensity}_{\txtbsn}\), and in the condensed case \eqref{expedition0012606} the complementarity equation gives \(\smchemicalpotential_{\mathrm{sel}}=\lambda\bar{\smnumberdensity}_{\txtbsn}\).

The term Euler equations means the first-order necessary conditions for the density-reduced mean-field free-energy variational problem after the Kac-measure collapse in Theorem \ref{expedition0012647}. The admissible variables are the zero-mode density \(\rho_0\geq0\) and the nonzero-mode occupation density \(n(k)\geq0\). For nonzero modes, we vary \(n\) by compactly supported real test functions \(\eta(k)\) in regions where \(n(k)>0\), with \(\abs{\epsilon}\) small enough that \(n+\epsilon\eta\geq0\). At such an interior minimizer these variations give vanishing first variation of the mean-field free-energy density \(\smfreeenergy_{\txtmeanfield}\) introduced in Proposition \ref{expedition0012649}. In the notation of that proposition this is the stationarity condition \(\delta\smfreeenergy_{\txtmeanfield}/\delta n(k)=0\). For the zero mode, the admissible set is the half-line \(\closedinterval{0}{\infty}\), so the Euler condition is the one-sided variational inequality \[\frac{\partial\smfreeenergy_{\txtmeanfield}}{\partial\rho_0}\geq0,
\quad
\rho_0\frac{\partial\smfreeenergy_{\txtmeanfield}}{\partial\rho_0}=0.\] The mean-field interaction does not change the Bose entropy term; it shifts the one-particle energy by the selected density \(\bar{\smnumberdensity}_{\txtbsn}\). Thus the nonzero modes are governed by the effective energy \(\physham[h]_{\txtparticle,0}(k)-\smchemicalpotential_{\mathrm{sel}}+\lambda\bar{\smnumberdensity}_{\txtbsn}\), while the zero mode satisfies a complementary condition for the condensate density.

\begin{prop}[Mean-field free-energy density and complementarity condition]\label{expedition0012649}
Let $\rho_0\geq0$ be the zero-mode density, let $n(k)\geq0$ be the nonzero-mode occupation density,
and set
$$\rho
=
\rho_0
+\frac{1}{(2\pi)^d}\int_{\fldreal^d}n(k)\opdmsr{k}.$$
Up to constants independent of $\rho_0$ and $n$, the mean-field free-energy density is
$$\begin{aligned}
\smfreeenergy_{\txtmeanfield}(\rho_0,n)
&=
\frac{1}{(2\pi)^d}
\int_{\fldreal^d}
\rbk{\physham[h]_{\txtparticle,0}(k)-\smchemicalpotential_{\mathrm{sel}}}n(k)\opdmsr{k}
-\smchemicalpotential_{\mathrm{sel}}\rho_0
+\frac{\lambda}{2}
\rbk{\rho_0
+\frac{1}{(2\pi)^d}\int_{\fldreal^d}n(k)\opdmsr{k}}^2
\\ 
&\quad
-\frac{1}{\sminvtemperature}
\frac{1}{(2\pi)^d}
\int_{\fldreal^d}
\rbk{\rbk{1+n(k)}
\fun{\log}{1+n(k)}
-n(k) \log n(k)}
\opdmsr{k}.
\end{aligned}$$
At the density $\bar{\smnumberdensity}_{\txtbsn}$ selected by Theorem \ref{expedition0012647},
we evaluate the Euler equations at $\rho=\bar{\smnumberdensity}_{\txtbsn}$.
The first variation in $n(k)$ is
\begin{equation}\label{expedition0012729}
\frac{\delta\smfreeenergy_{\txtmeanfield}}{\delta n(k)}
=
\physham[h]_{\txtparticle,0}(k)-\smchemicalpotential_{\mathrm{sel}}+\lambda\bar{\smnumberdensity}_{\txtbsn}
-\frac{1}{\sminvtemperature}
\log\frac{1+n(k)}{n(k)}.
\end{equation}
Thus the stationarity equation $\frac{\delta\smfreeenergy_{\txtmeanfield}}{\delta n(k)}=0$ gives
$$n(k)
=
\frac{1}
{\napiernum^{\sminvtemperature
\rbk{\physham[h]_{\txtparticle,0}(k)-\smchemicalpotential_{\mathrm{sel}}+\lambda\bar{\smnumberdensity}_{\txtbsn}}}-1}.$$
For the zero mode, the one-sided variation is
$$\frac{\partial\smfreeenergy_{\txtmeanfield}}{\partial\rho_0}
=
\lambda\bar{\smnumberdensity}_{\txtbsn}-\smchemicalpotential_{\mathrm{sel}}.$$
Since the constraint is $\rho_0\geq0$, minimization over the half-line gives the complementarity condition
\begin{equation}\label{expedition0018044}
\rho_0\rbk{\lambda\bar{\smnumberdensity}_{\txtbsn}-\smchemicalpotential_{\mathrm{sel}}}=0,
\quad
\lambda\bar{\smnumberdensity}_{\txtbsn}-\smchemicalpotential_{\mathrm{sel}}\geq0.
\end{equation}
\end{prop}

\begin{proof}
The free Bose gas contributes the energy density
$\frac{1}{(2\pi)^d}
\int_{\fldreal^d}
\physham[h]_{\txtparticle,0}(k) n(k)
\opdmsr{k}$
from the nonzero modes.
The zero mode has zero kinetic energy, and the grand-canonical chemical-potential term contributes
$-\smchemicalpotential_{\mathrm{sel}}\rho
=
-\smchemicalpotential_{\mathrm{sel}}\rho_0
-\frac{\smchemicalpotential_{\mathrm{sel}}}{(2\pi)^d}
\int_{\fldreal^d}n(k)\opdmsr{k}$.
The mean-field interaction contributes
$$\frac{\lambda}{2}\rho^2
=
\frac{\lambda}{2}
\rbk{\rho_0+\frac{1}{(2\pi)^d}\int_{\fldreal^d}n(k)\opdmsr{k}}^2.$$
For the bosonic occupation density $n(k)$, the entropy density is
$$\frac{1}{(2\pi)^d}
\int_{\fldreal^d}
\rbk{\rbk{1+n(k)}
\fun{\log}{1+n(k)}
-n(k)\log n(k)}
\opdmsr{k}.$$
The free energy is energy minus $\sminvtemperature^{-1}$ times entropy, so the formula for $\smfreeenergy_{\txtmeanfield}$ follows.

The derivative of the entropy integrand
$\rbk{1+x}
\fun{\log}{1+x} - x \log x$
is $\log\frac{1+x}{x}$.
Taking the first variation of $\smfreeenergy_{\txtmeanfield}$ with respect to
$n(k)$ and then inserting $\rho=\bar{\smnumberdensity}_{\txtbsn}$ gives
\eqref{expedition0012729}.
Solving $\frac{\delta\smfreeenergy_{\txtmeanfield}}{\delta n(k)}=0$ gives the Bose occupation formula.

For the zero mode, differentiating $\smfreeenergy_{\txtmeanfield}$ with respect to $\rho_0$ gives
$\frac{\partial\smfreeenergy_{\txtmeanfield}}{\partial\rho_0}
=
\lambda\rho-\smchemicalpotential_{\mathrm{sel}}$.
At the selected density this becomes
$\lambda\bar{\smnumberdensity}_{\txtbsn}-\smchemicalpotential_{\mathrm{sel}}$.
The admissible set for $\rho_0$ is the half-line $\closedinterval{0}{\infty}$.
If the minimum is attained at $\rho_0=0$, the one-sided derivative must be nonnegative:
$$\lambda\bar{\smnumberdensity}_{\txtbsn}-\smchemicalpotential_{\mathrm{sel}}\geq0
\quad\text{at }\rho_0=0.$$
If the minimum is attained at an interior point $\rho_0>0$, the ordinary derivative must vanish:
$$\lambda\bar{\smnumberdensity}_{\txtbsn}-\smchemicalpotential_{\mathrm{sel}}=0
\quad\text{when }\rho_0>0.$$
These two cases are equivalent to the complementarity condition.
\end{proof}

\begin{thm}[Mean-field equilibrium equations]
Let $\bar{\smnumberdensity}_{\txtbsn}$ be the limiting density selected
by Theorem \ref{expedition0012647},
and let $\smnumberdensity_{\txtbsn,0}(\sminvtemperature)$ be the condensate density
from \eqref{expedition0020735}.
Then the equilibrium equations are
\begin{equation}\label{expedition0012604}
\smnumberdensity_{\txtbsn,0}(\sminvtemperature)
\rbk{\smchemicalpotential_{\mathrm{sel}}-\lambda\bar{\smnumberdensity}_{\txtbsn}}=0
\end{equation}
and, for $k\ne0$,
\begin{equation}\label{expedition0012605}
E_{\bar{\smnumberdensity}_{\txtbsn}}(k)=\physham[h]_{\txtparticle,0}(k)-\smchemicalpotential_{\mathrm{sel}}+\lambda\bar{\smnumberdensity}_{\txtbsn},
\quad
n_{\bar{\smnumberdensity}_{\txtbsn}}(k)=
\frac{1}{\napiernum^{\sminvtemperature E_{\bar{\smnumberdensity}_{\txtbsn}}(k)}-1}.
\end{equation}
In the condensed or singular case,
i.e. $\smnumberdensity_{\txtbsn,0}(\sminvtemperature)>0$,
it holds that
\begin{equation}\label{expedition0012606}
\smchemicalpotential_{\mathrm{sel}}=\lambda\bar{\smnumberdensity}_{\txtbsn},
\quad
\bar{\smnumberdensity}_{\txtbsn}
=\smnumberdensity_{\txtbsn,0}(\sminvtemperature)
+\frac{1}{(2\pi)^d}
\int_{\fldreal^d}
\frac{1}
{\napiernum^{\sminvtemperature\physham[h]_{\txtparticle,0}(k)}-1}
\opdmsr{k}.
\end{equation}
In the normal case,
i.e. $\smnumberdensity_{\txtbsn,0}(\sminvtemperature)=0$,
it holds that
\begin{equation}\label{expedition0012607}
\bar{\smnumberdensity}_{\txtbsn}
=
\frac{1}{(2\pi)^d}
\int_{\fldreal^d}
\frac{1}
{\napiernum^{\sminvtemperature(\physham[h]_{\txtparticle,0}(k)-\smchemicalpotential_{\mathrm{sel}}+\lambda\bar{\smnumberdensity}_{\txtbsn})}-1}
\opdmsr{k}.
\end{equation}
\end{thm}

These are precisely the equilibrium equations in \cite[Section 4.4]{AndreVerbeure001}.

\begin{proof}
Proposition \ref{expedition0012649} gives the nonzero-mode occupation formula \eqref{expedition0012605}.
With $\rho_0=\smnumberdensity_{\txtbsn,0}(\sminvtemperature)$, the complementarity condition \eqref{expedition0018044} is equivalent to \eqref{expedition0012604}.
If $\smnumberdensity_{\txtbsn,0}(\sminvtemperature)>0$, the complementarity condition forces
$\smchemicalpotential_{\mathrm{sel}}=\lambda\bar{\smnumberdensity}_{\txtbsn}$.
Substituting this into the nonzero-mode occupation formula gives the free thermal occupation, and adding the condensate density gives \eqref{expedition0012606}.
If $\smnumberdensity_{\txtbsn,0}(\sminvtemperature)=0$, all density is carried by nonzero modes, and integrating the nonzero-mode occupation formula gives \eqref{expedition0012607}.
\end{proof}

The Kac-measure collapse selects the total density. It does not by itself identify the singularity of the KMS state. The zero-mode input used in the resolvent-algebra BEC construction is obtained only after the Euler equation and the complementarity relation have selected one of the cases \eqref{expedition0012606}--\eqref{expedition0012607} and the nonzero-mode density has been compared with its saturated value. Thus the scalar density selection in Theorem \ref{expedition0012647}, the Euler equations, the zero-mode excess density, and the zero-mode point-mass covariance are distinct steps. In particular, Theorem \ref{expedition0012647} alone gives only \(K_{\sminvtemperature,\lambda,\smchemicalpotential,L}\Rightarrow
\delta_{\bar{\smnumberdensity}_{\txtbsn}}\). The BEC state construction uses the subsequent implication \(\bar{\smnumberdensity}_{\txtbsn}>\smnumberdensity_{\txtbsn,\txtcritical}(\sminvtemperature)\), hence \(\smnumberdensity_{\txtbsn,0}(\sminvtemperature)>0\), hence a nonzero zero-mode covariance form \eqref{expedition0012440}.

\begin{prop}[Mean-field excess density and zero-mode input]
In the condensed case \eqref{expedition0012606} of the mean-field
equilibrium equations, the nonzero-mode occupation formula
\eqref{expedition0012605} is saturated at this value, and the zero-mode
density entering the zero-mode covariance form \eqref{expedition0012440} is
$$\smnumberdensity_{\txtbsn,0}(\sminvtemperature)
=
\bar{\smnumberdensity}_{\txtbsn}-\smnumberdensity_{\txtbsn,\txtcritical}(\sminvtemperature)>0.$$
In the normal case \eqref{expedition0012607}, the same formula is replaced by $\smnumberdensity_{\txtbsn,0}(\sminvtemperature)=0$ and the density equation is \eqref{expedition0012607}.
Consequently the Kac collapse fixes the total density.
\end{prop}

The BEC state constructed below uses the positive zero-mode excess density supplied by this saturation step.

\begin{proof}
If $\smnumberdensity_{\txtbsn,0}(\sminvtemperature)>0$, the complementarity relation \eqref{expedition0012604} gives $\smchemicalpotential_{\mathrm{sel}}=\lambda\bar{\smnumberdensity}_{\txtbsn}$.
Substitution into \eqref{expedition0012605} yields the free nonzero-mode occupation
$$n_{\bar{\smnumberdensity}_{\txtbsn}}(k)=\frac{1}{\napiernum^{\sminvtemperature\physham[h]_{\txtparticle,0}(k)}-1},\quad k\ne0.$$
Integrating this occupation gives $\smnumberdensity_{\txtbsn,\txtcritical}(\sminvtemperature)$.
The total density in \eqref{expedition0012606} is therefore the sum of this saturated nonzero-mode density and the zero-mode density.
Solving for the latter gives the asserted formula.
In the normal case \eqref{expedition0012607}, the complementary variable vanishes, and the density is carried by the nonzero modes with the effective energy in \eqref{expedition0012605}.
\end{proof}

\subsection{Order Parameter and Gauge Action Before Representation}\label{order-parameter-and-gauge-action-before-representation}

At this point the Kac-measure collapse has selected the density \(\bar{\smnumberdensity}_{\txtbsn}\), but no BEC state has yet been constructed. Therefore neither a represented center nor a BEC state defined by direct integral of the component states is defined here. Only the bounded zero-mode tests and the gauge action are available at the abstract resolvent-algebra level.

Define the test functions as \[\mathsf{b}_L^{(0)}
=
\frac{1}{V_L^{\onehalf}}\fndef{I_L^d},
\quad
\mathsf{b}_L^{(1)}
=
\frac{1}{V_L}\fndef{I_L^d},
\quad
\mathsf{b}_L^{(\#)}
=
\frac{1}{V_L^{\frac{1+\#}{2}}}\fndef{I_L^d},
\quad
\#
=0,1.\] These functions define bounded resolvent tests \(\fun{\oaresolvent}
{1,\mathsf{b}_L^{(\#)}}
\in\oaresolventalgebra\) for \(\#
=0,1,\) which probe the zero-momentum mode in the thermodynamic limit.

\begin{lem}[Fourier normalization of the order-parameter approximants]\label{expedition0012631}
On the finite-volume momentum lattice, Fourier transforms satisfy
\begin{equation}\label{expedition0012730}
\faftr{\mathsf{b}_L^{(0)}}(k)
=
\frac{V_L^{\onehalf}}{(2\pi)^{\frac{d}{2}}}\kroneckerdelta_{k,0},
\quad
\faftr{\mathsf{b}_L^{(1)}}(k)
=
\frac{1}{(2\pi)^{\frac{d}{2}}}\kroneckerdelta_{k,0}.
\end{equation}
Moreover it holds that
$$\twonorm{\mathsf{b}_L^{(0)}}=1,
\quad
\int_{I_L^d}\mathsf{b}_L^{(1)}(x)\opdmsr{x}=1,
\quad
\twonorm{\mathsf{b}_L^{(1)}}=\frac{1}{V_L^{\onehalf}}\to0.$$
\end{lem}

\begin{proof}
The functions $\mathsf{b}_L^{(0)}$ and $\mathsf{b}_L^{(1)}$ are constant on $I_L^d$ and vanish outside it.
With the Fourier normalization used in Subsection \ref{expedition0012145}, the constant function on $I_L^d$ has only the zero Fourier coefficient.
Multiplying that coefficient by $V_L^{-\onehalf}$ and by $V_L^{-1}$ gives \eqref{expedition0012730}.
The norm and integral identities follow from $\int_{I_L^d}1\opdmsr{x}
= V_L$.
\end{proof}

The two normalizations play different roles. The function \(\mathsf{b}_L^{(0)}\) is the \(L^2\)-normalized zero-mode wave and is the global analogue of the local support vector \(s_{\mathcal{O}}\) used in the Buchholz support criterion \cite{DetlevBuchholz005}. The function \(\mathsf{b}_L^{(1)}\) is the field-average test for the represented order parameter: its \(L^2\)-norm tends to zero, but its zero mode Fourier coefficient stays normalized in the macroscopic finite-volume representation. Thus test functions with nonzero zero Fourier coefficient pair with \(\mathsf{b}_L^{(0)}\)-type tests, whereas the represented center is obtained from \(\mathsf{b}_L^{(1)}\)-type macroscopic averages.

\begin{prop}[Asymptotic centrality of the zero-mode tests]\label{expedition0012632}
Let $A\in\oaresolventalgebra_{\txtloc}$ be fixed.
Then it holds that, for $\#=0,1$,
$$\lim_{L\to\infty}
\norm{\commutator{\fun{\oaresolvent}{1,\mathsf{b}_L^{(\#)}}}
{A}}
=0.$$
\end{prop}

\begin{proof}
It is enough to prove the assertion for products of local generators and then pass to the local $C^*$-closure.
For one generator $\oaresolvent(z,f)$ with $f$ belonging to a fixed finite-dimensional local subspace, the resolvent relation gives
$$\commutator{\fun{\oaresolvent}{1,\mathsf{b}_L^{(\#)}}}
{\oaresolvent(z,f)}
=
\imunit
\fun{\sigma}{\mathsf{b}_L^{(\#)},f}
\fun{\oaresolvent}{1,\mathsf{b}_L^{(\#)}}
\oaresolvent(z,f)^2
\fun{\oaresolvent}{1,\mathsf{b}_L^{(\#)}}.$$
Consequently we obtain
$$\norm{\commutator{\fun{\oaresolvent}{1,\mathsf{b}_L^{(\#)}}}
{\oaresolvent(z,f)}}
\leq
C_z
\abs{\fun{\sigma}{\mathsf{b}_L^{(\#)},f}}.$$
Since $f$ is fixed and local, its zero-momentum coefficient is bounded independently of $L$.
By Lemma \ref{expedition0012631},
the pairing with $\mathsf{b}_L^{(0)}$ is of order $V_L^{-\onehalf}$
after the finite-volume normalization of the symplectic form,
and the pairing with $\mathsf{b}_L^{(1)}$ is at most of order $V_L^{-\onehalf}$
because $\twonorm{\mathsf{b}_L^{(1)}}=V_L^{-\onehalf}$.
Thus the commutator tends to zero for each generator.
The Leibniz rule for commutators gives the same conclusion for finite products.
Let $A$ first belong to the local algebra $\oaresolventalgebra(\mathcal{O})$ of a fixed bounded region $\mathcal{O}$.
For all sufficiently large $L$, $\mathcal{O}\subset I_L^d$, and the finite-volume computation above applies to the projected local test functions through the inclusion/projection used to compare local observables with the periodic box.
Finite products of generators in this fixed local algebra are therefore covered by the preceding estimate, and norm approximation inside the local $C^*$-closure together with the uniform bound
$\norm{\commutator{\fun{\oaresolvent}{1,\mathsf{b}_L^{(\#)}}}
A}
\leq 2\norm{A}$
gives the assertion for every fixed element of $\oaresolventalgebra(\mathcal{O})$.
Taking the union over bounded regions gives the statement for
$\oaresolventalgebra_{\txtloc}$.
\end{proof}

For \(\theta_0\in\fldreal\), define the gauge automorphism on local generators by \[\fun{\gamma_{\theta_0}}{\oaresolvent(z,f)}
=
\fun{\oaresolvent}{z,\napiernum^{\imunit\theta_0}f}.\] The definition is compatible with the canonical inclusions of the local resolvent algebras and hence gives an automorphism of the resolvent algebra.

\begin{prop}[Gauge covariance of the zero-mode tests]
For $\#=0,1$ and every $\theta_0\in\fldreal$,
$$\fun{\gamma_{\theta_0}}{\fun{\oaresolvent}{1,\mathsf{b}_L^{(\#)}}}
=
\fun{\oaresolvent}
{1,\napiernum^{\imunit\theta_0}\mathsf{b}_L^{(\#)}}.$$
The asymptotic centrality statement of Proposition \ref{expedition0012632} is unchanged after applying $\gamma_{\theta_0}$.
\end{prop}

\begin{proof}
The first formula is the defining action of $\gamma_{\theta_0}$ on generators.
Since multiplication by a constant phase preserves the norm and only rotates the symplectic pairing,
$\abs{\fun{\sigma}{\napiernum^{\imunit\theta_0}\mathsf{b}_L^{(\#)},f}}$
has the same order in $L$ as
$\abs{\fun{\sigma}{\mathsf{b}_L^{(\#)},f}}$.
The proof of Proposition \ref{expedition0012632} therefore applies verbatim.
\end{proof}

Thus the abstract algebra records the zero-mode line and the gauge action on it. At this stage this is only an abstract resolvent-algebra statement on the zero-mode support and its gauge rotation; the BEC state obtained by integrating the component states and any represented center are constructed separately below.

\subsection{Extremal Decomposition States and Direct Integral}\label{extremal-decomposition-states-and-direct-integral}

The BEC state is introduced through its component states and their direct integral. The construction uses the selected density \(\bar{\smnumberdensity}_{\txtbsn}\) from Theorem \ref{expedition0012647}, the nonzero-mode occupation density \(n_{\bar{\smnumberdensity}_{\txtbsn}}\) in \eqref{expedition0012605}, the condensate density fixed by \eqref{expedition0012604}, and the zero-mode covariance form \eqref{expedition0012440}. The zero-mode part is represented on each component by a scalar character, while the regular part is represented on the kernel of the zero-mode covariance. The regular nonzero-mode reference algebra is the resolvent algebra over \(\Ker \opform{q}_{\txtbsn,0,\sminvtemperature}\) with the restricted symplectic form. The condensate ideal and its nonregular representation are defined only after the BEC state and its GNS representation have been constructed.

\subsubsection{Component States and Integrated BEC State}\label{component-states-and-integrated-bec-state}

The component-state construction fixes the zero-mode scalar shift, the regular nonzero-mode quasi-free state, and the integrated BEC state before the condensate ideal is introduced.

\begin{defn}[Zero-mode shift functional]
For $\theta\in[0,2\pi)$, the term zero-mode shift functional means the real linear functional
$\ell_{\sminvtemperature,\bar{\smnumberdensity}_{\txtbsn},\theta}\colon \sphilb{D}_{\txtbsn,\sminvtemperature,\smchemicalpotential_{\mathrm{sel}}}\to\fldreal$
given by
\begin{equation}\label{expedition0012610}
\ell_{\sminvtemperature,\bar{\smnumberdensity}_{\txtbsn},\theta}(f)
=
\sqrt{2(2\pi)^d\smnumberdensity_{\txtbsn,0}(\sminvtemperature)}
\opreal\rbk{\napiernum^{\imunit\theta}\faftr{f}(0)}.
\end{equation}
\end{defn}

This definition uses only the condensate density from \eqref{expedition0012604} and the zero-momentum value of the test function. It is a scalar input for the component states; the represented center is obtained later in the direct-integral GNS representation. Furthermore this functional is used in the free Bose gas: see \cite{AsaoArai28,YoshitsuguSekine004}.

Assume \(\smnumberdensity_{\txtbsn,0}(\sminvtemperature)>0\). By \eqref{expedition0012606}, the selected chemical potential exactly cancels the mean-field shift. The selected-density one-particle Hamiltonian is therefore defined by \begin{equation}\label{expedition0012694}
\physham[h]_{\bar{\smnumberdensity}_{\txtbsn}}(k)
=
\physham[h]_{\txtparticle,0}(k)
-\smchemicalpotential_{\mathrm{sel}}
+\lambda\bar{\smnumberdensity}_{\txtbsn}
=
\physham[h]_{\txtparticle,0}(k).
\end{equation} Thus, in the BEC regime, \(\physham[h]_{\bar{\smnumberdensity}_{\txtbsn}}\) reduces exactly to the free-field one-particle Hamiltonian. The subscript \(\bar{\smnumberdensity}_{\txtbsn}\) records the selected-density origin of the construction, not an additional interaction term in the resulting dynamics. For \(f\in\sphilb{D}_{\txtbsn,\sminvtemperature,\smchemicalpotential_{\mathrm{sel}}}\), write \(f_{\txtbec}\) for its regular nonzero-mode part in \(\Ker \opform{q}_{\txtbsn,0,\sminvtemperature}\). This notation records the regular nonzero-mode projection. The value of \(f\) appearing in \(\opform{q}_{\txtbsn,0,\sminvtemperature}\) is handled separately by the scalar functional \eqref{expedition0012610}. Let \(\opform{q}_{\txtmeanfield,\txtnonzero,\sminvtemperature,\bar{\smnumberdensity}_{\txtbsn}}\) be the quasi-bilinear form on \(\fun{\oaresolventalgebra}
{\Ker \opform{q}_{\txtbsn,0,\sminvtemperature},
\fnrestr{\sigma}{\Ker \opform{q}_{\txtbsn,0,\sminvtemperature}}}\) defined by \[\fun{\opform{q}_{\txtmeanfield,\txtnonzero,\sminvtemperature,\bar{\smnumberdensity}_{\txtbsn}}}
{f_{\txtbec},g_{\txtbec}}
=
\int_{\fldreal^d\setminus\setone{0}}
\coth\frac{\sminvtemperature
\physham[h]_{\bar{\smnumberdensity}_{\txtbsn}}(k)}{2}
\cmpconj{\faftr{f}(k)}
\faftr{g}(k)
\opdmsr{k}.\] Equivalently, by \eqref{expedition0012605} and \eqref{expedition0012606}, the multiplier is \(1+2n_{\bar{\smnumberdensity}_{\txtbsn}}(k)\). The regular nonzero-mode automorphism group used in the KMS statement is \(\alpha_{\txtmeanfield,\bar{\smnumberdensity}_{\txtbsn},\txtnonzero,t}\) on the resolvent algebra over \(\Ker \opform{q}_{\txtbsn,0,\sminvtemperature}\), defined on resolvent generators by \begin{equation}\label{expedition0018017}
\fun{\alpha_{\txtmeanfield,\bar{\smnumberdensity}_{\txtbsn},\txtnonzero,t}}
{\oaresolvent(z,f_{\txtbec})}
=
\fun{\oaresolvent}
{z,\napiernum^{\imunit t \physham[h]_{\bar{\smnumberdensity}_{\txtbsn}}}f_{\txtbec}}.
\end{equation} The elementary invariance properties needed for the KMS statement are verified in Proposition \ref{expedition0012677}, and the KMS property is proved in Proposition \ref{expedition0012653}.

\begin{defn}[Selected-density BEC states]\label{expedition0012650}
In the condensed case $\smnumberdensity_{\txtbsn,0}(\sminvtemperature)>0$,
define the state
$\oastate[\psi_{\txtmeanfield,\txtnonzero,\sminvtemperature,\bar{\smnumberdensity}_{\txtbsn}}]$
to be the regular gauge-invariant quasi-free $\sminvtemperature$-KMS state for
$\alpha_{\txtmeanfield,\bar{\smnumberdensity}_{\txtbsn},\txtnonzero,t}$
on the resolvent algebra over $\Ker \opform{q}_{\txtbsn,0,\sminvtemperature}$,
with one-particle occupation density $n_{\bar{\smnumberdensity}_{\txtbsn}}$ and
two-point resolvent functions
\begin{equation}\label{expedition0018018}
\begin{aligned}
&\fun{\oastate[\psi_{\txtmeanfield,\txtnonzero,\sminvtemperature,\bar{\smnumberdensity}_{\txtbsn}}]}
{\oaresolvent(z,f_{\txtbec})\oaresolvent(w,g_{\txtbec})}
\\ 
&=
-\int_0^{(\sgn\opreal z)\infty}
\int_0^{(\sgn\opreal w)\infty}
\fnexp{-zt-ws}
\fnexp{-\frac{\imunit}{2}ts\fun{\sigma}{f_{\txtbec},g_{\txtbec}}}
\\ 
&\quad\times
\fnexp{-\oneoverfour
\fun{\opform{q}_{\txtmeanfield,\txtnonzero,\sminvtemperature,\bar{\smnumberdensity}_{\txtbsn}}}
{tf_{\txtbec}+sg_{\txtbec}}}
\opdmsr{s}
\opdmsr{t}.
\end{aligned}
\end{equation}
The values on finite products of resolvent generators are fixed by this two-point function and the quasi-free property.
For each $\theta\in[0,2\pi)$, define
$\oastate[\psi_{\txtmeanfield,\sminvtemperature,\bar{\smnumberdensity}_{\txtbsn},\theta}]$
on the finite resolvent algebra inside $\fun{\oaresolventalgebra}{\sphilb{D}_{\txtbsn,\sminvtemperature,\smchemicalpotential_{\mathrm{sel}}},\sigma}$ by
\begin{equation}\label{expedition0018019}
\fun{\oastate[\psi_{\txtmeanfield,\sminvtemperature,\bar{\smnumberdensity}_{\txtbsn},\theta}]}
{\prod_{j=1}^{m}
\oaresolvent(z_j,f_j)}
=
\fun{\oastate[\psi_{\txtmeanfield,\txtnonzero,\sminvtemperature,\bar{\smnumberdensity}_{\txtbsn}}]}
{\prod_{j=1}^{m}
\fun{\oaresolvent}
{z_j+\imunit\ell_{\sminvtemperature,\bar{\smnumberdensity}_{\txtbsn},\theta}(f_j),f_{j,\txtbec}}}.
\end{equation}
Its continuous extension to the full resolvent algebra is denoted by the same symbol
$\oastate[\psi_{\txtmeanfield,\sminvtemperature,\bar{\smnumberdensity}_{\txtbsn},\theta}]$.
Define the BEC state obtained by integrating the component states by
\begin{equation}\label{expedition0012609}
\fun{\oastate[\psi_{\txtmeanfield,\sminvtemperature,\bar{\smnumberdensity}_{\txtbsn}}]}{A}
=
\int_0^{2\pi}
\fun{\oastate[\psi_{\txtmeanfield,\sminvtemperature,\bar{\smnumberdensity}_{\txtbsn},\theta}]}{A}
\frac{\opdmsr{\theta}}{2\pi},
\quad
A\in\fun{\oaresolventalgebra}{\sphilb{D}_{\txtbsn,\sminvtemperature,\smchemicalpotential_{\mathrm{sel}}},\sigma}.
\end{equation}
Let
$\oarepn_{\txtmeanfield,\sminvtemperature,\bar{\smnumberdensity}_{\txtbsn}}$
be its GNS representation and set
$$\oa{M}_{\txtmeanfield,\sminvtemperature,\bar{\smnumberdensity}_{\txtbsn}}
=
\oadoublecommutant{\oarepn_{\txtmeanfield,\sminvtemperature,\bar{\smnumberdensity}_{\txtbsn}}
(\oaresolventalgebra)}.$$
\end{defn}

Formula \eqref{expedition0018019} pulls back the regular nonzero-mode state along the projection \(f\mapsto f_{\txtbec}\) and the scalar resolvent shift \(z\mapsto
z+\imunit\ell_{\sminvtemperature,\bar{\smnumberdensity}_{\txtbsn},\theta}(f)\). This sign is the resolvent-parameter convention fixed in Subsection \ref{expedition0011278}. The integrated state \eqref{expedition0012609} is the phase average of these component states, and its relation with finite-volume gauge-invariant Gibbs states is supplied by the subsequent order-parameter and direct-integral analysis. If \(f_{\txtbec}=0\) and \(\fun{\opform{q}_{\txtbsn,0,\sminvtemperature}}{f}>0\), the component value of the shifted resolvent is a scalar depending only on \(\theta\).

\subsubsection{Selected-Density Mean-Field Automorphism Group and KMS Property}\label{selected-density-mean-field-automorphism-group-and-kms-property}

The Kac-measure collapse in Theorem \ref{expedition0012647} selects the density \(\bar{\smnumberdensity}_{\txtbsn}\). In the condensed phase, the selected-density dynamics is the free-field one-particle dynamics written in \eqref{expedition0012694}. Proposition \ref{expedition0012677} checks compatibility with the projection \(f\mapsto f_{\txtbec}\) and the zero-mode shift \eqref{expedition0012610}, and Proposition \ref{expedition0012653} proves the \(\sminvtemperature\)-KMS property of the component and integrated BEC states. The selected-density mean-field automorphism group is specified on resolvent generators by \begin{equation}\label{expedition0012675}
\fun{\alpha_{\txtmeanfield,\bar{\smnumberdensity}_{\txtbsn},t}}{\oaresolvent(z,f)}
=
\fun{\oaresolvent}
{z,\napiernum^{\imunit t \physham[h]_{\bar{\smnumberdensity}_{\txtbsn}}}f}.
\end{equation} For each \(t\in\fldreal\), the multiplier in \eqref{expedition0012675} gives the real-linear symplectic map \(T_t f
=
\napiernum^{\imunit t \physham[h]_{\bar{\smnumberdensity}_{\txtbsn}}}
f\). The universal property of the resolvent algebra then gives the corresponding one-parameter automorphism group of \(\oaresolventalgebra\).

\begin{prop}[Compatibility with the regular nonzero-mode projection and zero-mode shift]\label{expedition0012677}
Assume $\smnumberdensity_{\txtbsn,0}(\sminvtemperature)>0$.
For every $t\in\fldreal$, the multiplier
$\napiernum^{\imunit t \physham[h]_{\bar{\smnumberdensity}_{\txtbsn}}}$
preserves $\Ker\opform{q}_{\txtbsn,0,\sminvtemperature}$.
Moreover, for every $\theta\in[0,2\pi)$ and
$f\in\sphilb{D}_{\txtbsn,\sminvtemperature,\smchemicalpotential_{\mathrm{sel}}}$,
it holds that
\begin{equation}\label{expedition0012731}
\rbk{\napiernum^{\imunit t
\physham[h]_{\bar{\smnumberdensity}_{\txtbsn}}}f}_{\txtbec}
=
\napiernum^{\imunit t
\physham[h]_{\bar{\smnumberdensity}_{\txtbsn}}}f_{\txtbec},
\quad
\ell_{\sminvtemperature,\bar{\smnumberdensity}_{\txtbsn},\theta}
\rbk{\napiernum^{\imunit t \physham[h]_{\bar{\smnumberdensity}_{\txtbsn}}}f}
=
\ell_{\sminvtemperature,\bar{\smnumberdensity}_{\txtbsn},\theta}(f).
\end{equation}
\end{prop}

\begin{proof}
The operator
$\napiernum^{\imunit t \physham[h]_{\bar{\smnumberdensity}_{\txtbsn}}}$
acts diagonally in momentum space, so it preserves the zero-momentum subspace and its complement.
Consequently it preserves
$\Ker\opform{q}_{\txtbsn,0,\sminvtemperature}$, and the first displayed
identity follows from the definition of $f_{\txtbec}$.

By \eqref{expedition0012694},
$\physham[h]_{\bar{\smnumberdensity}_{\txtbsn}}(0)
=
\physham[h]_{\txtparticle,0}(0)
=0$.
The functional
$\ell_{\sminvtemperature,\bar{\smnumberdensity}_{\txtbsn},\theta}$
in \eqref{expedition0012610} depends only on the zero-momentum Fourier coefficient.
That coefficient is unchanged by multiplication with
$\napiernum^{\imunit t \physham[h]_{\bar{\smnumberdensity}_{\txtbsn}}}$,
so \eqref{expedition0012731} follows.
\end{proof}

\begin{prop}[KMS property for the selected-density mean-field automorphism group]\label{expedition0012653}
Assume $\smnumberdensity_{\txtbsn,0}(\sminvtemperature)>0$.
The state
$\oastate[\psi_{\txtmeanfield,\txtnonzero,\sminvtemperature,\bar{\smnumberdensity}_{\txtbsn}}]$
defined by \eqref{expedition0018018}
is the $\sminvtemperature$-KMS state for the regular nonzero-mode automorphism
group \eqref{expedition0018017}.
Equivalently, this is the automorphism group induced by
\eqref{expedition0012675} on
$\fun{\oaresolventalgebra}
{\Ker \opform{q}_{\txtbsn,0,\sminvtemperature},
\fnrestr{\sigma}{\Ker \opform{q}_{\txtbsn,0,\sminvtemperature}}}$.
For each $\theta\in[0,2\pi)$, the state
$\oastate[\psi_{\txtmeanfield,\sminvtemperature,\bar{\smnumberdensity}_{\txtbsn},\theta}]$
defined by \eqref{expedition0018019}
is an $\sminvtemperature$-KMS state for
$\alpha_{\txtmeanfield,\bar{\smnumberdensity}_{\txtbsn},t}$.
The state
$\oastate[\psi_{\txtmeanfield,\sminvtemperature,\bar{\smnumberdensity}_{\txtbsn}}]$
defined by the integral in \eqref{expedition0012609}
is also an $\sminvtemperature$-KMS state for
$\alpha_{\txtmeanfield,\bar{\smnumberdensity}_{\txtbsn},t}$.
On the direct-integral GNS von Neumann algebra
$\oa{M}_{\txtmeanfield,\sminvtemperature,\bar{\smnumberdensity}_{\txtbsn}}$,
the corresponding normal automorphism group is
$\alpha_{\txtmeanfield,\sminvtemperature,\bar{\smnumberdensity}_{\txtbsn},t}$, defined by
\begin{equation}\label{expedition0012732}
\fun{\alpha_{\txtmeanfield,\sminvtemperature,\bar{\smnumberdensity}_{\txtbsn},t}}
{\fun{\oarepn_{\txtmeanfield,\sminvtemperature,\bar{\smnumberdensity}_{\txtbsn}}}{A}}
=
\fun{\oarepn_{\txtmeanfield,\sminvtemperature,\bar{\smnumberdensity}_{\txtbsn}}}
{\fun{\alpha_{\txtmeanfield,\bar{\smnumberdensity}_{\txtbsn},t}}{A}},
\quad
A\in\oaresolventalgebra.
\end{equation}
\end{prop}

\begin{proof}
The state \eqref{expedition0018018} is the standard gauge-invariant
quasi-free $\sminvtemperature$-KMS state for the free Bose one-particle
Hamiltonian $\physham[h]_{\bar{\smnumberdensity}_{\txtbsn}}$ on
$\Ker \opform{q}_{\txtbsn,0,\sminvtemperature}$.
Indeed, the covariance in Definition \ref{expedition0012650} is the usual
free Bose covariance for the dynamics \eqref{expedition0018017}.

The component state \eqref{expedition0018019} is obtained from this
nonzero-mode KMS state by the time-invariant scalar shift
$\ell_{\sminvtemperature,\bar{\smnumberdensity}_{\txtbsn},\theta}$.
Proposition \ref{expedition0012677} shows that the projection
$f\mapsto f_{\txtbec}$ intertwines the full dynamics with
\eqref{expedition0018017} and that the scalar shift is fixed by the dynamics.
The KMS analytic functions for the nonzero-mode state therefore pull back to
the component state.
Thus each component state is $\sminvtemperature$-KMS for
$\alpha_{\txtmeanfield,\bar{\smnumberdensity}_{\txtbsn},t}$.

The phase average \eqref{expedition0012609} of KMS states for the same
automorphism group is again a KMS state, by integrating the corresponding
bounded strip-analytic functions in $\theta$.
Finally, if
$\fun{\oarepn_{\txtmeanfield,\sminvtemperature,\bar{\smnumberdensity}_{\txtbsn}}}{A}=0$,
then the KMS invariance of the integrated state gives
$$\begin{aligned}
\fun{\oastate[\psi_{\txtmeanfield,\sminvtemperature,\bar{\smnumberdensity}_{\txtbsn}}]}
{\fun{\alpha_{\txtmeanfield,\bar{\smnumberdensity}_{\txtbsn},t}}{A^*A}}
=0.
\end{aligned}$$
Hence \eqref{expedition0012732} is well defined on the GNS representation and
extends normally to
$\oa{M}_{\txtmeanfield,\sminvtemperature,\bar{\smnumberdensity}_{\txtbsn}}$.
\end{proof}

\subsubsection{Order-Parameter Limit and Represented Center}\label{order-parameter-limit-and-represented-center}

The order-parameter resolvent tests identify the non-scalar central variable in the GNS representation of the integrated BEC state.

Although \(\twonorm{\mathsf{b}_L^{(1)}}\to0\), the order-parameter resolvents below have a nontrivial represented thermodynamic limit because they are macroscopic finite-volume averages evaluated in the selected representation. This is not a norm-continuity statement in the abstract resolvent algebra with respect to the \(L^2\) norm. The abstract test space has no topology forcing \(\fun{\oaresolvent}{1,\mathsf{b}_L^{(1)}}\) to converge to \(\oaresolvent(1,0)\) in \(C^*\)-norm, and the zero Fourier coefficient of \(\mathsf{b}_L^{(1)}\) is controlled by the \(L^1\) normalization rather than by its \(L^2\) norm.

\begin{prop}[Order-parameter test for selected-density BEC states]\label{expedition0012652}
For the component states defined by \eqref{expedition0018019}, it holds that
\begin{equation}\label{expedition0020736}
\begin{aligned}
\lim_{L\to\infty}
\fun{\oastate[\psi_{\txtmeanfield,\sminvtemperature,\bar{\smnumberdensity}_{\txtbsn},\theta}]}
{\fun{\oaresolvent}{1,\mathsf{b}_L^{(1)}}}
=
F_{\sminvtemperature,\bar{\smnumberdensity}_{\txtbsn}}(\theta),
\quad
F_{\sminvtemperature,\bar{\smnumberdensity}_{\txtbsn}}(\theta)
=
\frac{1}
{\imunit\rbk{1-\imunit\sqrt{2\smnumberdensity_{\txtbsn,0}(\sminvtemperature)}\cos\theta}}.
\end{aligned}
\end{equation}
\end{prop}

\begin{proof}
Lemma \ref{expedition0012631} gives
$\faftr{\mathsf{b}_L^{(1)}}(0)=(2\pi)^{-d/2}$ and
$\twonorm{\mathsf{b}_L^{(1)}}\to0$.
Substitution into \eqref{expedition0012610} gives
$\fun{\ell_{\sminvtemperature,\bar{\smnumberdensity}_{\txtbsn},\theta}}
{\mathsf{b}_L^{(1)}}
=
\sqrt{2 \smnumberdensity_{\txtbsn,0}(\sminvtemperature)}
\cos\theta$.
The centered part of the resolvent test disappears in the limit because
$\twonorm{\mathsf{b}_L^{(1)}}\to0$, and the limit follows from the resolvent identity.
\end{proof}

\begin{prop}[Mean-field BEC and the represented center in the extremal decomposition]\label{expedition0012651}
Assume $\smnumberdensity_{\txtbsn,0}(\sminvtemperature)>0$, and let
$\oarepn_{\txtmeanfield,\sminvtemperature,\bar{\smnumberdensity}_{\txtbsn}}$
and $\oa{M}_{\txtmeanfield,\sminvtemperature,\bar{\smnumberdensity}_{\txtbsn}}$
be the GNS data from Definition \ref{expedition0012650}.
Set
$$X_L =
\fun{\oarepn_{\txtmeanfield,\sminvtemperature,\bar{\smnumberdensity}_{\txtbsn}}}
{\fun{\oaresolvent}{1,\mathsf{b}_L^{(1)}}}.$$
The family $\fml{X_L}{L>0}$ has weak operator cluster points in
$\oa{M}_{\txtmeanfield,\sminvtemperature,\bar{\smnumberdensity}_{\txtbsn}}$.
Every such cluster point belongs to the center
$\oacenter(\oa{M}_{\txtmeanfield,\sminvtemperature,\bar{\smnumberdensity}_{\txtbsn}})$.
In the direct-integral decomposition associated with \eqref{expedition0012609},
it acts as multiplication by the function in \eqref{expedition0020736};
hence it is non-scalar.
\end{prop}

\begin{proof}
The operators $X_L$ are uniformly bounded, so they have weak operator cluster
points in the direct-integral GNS representation.
Proposition \ref{expedition0012632} gives asymptotic centrality against every local resolvent observable.
Since the local resolvent algebra is strongly dense in
$\oa{M}_{\txtmeanfield,\sminvtemperature,\bar{\smnumberdensity}_{\txtbsn}}$,
every such cluster point lies in the represented center.

The non-scalar character is best read through the central decomposition of the
state defined by the integral in
\eqref{expedition0012609}.
The state is the normalized integral over the component states in the extremal
decomposition, and its GNS representation decomposes, up to null sets, as the
direct integral of the corresponding representations.
On the direct-integral component indexed by $\theta$, \eqref{expedition0020736}
gives the scalar limit
$F_{\sminvtemperature,\bar{\smnumberdensity}_{\txtbsn}}(\theta)$.
Hence any weak cluster point of the central sequence acts in the direct integral as multiplication by this function of the direct-integral parameter.
Since
$\smnumberdensity_{\txtbsn,0}(\sminvtemperature)>0$, this function is not constant in $\theta$.
Multiplication by a nonconstant function of the direct-integral parameter is a non-scalar central operator.
\end{proof}

Thus the density is fixed by Theorem \ref{expedition0012647}, while the non-scalar represented central variable is obtained in Proposition \ref{expedition0012651} after the state \eqref{expedition0012609} has been constructed by integrating the component states.

\subsection{Point-Mass Covariance, Condensate Ideal, Quotient, and Direct-Integral Center}\label{point-mass-covariance-condensate-ideal-quotient-and-direct-integral-center}

\subsubsection{Condensate Ideal from the Zero-Mode Point Mass}\label{condensate-ideal-from-the-zero-mode-point-mass}

The zero-mode condensate ideal is an abstract \(\oacstar\)-algebraic device for recording the resolvents \(\oaresolvent(\lambda,f)\) whose test functions satisfy \(\fun{\opform{q}_{\txtbsn,0,\sminvtemperature}}{f}>0\). It cannot be a central decomposition of the full resolvent algebra, because the full resolvent algebra has trivial center \cite{BuchholzGrundling2}. Proposition \ref{expedition0012642} identifies its role as a kernel for the GNS representation of the BEC state constructed in Definition \ref{expedition0012650}. The definition below records, at the abstract resolvent-algebra level, the test functions satisfying \(\fun{\opform{q}_{\txtbsn,0,\sminvtemperature}}{f}>0\), which give nontrivial represented order-parameter limits in the BEC GNS representation of Definition \ref{expedition0012650}.

\begin{defn}[Mean-field zero-mode condensate ideal]\label{expedition0012641}
Assume $\smnumberdensity_{\txtbsn,0}(\sminvtemperature)>0$.
The mean-field zero-mode condensate ideal is the closed two-sided ideal
\begin{equation}\label{expedition0019398}
\begin{aligned}
\oaideal{J}_{\txtmeanfield,\txtbec}
=
\clos{\opideal
\set{\oaresolvent(\lambda,f)}
{\lambda\in\fldreal\setminus\setone{0},
f\in\sphilb{D}_{\txtbsn,\sminvtemperature,\smchemicalpotential_{\mathrm{sel}}},
\fun{\opform{q}_{\txtbsn,0,\sminvtemperature}}{f}>0}}
\subset
\oaresolventalgebra
\rbk{\sphilb{D}_{\txtbsn,\sminvtemperature,\smchemicalpotential_{\mathrm{sel}}},\sigma}.
\end{aligned}
\end{equation}
Recall the definition of the form $\opform{q}_{\txtbsn,0,\sminvtemperature}$
by \eqref{expedition0012440}.
\end{defn}

In complex one-particle notation the condition in Definition \ref{expedition0012641} is expressed by the zero Fourier coefficient \(\faftr{f}(0)\). The resolvent algebra is built on the realification of the one-particle space \(\sphilb{H}_{\txtparticle}\), and the complex zero-mode line is therefore represented by the real plane corresponding to the real and imaginary parts of the zero Fourier coefficient. After restriction to a bounded region \(\mathcal{O}\), the corresponding local complex singular line is \(\fldcmp s_{\mathcal{O}}\), where \(s_{\mathcal{O}}
=
\absvol{\mathcal{O}}^{-\onehalf}
\fndef{\mathcal{O}}\). In the real symplectic space this is the real plane \(\splinspan_{\fldreal}
\setone{s_{\mathcal{O}},\imunit s_{\mathcal{O}}}\). The ideal \eqref{expedition0019398} is not generated only by the condensate wave function itself. It is generated by every resolvent \(\oaresolvent(\lambda,f)\) for which \(f\) has nonzero zero-mode projection, equivalently by the condition \(\fun{\opform{q}_{\txtbsn,0,\sminvtemperature}}{f}>0\).

\begin{rem}
The role of $\oaideal{J}_{\txtmeanfield,\txtbec}$ is different from the density Kac measure.
The Kac measure records the density variable and collapses to
$\delta_{\bar{\smnumberdensity}_{\txtbsn}}$ in the mean-field equilibrium state.
The ideal records the resolvents whose test functions satisfy
$\fun{\opform{q}_{\txtbsn,0,\sminvtemperature}}{f}>0$ and on which the
represented order-parameter limit is still visible.
At finite condensate density this ideal records exactly the test functions
with $\fun{\opform{q}_{\txtbsn,0,\sminvtemperature}}{f}>0$ in the BEC state
constructed above.
In Proposition \ref{expedition0012642}, it becomes the kernel of the GNS
representation in Definition \ref{expedition0012650}.
This is the same operator-algebraic separation as the BEC ideal used
for the Nelson and Pauli--Fierz discussions \cite{YoshitsuguSekine009},
but here no infrared quotient is introduced.
\end{rem}

\subsubsection{Nonregular Quotient for a Proper Singularity}\label{nonregular-quotient-for-a-proper-singularity}

The quotient by the ideal \eqref{expedition0019398} is computed for the GNS representation in Definition \ref{expedition0012650}. The quotient algebra is the algebra obtained from the condensate ideal \eqref{expedition0019398}. The representation is \(\oarepn_{\txtmeanfield,\sminvtemperature,\bar{\smnumberdensity}_{\txtbsn}}\), the GNS representation of the state \eqref{expedition0012609}.

\begin{lem}[Regular subspace of a nonregular mean-field quotient representation]\label{expedition0019390}
For the BEC state \eqref{expedition0012609}, the GNS representation
$\oarepn_{\txtmeanfield,\sminvtemperature,\bar{\smnumberdensity}_{\txtbsn}}$
is nonregular for
$f\in\sphilb{D}_{\txtbsn,\sminvtemperature,\smchemicalpotential_{\mathrm{sel}}}$
with $\fun{\opform{q}_{\txtbsn,0,\sminvtemperature}}{f}>0$.
On $\Ker\opform{q}_{\txtbsn,0,\sminvtemperature}$ it is the GNS representation
of the regular nonzero-mode state from Definition \ref{expedition0012650};
the component states \eqref{expedition0018019} supply the zero-mode scalar
characters, and the nonregular representation is singular for
$f$ with $\fun{\opform{q}_{\txtbsn,0,\sminvtemperature}}{f}>0$.
Its regular subspace in the ideal-theoretic sense of the resolvent algebra
\cite{DetlevBuchholz001} is
$$X_R
=
\set{f\in\sphilb{D}_{\txtbsn,\sminvtemperature,\smchemicalpotential_{\mathrm{sel}}}}
{\text{$\fun{\oarepn_{\txtmeanfield,\sminvtemperature,\bar{\smnumberdensity}_{\txtbsn}}}{\oaresolvent(1,f)}$ is injective}}.$$
Then it follows that,
for the mean-field zero-mode covariance form \eqref{expedition0012440},
$$X_R
=
\Ker \opform{q}_{\txtbsn,0,\sminvtemperature}.$$
\end{lem}

\begin{proof}
For $f$ with
$\fun{\opform{q}_{\txtbsn,0,\sminvtemperature}}{f}=0$,
the component formula \eqref{expedition0018019} has no zero-mode scalar shift
and $\oarepn_{\txtmeanfield,\sminvtemperature,\bar{\smnumberdensity}_{\txtbsn}}$
restricts to the regular nonzero-mode
GNS representation from Definition \ref{expedition0012650}.
Regularity of the finite-dimensional resolvent subalgebras then
gives injectivity of
$\fun{\oarepn_{\txtmeanfield,\sminvtemperature,\bar{\smnumberdensity}_{\txtbsn}}}
{\oaresolvent(1,f)}$.
If $\fun{\opform{q}_{\txtbsn,0,\sminvtemperature}}{f}>0$,
the component states \eqref{expedition0018019} evaluate the same zero-mode data
through the $\theta$-dependent scalar character, while
$\oarepn_{\txtmeanfield,\sminvtemperature,\bar{\smnumberdensity}_{\txtbsn}}$
treats
that $f$ as singular.
The regularity criterion for resolvent representations
\cite{BuchholzGrundling2,DetlevBuchholz001} therefore puts such $f$ outside $X_R$.
These cases exhaust the positive semidefinite zero-mode form,
so $X_R = \Ker \opform{q}_{\txtbsn,0,\sminvtemperature}$.
\end{proof}

\begin{prop}[Regular quotient for a proper-condensate singularity]\label{expedition0012642}
For the nonregular representation
$\oarepn_{\txtmeanfield,\sminvtemperature,\bar{\smnumberdensity}_{\txtbsn}}$
of Lemma \ref{expedition0019390},
the kernel is
$$\Ker \oarepn_{\txtmeanfield,\sminvtemperature,\bar{\smnumberdensity}_{\txtbsn}}
=
\oaideal{J}_{\txtmeanfield,\txtbec}.$$
The induced representation is faithful on the nonzero-mode
resolvent algebra, and hence
$$\setquot{\oaresolventalgebra
\rbk{\sphilb{D}_{\txtbsn,\sminvtemperature,\smchemicalpotential_{\mathrm{sel}}},\sigma}}
{\oaideal{J}_{\txtmeanfield,\txtbec}}
\eqalgisom
\fun{\oaresolventalgebra}
{\Ker \opform{q}_{\txtbsn,0,\sminvtemperature},
\fnrestr{\sigma}
{\Ker \opform{q}_{\txtbsn,0,\sminvtemperature}}}.$$
\end{prop}

\begin{proof}
Set
$$S
=
\set{f\in\sphilb{D}_{\txtbsn,\sminvtemperature,\smchemicalpotential_{\mathrm{sel}}}}
{\fun{\opform{q}_{\txtbsn,0,\sminvtemperature}}{f}=0}.$$
The zero-mode form is a positive semidefinite form,
so $S=\Ker \opform{q}_{\txtbsn,0,\sminvtemperature}$ is a real symplectic subspace.
The quotient algebra in the statement means the resolvent algebra over this subspace
with the restricted symplectic form,
namely $\fun{\oaresolventalgebra}{S,\fnrestr{\sigma}{S}}$.
The identification of the quotient is not a generator-by-generator calculation.
It is the standard resolvent-algebra consequence of the regular and singular
decomposition of a representation.
By Lemma \ref{expedition0019390}, the regular subspace of
$\oarepn_{\txtmeanfield,\sminvtemperature,\bar{\smnumberdensity}_{\txtbsn}}$
is $S$.
The ideal theorem for nonregular resolvent representations
\cite{DetlevBuchholz001} says that the kernel of such a representation is the
closed ideal generated by the resolvents $\oaresolvent(\lambda,f)$ with
$\lambda\in\fldreal\setminus\setone{0}$ and
$\fun{\opform{q}_{\txtbsn,0,\sminvtemperature}}{f}>0$.
By Definition \ref{expedition0012641}, this closed ideal is
$\oaideal{J}_{\txtmeanfield,\txtbec}$.
The quotient therefore retains precisely the resolvent algebra over $S$:
Thus
$$\setquot{\oaresolventalgebra
\rbk{\sphilb{D}_{\txtbsn,\sminvtemperature,\smchemicalpotential_{\mathrm{sel}}},\sigma}}
{\oaideal{J}_{\txtmeanfield,\txtbec}}
\eqalgisom
\oaresolventalgebra\rbk{S,\fnrestr{\sigma}{S}}.$$

The representation induced by
$\oarepn_{\txtmeanfield,\sminvtemperature,\bar{\smnumberdensity}_{\txtbsn}}$
on this quotient is the GNS representation of the regular nonzero-mode state
from Definition \ref{expedition0012650} restricted to $S$.
By Proposition \ref{expedition0011838},
its restriction to each finite-dimensional local resolvent subalgebra is faithful.
Passing to the inductive-limit closure gives faithfulness on $\fun{\oaresolventalgebra}{S,\sigma}$.
Therefore the induced quotient representation has trivial kernel,
and the kernel of the nonregular representation is exactly $\oaideal{J}_{\txtmeanfield,\txtbec}$.
This quotient sends the generators in Definition \ref{expedition0012641} to
zero.
\end{proof}

The selected-density dynamics is compatible with this quotient only after the quotient has been identified. Indeed, Proposition \ref{expedition0012677} shows that \(\napiernum^{\imunit t \physham[h]_{\bar{\smnumberdensity}_{\txtbsn}}}\) preserves \(\Ker \opform{q}_{\txtbsn,0,\sminvtemperature}\). Hence the automorphism group \eqref{expedition0012675} induces the nonzero-mode quotient automorphism group on \(\fun{\oaresolventalgebra}
{\Ker \opform{q}_{\txtbsn,0,\sminvtemperature},
\fnrestr{\sigma}{\Ker \opform{q}_{\txtbsn,0,\sminvtemperature}}}\).

The quotient above is formed from the GNS representation introduced in Definition \ref{expedition0012650}; the full direct-integral representation still retains the component parameter through its represented center.

\subsubsection{Relation with the Direct-Integral Center}\label{relation-with-the-direct-integral-center}

The condensate ideal, the quotient representation, the component states, and the represented center encode different parts of the same zero-mode data. The theorem records this separation.

\begin{thm}[Relation among the mean-field condensate ideal, quotient, and represented center]\label{expedition0012643}
Let
$\oarepn_{\txtmeanfield,\sminvtemperature,\bar{\smnumberdensity}_{\txtbsn}}$
be the GNS representation of the BEC state \eqref{expedition0012609}
obtained by integrating the component states.
Then the following assertions hold.
\begin{enumerate}
\item
The ideal $\oaideal{J}_{\txtmeanfield,\txtbec}$ is generated by the finite-density
resolvents with
$\fun{\opform{q}_{\txtbsn,0,\sminvtemperature}}{f}>0$.
\item
The quotient by $\oaideal{J}_{\txtmeanfield,\txtbec}$ is the regular
nonzero-mode resolvent algebra of Proposition \ref{expedition0012642}.
\item
The component states \eqref{expedition0018019} are c-number shifted pullbacks
from this quotient, whereas the integrated state \eqref{expedition0012609}
retains the phase parameter as the represented central variable of
Proposition \ref{expedition0012651}.
\item
$\oaideal{J}_{\txtmeanfield,\txtbec}=\setone{0}$ if and only if
$\smnumberdensity_{\txtbsn,0}(\sminvtemperature)=0$.
\end{enumerate}
\end{thm}

\begin{proof}
For $f\in\sphilb{D}_{\txtbsn,\sminvtemperature,\smchemicalpotential_{\mathrm{sel}}}$
with $\fun{\opform{q}_{\txtbsn,0,\sminvtemperature}}{f}>0$,
\eqref{expedition0012440} gives $\faftr{f}(0)\ne0$.
Hence the zero-mode functional
$\ell_{\sminvtemperature,\bar{\smnumberdensity}_{\txtbsn},\theta}(f)$
can be nonzero and depends nontrivially on $\theta$ in the BEC states
$\oastate[\psi_{\txtmeanfield,\sminvtemperature,\bar{\smnumberdensity}_{\txtbsn},\theta}]$.
The corresponding resolvents are precisely the generators in
Definition \ref{expedition0012641}.
Proposition \ref{expedition0012642} identifies their closed ideal as the
kernel of the nonregular representation and identifies the quotient with the
regular nonzero-mode resolvent algebra.

For $f\in\Ker \opform{q}_{\txtbsn,0,\sminvtemperature}$, one has
$\faftr{f}(0)=0$.
Hence the zero-mode functional vanishes for every $\theta$, and
\eqref{expedition0018019} shows that the component states
$\oastate[\psi_{\txtmeanfield,\sminvtemperature,\bar{\smnumberdensity}_{\txtbsn},\theta}]$
restrict to the same regular nonzero-mode KMS state on the quotient.
Generators with
$\fun{\opform{q}_{\txtbsn,0,\sminvtemperature}}{f}>0$ are evaluated through a
$\theta$-dependent c-number shift before the quotient state is applied, while
the quotient representation sends the corresponding resolvents to zero.
After the integral in \eqref{expedition0012609}, Proposition
\ref{expedition0012651} records the remaining phase parameter as a represented
central variable.

If $\smnumberdensity_{\txtbsn,0}(\sminvtemperature)=0$, the zero-mode form in \eqref{expedition0012440} is zero and the generating set in Definition \ref{expedition0012641} is empty,
so the generated ideal is $\setone{0}$.
Conversely, if $\smnumberdensity_{\txtbsn,0}(\sminvtemperature)>0$,
any test function with $\faftr{f}(0)\ne0$ satisfies $\fun{\opform{q}_{\txtbsn,0,\sminvtemperature}}{f}>0$,
so the generated ideal is nonzero.
\end{proof}

\subsection{Mean-Field ODLRO}\label{mean-field-odlro}

The ODLRO statement below is attached to the zero-mode covariance already isolated in \eqref{expedition0012440}. The assertion is made at the level of represented two-point functions in the selected-density construction; it is not a statement about unbounded fields inside the abstract resolvent algebra.

\begin{prop}[Mean-field ODLRO and spatial cluster breakdown]\label{expedition0012913}
Assume $\smnumberdensity_{\txtbsn,0}(\sminvtemperature)>0$.
Then ODLRO is expressed already by the abstract resolvent tests:
the bounded resolvent tests belong to the abstract resolvent algebra, and the
condition is the support condition
$\faftr{f}(0)\faftr{g}(0)\ne0$.
In the selected-density GNS representation this gives the represented
two-point function
$$\fun{\opform{q}_{\txtbsn,0,\sminvtemperature}}{f,\fun{\tau_x}{g}}
=
2(2\pi)^d
\smnumberdensity_{\txtbsn,0}(\sminvtemperature)
\cmpconj{\faftr{f}(0)}
\faftr{g}(0).$$
This non-decaying singular term is the mean-field ODLRO and causes spatial
cluster breakdown for the state \eqref{expedition0012609}.
\end{prop}

\begin{proof}
The zero mode is invariant under spatial translations,
so $\faftr{\fun{\tau_x}{g}}(0)=\faftr{g}(0)$.
Substitution into the zero-mode covariance \eqref{expedition0012440},
with polarization, gives the stated formula.
The condition $\faftr{f}(0)\faftr{g}(0)\ne0$ is exactly the abstract support condition
recorded by the ideal in Definition \ref{expedition0012641}.
The regular nonzero-mode covariance is the mixing part of the selected-density quasi-free state
and decays after the zero-mode Dirac point mass is separated.
The zero-mode contribution therefore remains as the non-decaying part of the truncated two-point function.
\end{proof}

\subsection{Mean-Field BEC from the Viewpoint of Proper Condensates}\label{mean-field-bec-from-the-viewpoint-of-proper-condensates}

Buchholz's criteria distinguish the local condensate line from the stricter infinite-occupation test. For a primary state, the number-resolvent criterion of \cite[Section 2, Definition and Proposition 2.1]{DetlevBuchholz004} asks, for \(f\in\fun{\lp^2}{\mathcal{O}}\), that \begin{equation}\label{expedition0019101}
\fun{\oastate[\psi]}
{\opfnresolvent{\mu
+\fun{\opfockcr_{\txtfock}}{f}
\fun{\opfockan_{\txtfock}}{f}}}
=
0,
\quad
\mu>0.
\end{equation} For a nonprimary state, the corresponding central-decomposition diagnostic is \begin{equation}\label{expedition0019102}
\limsup_{\mu\to\infty}
\mu
\fun{\oastate[\psi]}
{\opfnresolvent{\mu
+\fun{\opfockcr_{\txtfock}}{f}
\fun{\opfockan_{\txtfock}}{f}}}
<
1.
\end{equation}

The ODLRO criterion of \cite[Definitions I--III, Lemma 2.1, Lemma 3.1, Theorem 3.2, and the Criterion in Section 4]{DetlevBuchholz005} uses the local regular subspace \begin{equation}\label{expedition0019103}
\sphilb{H}_{\txtmeanfield,\txtregular}(\mathcal{O})
=
\set{f\in\fun{\lp^2}{\mathcal{O}}}
{\limsup_{\sigma}
\fun{\oastate[\psi_{\sigma}]}
{\fun{\opfockcr_{\txtfock}}{f}
\fun{\opfockan_{\txtfock}}{f}}<\infty}.
\end{equation} In the homogeneous mean-field model, the finite-density ideal identifies the local condensate line, while \eqref{expedition0019101} appears only after the local condensate occupation diverges.

\subsubsection{ODLRO Local-Regular-Space Criterion}\label{odlro-local-regular-space-criterion}

Let \(\mathcal{O}
\subset\fldreal^d\) be a bounded open region and put \begin{equation}\label{expedition0012890}
s_{\mathcal{O}}(x)
=
\absvol{\mathcal{O}}^{-\onehalf}
\fndef{\mathcal{O}}(x).
\end{equation} For \(f
\in \fun{\lp^2}{\mathcal{O}}\), the relevant local quantity is the overlap with this constant vector: \begin{equation}\label{expedition0012891}
\bkt{s_{\mathcal{O}}}{f}_{\fun{\lp^2}{\mathcal{O}}}
=
\absvol{\mathcal{O}}^{-\onehalf}
\int_{\mathcal{O}}
f(x)
\opdmsr{x}.
\end{equation} The homogeneous zero-mode contribution is the rank-one covariance onto \(\fldcmp s_{\mathcal{O}}\).

\begin{prop}[Local regular and singular spaces selected by the mean-field zero mode]\label{expedition0012892}
Assume $\smnumberdensity_{\txtbsn,0}(\sminvtemperature)>0$.
Then the ODLRO local-space criterion of \cite{DetlevBuchholz005}, in the form
\eqref{expedition0019103}, selects the local regular subspace
\begin{equation}\label{expedition0012893}
\sphilb{H}_{\txtmeanfield,\txtregular}(\mathcal{O})
=
\set{f\in\fun{\lp^2}{\mathcal{O}}}{\int_{\mathcal{O}} f(x)\opdmsr{x}=0},
\end{equation}
and its one-dimensional orthogonal complement is
\begin{equation}\label{expedition0012894}
\sphilb{S}_{\txtmeanfield}(\mathcal{O})
=
\sphilb{H}_{\txtmeanfield,\txtregular}(\mathcal{O})^\perp
=
\fldcmp s_{\mathcal{O}} .
\end{equation}
Thus the mean-field BEC ideal \eqref{expedition0019398}
records the homogeneous local condensate line $\fldcmp s_{\mathcal{O}}$.
\end{prop}

\begin{proof}
The homogeneous zero-mode contribution, restricted to the local one-particle
space $\fun{\lp^2}{\mathcal{O}}$, is the rank-one covariance proportional to
$\smnumberdensity_{\txtbsn,0}(\sminvtemperature)
\abs{\bkt{s_{\mathcal{O}}}{f}_{\fun{\lp^2}{\mathcal{O}}}}^2$.
Since the coefficient is strictly positive, \eqref{expedition0012891} gives the
regular subspace as the kernel of
$f\mapsto \int_{\mathcal{O}}f(x)\opdmsr{x}$, and its orthogonal complement is
$\fldcmp s_{\mathcal{O}}$.
\end{proof}

Thus the ideal-theoretic support consists of all local test functions with nonzero overlap with \(s_{\mathcal{O}}\), while the condensate line itself is the orthogonal complement of the regular subspace.

\subsubsection{Primary-State Number-Resolvent Test}\label{primary-state-number-resolvent-test}

The number-resolvent test \eqref{expedition0019101} is applied to the primary component states \(\oastate[\psi_{\txtmeanfield,\sminvtemperature,\bar{\smnumberdensity}_{\txtbsn},\theta}]\) from Definition \ref{expedition0012650}, not to the direct integral itself. At fixed finite local condensate number the test stays positive; it vanishes only in the divergent local-occupation limit.

\begin{prop}[Finite mean-field density is not strict in the local Buchholz number-resolvent sense]\label{expedition0012895}
At fixed finite condensate density
$$0<\smnumberdensity_{\txtbsn,0}(\sminvtemperature)<\infty,$$
for every $\theta\in[0,2\pi)$ the component state
$\oastate[\psi_{\txtmeanfield,\sminvtemperature,\bar{\smnumberdensity}_{\txtbsn},\theta}]$
does not satisfy the strict number-resolvent criterion of
\cite{DetlevBuchholz004} on any fixed bounded region $\mathcal{O}$.
The value of \eqref{expedition0019101} for $s_{\mathcal{O}}$ is positive for
every $\mu>0$.
\end{prop}

\begin{proof}
For fixed $\mathcal{O}$, the local condensate particle number carried by the
constant local mode is
\begin{equation}\label{expedition0012896}
n_C(\mathcal{O})=\smnumberdensity_{\txtbsn,0}(\sminvtemperature)\absvol{\mathcal{O}}.
\end{equation}
Put
$$N_{s_{\mathcal{O}}}
=
\fun{\opfockcr_{\txtfock}}{s_{\mathcal{O}}}
\fun{\opfockan_{\txtfock}}{s_{\mathcal{O}}}.$$
Thus the local number-resolvent
$\invrbk{\mu+N_{s_{\mathcal{O}}}}$ is evaluated in a locally normal
finite-density primary component.
Let $E_{s_{\mathcal{O}}}$ be the spectral measure of
$N_{s_{\mathcal{O}}}$ in the locally normal finite-density primary component.
Then
\begin{equation}\label{expedition0019744}
\fun{\oastate[\psi_{\txtmeanfield,\sminvtemperature,
\bar{\smnumberdensity}_{\txtbsn},\theta}]}
{\invrbk{\mu+N_{s_{\mathcal{O}}}}}
=
\int_{\closedinterval{0}{\infty}}
\frac{1}{\mu+t}
\opdmsr{E_{s_{\mathcal{O}}}(t)}.
\end{equation}
The integrand in \eqref{expedition0019744} is strictly positive and
$\fun{E_{s_{\mathcal{O}}}}{\closedinterval{0}{\infty}}=1$.
Hence \eqref{expedition0019744} is strictly positive for every $\mu>0$, whereas
\eqref{expedition0019101} requires vanishing.
\end{proof}

\begin{thm}[Proper-condensate limit of the mean-field zero-mode component states]\label{expedition0012897}
Let
$\fml{\oastate[\psi_{\txtmeanfield,\sminvtemperature,\bar{\smnumberdensity}_{\txtbsn}^{(\sigma)},\theta_\sigma}]}{\sigma}$
be a family of primary component states from the extremal decomposition in
Proposition \ref{expedition0012651}, indexed so that
$\smnumberdensity_{\txtbsn,0}^{(\sigma)}(\sminvtemperature)
\to \infty$.
For every bounded open region $\mathcal{O}$, set
$$n_C^{(\sigma)}(\mathcal{O})=\smnumberdensity_{\txtbsn,0}^{(\sigma)}(\sminvtemperature)\absvol{\mathcal{O}}.$$
If $n_C^{(\sigma)}(\mathcal{O})\to\infty$, then the local occupation of
$s_{\mathcal{O}}$ satisfies the number-resolvent criterion of
\cite{DetlevBuchholz004}, in the form \eqref{expedition0019101},
and the homogeneous local spaces required by \cite{DetlevBuchholz005} are
$$\sphilb{H}_{\txtmeanfield,\txtregular}(\mathcal{O})
=
\set{f\in\fun{\lp^2}{\mathcal{O}}}{\int_{\mathcal{O}} f(x)\opdmsr{x}=0},
\quad
\sphilb{S}_{\txtmeanfield}(\mathcal{O})=\fldcmp s_{\mathcal{O}}.$$
Equivalently, for every $f\in\fun{\lp^2}{\mathcal{O}}$ with
$\bkt{s_{\mathcal{O}}}{f}_{\fun{\lp^2}{\mathcal{O}}}\ne 0$,
\begin{equation}\label{expedition0012898}
\lim_{\sigma}
\fun{\oastate[\psi_{\txtmeanfield,\sminvtemperature,\bar{\smnumberdensity}_{\txtbsn}^{(\sigma)},\theta_\sigma}]}
{\invrbk{\mu+\fun{\opfockcr_{\txtfock}}{f}
\fun{\opfockan_{\txtfock}}{f}}}
=0,
\quad \mu>0,
\end{equation}
whereas the functions orthogonal to $s_{\mathcal{O}}$ remain regular provided the
nonzero-mode thermal covariance stays locally bounded.
\end{thm}

\begin{proof}
Write
$\oastate[\psi_\sigma]
=
\oastate[\psi_{\txtmeanfield,\sminvtemperature,
\bar{\smnumberdensity}_{\txtbsn}^{(\sigma)},\theta_\sigma}]$ and set
\begin{equation}\label{expedition0019745}
N_f
=
\fun{\opfockcr_{\txtfock}}{f}\fun{\opfockan_{\txtfock}}{f}.
\end{equation}
Let $E_{\sigma,f}$ be the spectral measure of $N_f$ in the component state
$\oastate[\psi_\sigma]$:
\begin{equation}\label{expedition0019746}
\int_{\closedinterval{0}{\infty}}\varphi(t)\opdmsr{E_{\sigma,f}(t)}
=
\fun{\oastate[\psi_\sigma]}{\fun{\varphi}{N_f}}
\end{equation}
for every bounded Borel function $\varphi$.
The homogeneous zero-mode contribution to the local occupation is
\begin{equation}\label{expedition0019747}
m_{\sigma}(f)
=
\smnumberdensity_{\txtbsn,0}^{(\sigma)}(\sminvtemperature)
\absvol{\mathcal{O}}
\abs{\bkt{s_{\mathcal{O}}}{f}_{\fun{\lp^2}{\mathcal{O}}}}^2 .
\end{equation}
Indeed,
\begin{equation}\label{expedition0019748}
\smnumberdensity_{\txtbsn,0}^{(\sigma)}(\sminvtemperature)
\abs{\int_{\mathcal{O}}\cmpconj{f(x)}\opdmsr{x}}^2
=
m_{\sigma}(f).
\end{equation}
If $\bkt{s_{\mathcal{O}}}{f}_{\fun{\lp^2}{\mathcal{O}}}\ne0$, then
$m_{\sigma}(f)\to\infty$.
In the component representation put
\begin{equation}\label{expedition0019749}
\begin{aligned}
\gamma_{\sigma,f}
&=
\smnumberdensity_{\txtbsn,0}^{(\sigma)}(\sminvtemperature)^{\onehalf}
\absvol{\mathcal{O}}^{\onehalf}
\bkt{s_{\mathcal{O}}}{f}_{\fun{\lp^2}{\mathcal{O}}},
\\
\fun{\opfockan_{\sigma,\txtregular}}{f}
&=
\fun{\opfockan_{\txtfock}}{f}
-\gamma_{\sigma,f},
\quad
\fun{\opfockcr_{\sigma,\txtregular}}{f}
=
\fun{\opfockcr_{\txtfock}}{f}
-\cmpconj{\gamma_{\sigma,f}}.
\end{aligned}
\end{equation}
Then
\begin{equation}\label{expedition0019750}
\begin{aligned}
C_{\sigma,f}
&=
\cmpconj{\gamma_{\sigma,f}}
\fun{\opfockan_{\sigma,\txtregular}}{f}
+\gamma_{\sigma,f}
\fun{\opfockcr_{\sigma,\txtregular}}{f},
\\
Y_{\sigma,f}
&=
\fun{\opfockcr_{\sigma,\txtregular}}{f}
\fun{\opfockan_{\sigma,\txtregular}}{f},
\\
N_f
&=
m_{\sigma}(f)+C_{\sigma,f}+Y_{\sigma,f}.
\end{aligned}
\end{equation}
The nonzero-mode local covariance is bounded along the family, hence
\begin{equation}\label{expedition0019751}
\fun{\oastate[\psi_\sigma]}{Y_{\sigma,f}}=O(1),
\quad
\fun{\oastate[\psi_\sigma]}{\abs{C_{\sigma,f}}}
\leq
2m_{\sigma}(f)^{\onehalf}
\fun{\oastate[\psi_\sigma]}{Y_{\sigma,f}}^{\onehalf}
=
\fun{O}{m_{\sigma}(f)^{\onehalf}}.
\end{equation}
Therefore
\begin{equation}\label{expedition0019752}
\int_{\closedinterval{0}{\infty}}
\abs{\frac{t}{m_{\sigma}(f)}-1}
\opdmsr{E_{\sigma,f}(t)}
\leq
\frac{\fun{\oastate[\psi_\sigma]}{\abs{C_{\sigma,f}}}}{m_{\sigma}(f)}
+\frac{\fun{\oastate[\psi_\sigma]}{Y_{\sigma,f}}}{m_{\sigma}(f)}
\to0.
\end{equation}
For fixed $M>0$,
\begin{equation}\label{expedition0019753}
\fun{E_{\sigma,f}}{\closedinterval{0}{M}}
\leq
\frac{1}{1-\frac{M}{m_{\sigma}(f)}}
\int_{\closedinterval{0}{\infty}}
\abs{\frac{t}{m_{\sigma}(f)}-1}
\opdmsr{E_{\sigma,f}(t)}
\to0.
\end{equation}
Thus
\begin{equation}\label{expedition0019754}
\int_{\closedinterval{0}{\infty}}
\frac{1}{\mu+t}
\opdmsr{E_{\sigma,f}(t)}
\leq
\frac{1}{\mu}
\fun{E_{\sigma,f}}{\closedinterval{0}{M}}
+\frac{1}{\mu+M}.
\end{equation}
Taking first $\sigma\to\infty$ and then $M\to\infty$ gives
\eqref{expedition0012898}.

If $\bkt{s_{\mathcal{O}}}{f}_{\fun{\lp^2}{\mathcal{O}}}=0$, the singular
zero-mode contribution \eqref{expedition0019747} vanishes.
The local boundedness of the nonzero-mode covariance then puts $f$ in the
regular subspace.
Consequently the regular subspace is
$\set{f\in\fun{\lp^2}{\mathcal{O}}}{\int_{\mathcal{O}} f\opdmsr{x}=0}$, and
the singular local condensate space is $\fldcmp s_{\mathcal{O}}$.
\end{proof}

Thus Proposition \ref{expedition0012892} identifies the local line \(\fldcmp s_{\mathcal{O}}\), Proposition \ref{expedition0012895} excludes the strict number-resolvent criterion at finite local occupation, and Theorem \ref{expedition0012897} recovers that criterion in the infinite-occupation limit.

\subsection{Mean-Field Effective Hamiltonian and Density Fluctuations}\label{mean-field-effective-hamiltonian-and-density-fluctuations}

The finite-volume algebraic identity below is the basis for the mean-field effective Hamiltonian comparison. It isolates the density-fluctuation remainder which is not removed by simply replacing the density with the Kac-selected scalar value.

\begin{prop}[Effective Hamiltonian and density-fluctuation remainder]\label{expedition0012654}
For a density value $r\geq0$, put
$$\physham_{\txtbsn,\txteff,r,L}
=
\fun{\opfocksndqntdiff_{\txtbsn}}
{\physham[h]_{\txtparticle,0,L}-\smchemicalpotential+\lambda r}.$$
Then it holds that
\begin{equation}\label{expedition0012655}
\physham_{\txtbsn,\txtmeanfield,\smchemicalpotential,L}
-\physham_{\txtbsn,\txteff,r,L}
=
\frac{\lambda V_L}{2}
\rbk{\frac{\opfocknumber_{\txtbsn,L}}{V_L}-r}^{2}
-\frac{\lambda V_L}{2}
r^{2}.
\end{equation}
\end{prop}

This is a finite-volume identity at the finite-volume parameter \(\smchemicalpotential\). The selected-density mean-field automorphism group \eqref{expedition0012675} is the separate post-Kac object with \(\smchemicalpotential_{\mathrm{sel}}\), introduced after the Euler equations. Thus \eqref{expedition0012655} records the density-fluctuation remainder in the finite-volume comparison; it does not identify the finite-volume Hamiltonian with the selected-density mean-field automorphism group on fluctuation or other nonlocal observables.

\begin{proof}
By \eqref{expedition0012600} we obtain
$$\physham_{\txtbsn,\txtmeanfield,\smchemicalpotential,L}
=
\physham_{\txtbsn,\txtfr,0,L}
-\smchemicalpotential\opfocknumber_{\txtbsn,L}
+\frac{\lambda}{2V_L}\opfocknumber_{\txtbsn,L}^{2}.$$
The definition of $\physham_{\txtbsn,\txteff,r,L}$ gives
$$\physham_{\txtbsn,\txteff,r,L}
=
\physham_{\txtbsn,\txtfr,0,L}
-\smchemicalpotential\opfocknumber_{\txtbsn,L}
+\lambda r \opfocknumber_{\txtbsn,L}.$$
Subtracting these two identities and completing the square yields
$$\frac{\lambda}{2V_L}\opfocknumber_{\txtbsn,L}^{2}
-\lambda r \opfocknumber_{\txtbsn,L}
=
\frac{\lambda V_L}{2}
\rbk{\frac{\opfocknumber_{\txtbsn,L}}{V_L}-r}^{2}
-
\frac{\lambda V_L}{2}r^{2},$$
which proves \eqref{expedition0012655}.

Theorem \ref{expedition0012647} selects
$r=\bar{\smnumberdensity}_{\txtbsn}$
in the equilibrium density law.
With this insertion the finite-volume comparison has one-particle
energy
$\physham[h]_{\txtparticle,0,L}-\smchemicalpotential
+\lambda\bar{\smnumberdensity}_{\txtbsn}$.
The selected-density local resolvent automorphism group is instead the
automorphism group
with
energy
$\physham[h]_{\txtparticle,0}-\smchemicalpotential_{\mathrm{sel}}
+\lambda\bar{\smnumberdensity}_{\txtbsn}$ from \eqref{expedition0012675}.
The selected-density mean-field automorphism group is introduced in
\eqref{expedition0012675}, and Remark \ref{expedition0018016} records its
separation from the finite-volume comparison.
However, the first term on the right-hand side of \eqref{expedition0012655} is the square of the total density fluctuation.
It is invisible in the local Kac-measure collapse argument only for fixed local observables, and it is precisely the term tested by density fluctuations and other nonlocal observables.
Therefore the proposition gives a finite-volume comparison with an
explicit fluctuation remainder, not a fluctuation-dynamics equivalence.
\end{proof}

This proposition is the resolvent-algebra translation of the limitation stated in \cite[Eq. (5.7) and Section 5.1]{AndreVerbeure001}. In \cite[Section 5.1, after Eq. (5.7)]{AndreVerbeure001}, the effective Hamiltonian is said to exclude all effects of the total density fluctuations. The same passage says that the two Hamiltonians may show serious differences for fluctuation dynamics and nonlocal quantities.

\begin{prop}[Trivial dynamics of the total-density fluctuation]\label{expedition0012656}
Let
$$\mathsf{D}_{\bar{\smnumberdensity}_{\txtbsn},L}
=
V_L^{\onehalf}
\rbk{\frac{\opfocknumber_{\txtbsn,L}}{V_L}-\bar{\smnumberdensity}_{\txtbsn}}.$$
For every $z\in\fldcmp\setminus\fldreal$ and every $t\in\fldreal$,
\begin{equation}\label{expedition0012657}
\napiernum^{\imunit t\physham_{\txtbsn,\txtmeanfield,\smchemicalpotential,L}}
\invrbk{z-\mathsf{D}_{\bar{\smnumberdensity}_{\txtbsn},L}}
\napiernum^{-\imunit t\physham_{\txtbsn,\txtmeanfield,\smchemicalpotential,L}}
=
\invrbk{z-\mathsf{D}_{\bar{\smnumberdensity}_{\txtbsn},L}}.
\end{equation}
Consequently, any weak, strong, or distributional limit of the resolvents
$\invrbk{z-\mathsf{D}_{\bar{\smnumberdensity}_{\txtbsn},L}}$
inherits the identity dynamics.
\end{prop}

\begin{proof}
The free Hamiltonian
$\physham_{\txtbsn,\txtfr,0,L}$
is the second quantization of a one-particle operator, hence it commutes with the total number operator
$\opfocknumber_{\txtbsn,L}$.
The chemical-potential term is also a function of
$\opfocknumber_{\txtbsn,L}$.
The mean-field part in \eqref{expedition0012600} is the function
$\lambda \opfocknumber_{\txtbsn,L}^{2}/(2V_L)$
of the same number operator.
Therefore
$\commutator{\physham_{\txtbsn,\txtmeanfield,\smchemicalpotential,L}}
{\opfocknumber_{\txtbsn,L}}=0.$
Since
$\mathsf{D}_{\bar{\smnumberdensity}_{\txtbsn},L}$
is an affine function of
$\opfocknumber_{\txtbsn,L}$,
the spectral calculus gives
$$\commutator{\physham_{\txtbsn,\txtmeanfield,\smchemicalpotential,L}}
{\invrbk{z-\mathsf{D}_{\bar{\smnumberdensity}_{\txtbsn},L}}}
= 0$$
for every non-real $z$.
This is exactly \eqref{expedition0012657}.
Passing to any topology in which these bounded resolvents converge preserves the equality, so the induced limiting action on the total-density fluctuation resolvents is the identity.
\end{proof}

Thus the total-density fluctuation itself has no nontrivial reversible motion in this model. What Verbeure's warning excludes is stronger: the effective Hamiltonian cannot be used as if it controlled all fluctuation and nonlocal observables.

\section{Occupation Number Probabilistic Approach}\label{expedition0019615}

The occupation-number probability space gives a second formulation of the mean-field model. The density selection, Euler equations, condensate excess, component decomposition, local condensate tests, proper-condensate occupation escape, and density-fluctuation remainder can be stated directly in terms of random occupation variables. The construction keeps the argument on this probability space.

\subsection{Occupation-Number Probability Measures}\label{occupation-number-probability-measures}

The finite-volume Gibbs state \eqref{expedition0012601} induces an ordinary probability measure on occupation-number configurations. Recalling that \(\mathcal{O}_L\) is defined by \eqref{expedition0019646}, then let \begin{equation}\label{expedition0019667}
\begin{aligned}
N_L(\mathsf{n})
=
\sum_{k\in\setlattice_L^d}\mathsf{n}_k,
\quad
R_L(\mathsf{n})
=
\frac{N_L(\mathsf{n})}{V_L},
\quad
\mathsf{n}
\in \mathcal{O}_L.
\end{aligned}
\end{equation} Here \(\setone{\mathsf{n}}\) denotes the singleton subset of \(\mathcal{O}_L\) containing the occupation configuration \(\mathsf{n}\). On the set \(\mathcal{O}_L\) of occupation configurations, define \begin{equation}\label{expedition0019668}
\fun{\msrbb{P}_{\txtmeanfield,L}}{\setone{\mathsf{n}}}
=
\frac{1}{\smpartitionfunc_{\sminvtemperature,\lambda,\smchemicalpotential,L}}
\fnexp{-\sminvtemperature
\rbk{\sum_{k\in\setlattice_L^d}
\rbk{\physham[h]_{\txtparticle,0,L}(k)-\smchemicalpotential}
\mathsf{n}_k
+\frac{\lambda}{2V_L}N_L(\mathsf{n})^2}}.
\end{equation} The free occupation-number measure is \begin{equation}\label{expedition0019669}
\fun{\msrbb{P}_{\txtfr,L}}{\setone{\mathsf{n}}}
=
\frac{1}{\smpartitionfunc_{\sminvtemperature,0,\smchemicalpotential,L}}
\fnexp{-\sminvtemperature
\sum_{k\in\setlattice_L^d}
\rbk{\physham[h]_{\txtparticle,0,L}(k)-\smchemicalpotential}
\mathsf{n}_k}.
\end{equation} Then the law of \(R_L\) in \eqref{expedition0019667} under the measure \(\msrbb{P}_{\txtmeanfield,L}\) defined by \eqref{expedition0019668} is exactly \(K_{\sminvtemperature,\lambda,\smchemicalpotential,L}\) in \eqref{expedition0012728}. Equivalently, Proposition \ref{expedition0012680} says that this law is the free density law \eqref{expedition0018007} tilted by the quadratic weight \eqref{expedition0012603}.

\begin{prop}[Probabilistic form of the finite-volume mean-field measure]
For every bounded Borel function $F$ on $\fldreal_{\geq0}$,
$$\sqfun{\prbexp_{\msrbb{P}_{\txtmeanfield,L}}}{F(R_L)}
=
\int_{\fldreal_{\geq0}}
F(r)
\opdmsr{K_{\sminvtemperature,\lambda,\smchemicalpotential,L}(r)}.$$
The same expectation is obtained from the free law by
$$\sqfun{\prbexp_{\msrbb{P}_{\txtmeanfield,L}}}{F(R_L)}
=
\frac{
\int_{\fldreal_{\geq0}}
F(r)
\fnexp{-\frac{\sminvtemperature\lambda V_L r^2}{2}}
\opdmsr{K_{\sminvtemperature,0,\smchemicalpotential,L}(r)}
}{
\int_{\fldreal_{\geq0}}
\fnexp{-\frac{\sminvtemperature\lambda V_L s^2}{2}}
\opdmsr{K_{\sminvtemperature,0,\smchemicalpotential,L}(s)}
}.$$
\end{prop}

This is the probability formulation of the finite-volume occupation-basis diagonalization used in the proofs of Propositions \ref{expedition0012679} and \ref{expedition0012680}.

\begin{proof}
The random variable $R_L$ takes the value
$V_L^{-1}
\sum_{k\in\setlattice_L^d}
\mathsf{n}_k$ on the configuration
$\mathsf{n}$.
Hence
$$\begin{aligned}
&\sqfun{\prbexp_{\msrbb{P}_{\txtmeanfield,L}}}{F(R_L)}
=
\sum_{\mathsf{n}\in\mathcal{O}_L}
\fun{F}
{\frac{1}{V_L}\sum_{k\in\setlattice_L^d}\mathsf{n}_k}
\fun{\msrbb{P}_{\txtmeanfield,L}}{\setone{\mathsf{n}}}
\\ 
&=
\frac{1}
{\smpartitionfunc_{\sminvtemperature,\lambda,\smchemicalpotential,L}}
\sum_{\mathsf{n}\in\mathcal{O}_L}
\fun{F}{\frac{1}{V_L}\sum_{k\in\setlattice_L^d}\mathsf{n}_k}
\fnexp{-\sminvtemperature
\rbk{\sum_{k\in\setlattice_L^d}
\rbk{\physham[h]_{\txtparticle,0,L}(k)-\smchemicalpotential}\mathsf{n}_k
+\frac{\lambda}{2V_L}
\rbk{\sum_{k\in\setlattice_L^d}
\mathsf{n}_k}^{2}}}.
\end{aligned}$$
This is the integral of $F$ against the point-mass sum in
\eqref{expedition0012728}, so it is
$\int F(r)\opdmsr{K_{\sminvtemperature,\lambda,\smchemicalpotential,L}(r)}$.

For the second identity, write
$$Z_{\txtfr,L}
=
\smpartitionfunc_{\sminvtemperature,0,\smchemicalpotential,L},
\quad
Z_{\txtmeanfield,L}
=
\smpartitionfunc_{\sminvtemperature,\lambda,\smchemicalpotential,L}.$$
Using \eqref{expedition0018007}, it follows that
$$\begin{aligned}
&\int_{\fldreal_{\geq0}}
F(r)
\fnexp{-\frac{\sminvtemperature\lambda V_L r^2}{2}}
\opdmsr{K_{\sminvtemperature,0,\smchemicalpotential,L}(r)}
\\ 
&=
\frac{1}{Z_{\txtfr,L}}
\sum_{\mathsf{n}\in\mathcal{O}_L}
F(R_L(\mathsf{n}))
\fnexp{-\frac{\sminvtemperature\lambda V_L R_L(\mathsf{n})^2}{2}}
\fnexp{-\sminvtemperature
\sum_{k\in\setlattice_L^d}
\rbk{\physham[h]_{\txtparticle,0,L}(k)-\smchemicalpotential}\mathsf{n}_k}
\\ 
&=
\frac{Z_{\txtmeanfield,L}}{Z_{\txtfr,L}}
\sqfun{\prbexp_{\msrbb{P}_{\txtmeanfield,L}}}{F(R_L)}.
\end{aligned}$$
The same calculation with $F=1$ gives
$$
\int_{\fldreal_{\geq0}}
\fnexp{-\frac{\sminvtemperature\lambda V_L s^2}{2}}
\opdmsr{K_{\sminvtemperature,0,\smchemicalpotential,L}(s)}
=
\frac{Z_{\txtmeanfield,L}}{Z_{\txtfr,L}}.
$$
Dividing the two identities gives the claimed tilted expectation formula.
\end{proof}

\subsection{Density Selection by the Tilted Large-Deviation Principle}\label{density-selection-by-the-tilted-large-deviation-principle}

The scalar density selected in the thermodynamic limit is obtained by a standard exponential tilting of the free density LDP. The free input is Proposition \ref{expedition0012645}; the mean-field tilt is the finite-volume identity \eqref{expedition0012603}.

\begin{prop}[Probabilistic Kac collapse]\label{expedition0019755}
The random variables $R_L$ in \eqref{expedition0019667} converge in law under
$\msrbb{P}_{\txtmeanfield,L}$ defined by \eqref{expedition0019668} to the constant
$\bar{\smnumberdensity}_{\txtbsn}$ selected in Theorem
\ref{expedition0012647}.
Equivalently, for every $\epsilon>0$,
$$\lim_{L\to\infty}
\fun{\msrbb{P}_{\txtmeanfield,L}}
{\setone{\abs{R_L-\bar{\smnumberdensity}_{\txtbsn}}>\epsilon}}
=0.$$
\end{prop}

The rate function is \(I_{\txtmeanfield,\sminvtemperature,\smchemicalpotential,\lambda}\) from Theorem \ref{expedition0012647}. Its unique minimizer is the selected density.

\begin{proof}
By Proposition \ref{expedition0012645}, the free laws
$K_{\sminvtemperature,0,\smchemicalpotential,L}$ satisfy an LDP with good rate
function
$I_{\txtfr,\sminvtemperature,\smchemicalpotential}$.
Set
$$
\Phi(r)=\frac{\sminvtemperature\lambda r^2}{2},
\quad
J(r)=I_{\txtfr,\sminvtemperature,\smchemicalpotential}(r)+\Phi(r),
\quad
m=\inf_{s\geq0}J(s).
$$
By the finite-volume tilted density identity \eqref{expedition0012603},
$$
K_{\sminvtemperature,\lambda,\smchemicalpotential,L}(A)
=
\frac{
\int_A \fnexp{-V_L\Phi(r)}
\opdmsr{K_{\sminvtemperature,0,\smchemicalpotential,L}(r)}
}{
\int_{\fldreal_{\geq0}} \fnexp{-V_L\Phi(s)}
\opdmsr{K_{\sminvtemperature,0,\smchemicalpotential,L}(s)}
}
$$
for every Borel set $A$.
For a closed set $F$, the LDP upper bound and the usual partition of $F$ into
small intervals give
$$
\limsup_{L\to\infty}
\frac{1}{V_L}
\log
\int_F \fnexp{-V_L\Phi(r)}
\opdmsr{K_{\sminvtemperature,0,\smchemicalpotential,L}(r)}
\leq
-\inf_{r\in F}J(r).
$$
For the denominator, if $G$ is any open neighbourhood of a point $r_*$ with
$J(r_*)<m+\eta$, then the LDP lower bound gives
$$\begin{aligned}
\liminf_{L\to\infty}
\frac{1}{V_L}
\log
\int_{\fldreal_{\geq0}} \fnexp{-V_L\Phi(s)}
\opdmsr{K_{\sminvtemperature,0,\smchemicalpotential,L}(s)}
&\geq
\liminf_{L\to\infty}
\frac{1}{V_L}
\log
\int_G \fnexp{-V_L\Phi(s)}
\opdmsr{K_{\sminvtemperature,0,\smchemicalpotential,L}(s)}
\\ 
&\geq
-\inf_{s\in G}J(s)
\geq
-m-\eta.
\end{aligned}$$
Letting $\eta\downarrow0$ yields the denominator exponent $-m$.
Therefore the tilted laws satisfy the closed-set estimate
$$
\limsup_{L\to\infty}
\frac{1}{V_L}
\log
K_{\sminvtemperature,\lambda,\smchemicalpotential,L}(F)
\leq
-\inf_{r\in F}\rbk{J(r)-m}.
$$
The corresponding open-set lower estimate is obtained by applying the same
calculation to an open set in the numerator and the upper estimate to the
denominator.
Thus the mean-field rate function is $J-m$, which is the rate function in
Theorem \ref{expedition0012647}.

Since $J-m$ has the unique minimizer
$\bar{\smnumberdensity}_{\txtbsn}$, for each $\epsilon>0$ the closed set
$$
F_\epsilon
=
\set{r\in\fldreal_{\geq0}}
{\abs{r-\bar{\smnumberdensity}_{\txtbsn}}\geq\epsilon}
$$
satisfies
$c_\epsilon=\inf_{r\in F_\epsilon}\rbk{J(r)-m}>0$.
Hence
$$
\limsup_{L\to\infty}
\frac{1}{V_L}
\log
\fun{\msrbb{P}_{\txtmeanfield,L}}
{\setone{\abs{R_L-\bar{\smnumberdensity}_{\txtbsn}}\geq\epsilon}}
\leq
-c_\epsilon.
$$
This gives convergence in probability, and convergence in law to the point
mass $\delta_{\bar{\smnumberdensity}_{\txtbsn}}$ follows.
\end{proof}

\subsection{Occupation Statistics Under Density Conditioning and Euler Equations}\label{occupation-statistics-under-density-conditioning-and-euler-equations}

After the total density has been selected, the nonzero modes are governed by the entropy and energy balance of the occupation-number variables after conditioning on the value of \(R_L\). This is the probabilistic form of the Euler equations in \eqref{expedition0012604}--\eqref{expedition0012607}. The finite-volume interaction depends only on \(R_L\) in \eqref{expedition0019667}, so conditioning on \(R_L\) leaves the remaining variational problem as a free Bose variational problem at the selected density.

\begin{prop}[Probabilistic form of the mean-field Euler equations]
Let $R_L$ be the random variable defined in \eqref{expedition0019667} on
the probability space
$\rbk{\mathcal{O}_L,\msrbb{P}_{\txtmeanfield,L}}$ with
$\msrbb{P}_{\txtmeanfield,L}$ defined by \eqref{expedition0019668}, and let
$\bar{\smnumberdensity}_{\txtbsn}$ be the limit in Theorem
\ref{expedition0012647}.
Every thermodynamic limit of the occupation-number laws under conditioning on
$R_L$
satisfies the complementarity equation \eqref{expedition0012604} and the
nonzero-mode occupation formula \eqref{expedition0012605}.
Consequently the condensed and normal cases are exactly
\eqref{expedition0012606} and \eqref{expedition0012607}.
\end{prop}

It says that the limiting occupation-number law is determined by the minimizer of the density-constrained free-energy functional in Proposition \ref{expedition0012649}.

\begin{proof}
For $\mathsf{n}\in\mathcal{O}_L$, \eqref{expedition0019668} gives
\begin{equation}\label{expedition0019740}
\fun{\msrbb{P}_{\txtmeanfield,L}}{\setone{\mathsf{n}}}
\propto
\fnexp{-\sminvtemperature
\rbk{
\sum_{k\in\setlattice_L^d}
\rbk{\physham[h]_{\txtparticle,0,L}(k)-\smchemicalpotential}
\mathsf{n}_k
+
\frac{\lambda V_L}{2}R_L(\mathsf{n})^2}}.
\end{equation}
Conditioning on $R_L=r$ turns the last term in \eqref{expedition0019740} into
the constant $\lambda V_L r^2/2$.
Thus the remaining large-volume variational problem is the density-constrained
free Bose variational problem at total density $r$.
By the closed-set LDP estimate in the proof of Proposition
\ref{expedition0019755},
\begin{equation}\label{expedition0019741}
R_L\to \bar{\smnumberdensity}_{\txtbsn}
\end{equation}
in probability under $\msrbb{P}_{\txtmeanfield,L}$.
Inserting $r=\bar{\smnumberdensity}_{\txtbsn}$, the density-constrained
variational problem is exactly the one treated in the proof of Proposition
\ref{expedition0012649}.
Repeating its first-variation calculation, namely
\eqref{expedition0012729} for the nonzero modes and
\eqref{expedition0018044} for the zero-mode half-line, gives
\eqref{expedition0012604} and \eqref{expedition0012605}.
The alternatives \eqref{expedition0012606} and \eqref{expedition0012607} are
then obtained by the same split according to
$\smnumberdensity_{\txtbsn,0}(\sminvtemperature)>0$ or
$\smnumberdensity_{\txtbsn,0}(\sminvtemperature)=0$.
\end{proof}

\subsection{Zero-Mode Occupation as a Random Variable}\label{zero-mode-occupation-as-a-random-variable}

The zero-mode density is a random occupation variable before it becomes a covariance contribution in the resolvent-algebra formulation. Put \[R_{0,L}(\mathsf{n})
=
\frac{\mathsf{n}_0}{V_L},
\quad
R_{\neq0,L}(\mathsf{n})
=
\frac{1}{V_L}
\sum_{k\in\setlattice_L^d\setminus\setone{0}}\mathsf{n}_k,
\quad
\mathsf{n}
\in
\mathcal{O}_L.\] Then \(R_L=R_{0,L}+R_{\neq0,L}\).

\begin{prop}[Probabilistic zero-mode excess]\label{expedition0019803}
In the condensed case \eqref{expedition0012606},
$$
R_{0,L}
\to
\bar{\smnumberdensity}_{\txtbsn}
-\smnumberdensity_{\txtbsn,\txtcritical}(\sminvtemperature)
=
\smnumberdensity_{\txtbsn,0}(\sminvtemperature)
$$
in probability under $\msrbb{P}_{\txtmeanfield,L}$ defined by
\eqref{expedition0019668}.
In the normal case \eqref{expedition0012607}, $R_{0,L}\to0$ in probability.
\end{prop}

Thus the probability statement corresponding to the zero-mode input is the concentration of the random variable \(R_{0,L}\). The zero-mode covariance form \eqref{expedition0012440} is the subsequent field-theoretic encoding of the same scalar excess density. Thus BEC occurrence is represented here by the positive limit of the random zero-mode density. The normal case gives the vanishing zero-mode density and hence no BEC input for the later ODLRO and local-condensate statements.

\begin{proof}
The Euler equations \eqref{expedition0012604}--\eqref{expedition0012607}
identify the limiting nonzero-mode occupation density as
$n_{\bar{\smnumberdensity}_{\txtbsn}}$ in \eqref{expedition0012605}.
In the condensed case \eqref{expedition0012606}, this equation reduces
the density to the
free saturated occupation
$\rbk{\napiernum^{\sminvtemperature\physham[h]_{\txtparticle,0}(k)}-1}^{-1}$
for $k\ne0$.
Its integral is
$\smnumberdensity_{\txtbsn,\txtcritical}(\sminvtemperature)$ by \eqref{expedition0019397}.
The nonzero-mode density is the Riemann-sum limit
$$
\lim_{L\to\infty}
\frac{1}{V_L}
\sum_{k\in\setlattice_L^d\setminus\setone{0}}
\frac{1}{\napiernum^{\sminvtemperature\physham[h]_{\txtparticle,0,L}(k)}-1}
=
\smnumberdensity_{\txtbsn,\txtcritical}(\sminvtemperature).
$$
The corresponding variance is of order
$$
\frac{1}{V_L^2}
\sum_{k\in\setlattice_L^d\setminus\setone{0}}
n_{\bar{\smnumberdensity}_{\txtbsn},L}(k)
\rbk{1+n_{\bar{\smnumberdensity}_{\txtbsn},L}(k)},
$$
which tends to zero under the same local integrability condition used in the
free critical-density formula.
Thus
$R_{\neq0,L}\to\smnumberdensity_{\txtbsn,\txtcritical}(\sminvtemperature)$ in probability.
Since
$$
R_{0,L}
=
R_L-R_{\neq0,L},
$$
and $R_L\to\bar{\smnumberdensity}_{\txtbsn}$ in probability, the difference
converges in probability to
$\bar{\smnumberdensity}_{\txtbsn}-\smnumberdensity_{\txtbsn,\txtcritical}(\sminvtemperature)$.
In the normal case \eqref{expedition0012607}, the total density is already carried by the nonzero-mode
occupation in \eqref{expedition0012607}.
The same subtraction gives
$$
R_{0,L}
=
R_L-R_{\neq0,L}
\to
\bar{\smnumberdensity}_{\txtbsn}-\bar{\smnumberdensity}_{\txtbsn}
=0
$$
in probability.
\end{proof}

\subsection{Direct-Integral Components and the Integrated Equilibrium Law}\label{direct-integral-components-and-the-integrated-equilibrium-law}

The probabilistic direct integral is a disintegration of the zero-mode amplitude law over the circle parameter. Let \(\Theta\) be uniformly distributed on \([0,2\pi)\). Denote its law by \(\msrbb{P}_{\Theta}\). In the condensed case \eqref{expedition0012606}, set \[\begin{aligned}
z_{\txtbec}(\theta)
=
\sqrt{\smnumberdensity_{\txtbsn,0}(\sminvtemperature)}
\napiernum^{\imunit\theta},
\quad
\fun{\msrbb{P}_{\txtbec,\theta}}{B}
=
\fun{\diracdelta_{z_{\txtbec}(\theta)}}{B},
\quad
B \subset \mblfmlborel(\fldcmp).
\end{aligned}\] The probability kernel \(\theta
\mapsto
\msrbb{P}_{\txtbec,\theta}\) is the component measure. Define \[\Omega_{\txtbec}
=
[0,2\pi)\times\fldcmp,
\quad
\fun{\Theta_{\txtbec}}{\theta,z}
=
\theta,
\quad
\fun{Z_{\txtbec}}{\theta,z}
=
z.\] The joint probability measure \(\msrbb{P}_{\txtbec}\) on \(\Omega_{\txtbec}\) is defined, for every Borel set \(C\subset\Omega_{\txtbec}\), by \[\fun{\msrbb{P}_{\txtbec}}{C}
=
\frac{1}{2\pi}
\int_0^{2\pi}
\fun{\msrbb{P}_{\txtbec,\theta}}
{\set{z\in\fldcmp}{(\theta,z)\in C}}
\opdmsr{\theta}.\] The regular conditional measure of \(Z_{\txtbec}\) given \(\Theta_{\txtbec}=\theta\) is \[\funcond{\msrbb{P}_{\txtbec}}
{\setone{Z_{\txtbec}\in B}}
{\Theta_{\txtbec}=\theta}
=
\fun{\msrbb{P}_{\txtbec,\theta}}{B}.\] The integrated zero-mode amplitude law is the marginal probability measure \(\msrbb{P}_{\txtbec}\) on \(\fldcmp\): \[\fun{\msrbb{P}_{\txtbec}}{B}
=
\fun{\msrbb{P}_{\txtbec}}{\setone{Z_{\txtbec}\in B}}
=
\frac{1}{2\pi}
\int_0^{2\pi}
\fun{\msrbb{P}_{\txtbec,\theta}}{B}
\opdmsr{\theta}.\] Equivalently, for every bounded Borel function \(H\) on \(\fldcmp\), \[\begin{aligned}
\int_{\fldcmp}H(z)\opdmsr{\msrbb{P}_{\txtbec}(z)}
=
\frac{1}{2\pi}
\int_0^{2\pi}
\int_{\fldcmp}H(z)
\opdmsr{\msrbb{P}_{\txtbec,\theta}(z)}
=
\sqfun{\prbexp_{\msrbb{P}_{\txtbec}}}
{H(Z_{\txtbec})}.
\end{aligned}\] This is the probability-space construction corresponding to the direct-integral parameter: the component is the Dirac law \(\msrbb{P}_{\txtbec,\theta}\), and the integrated law is their mixture with respect to the marginal \(\msrbb{P}_{\Theta}\) of \(\Theta_{\txtbec}\). The parameter value \(\theta\) indexes the component law \(\msrbb{P}_{\txtbec,\theta}\), and averaging with respect to \(\msrbb{P}_{\Theta}\) gives the gauge-invariant law \(\msrbb{P}_{\txtbec}\).

\begin{prop}[Probabilistic direct-integral component mixture]
Assume $\smnumberdensity_{\txtbsn,0}(\sminvtemperature)>0$.
For the random amplitude $Z_{\txtbec}$ just defined,
$$\sqfun{\prbexp_{\msrbb{P}_{\txtbec}}}{Z_{\txtbec}}=0,
\quad
\sqfun{\prbexp_{\msrbb{P}_{\txtbec}}}{\abs{Z_{\txtbec}}^2}
=
\smnumberdensity_{\txtbsn,0}(\sminvtemperature).$$
For test functions whose zero-momentum values are defined, set
$$\Xi_f
=
\sqrt{2(2\pi)^d}Z_{\txtbec}\faftr{f}(0),
\quad
\Xi_g
=
\sqrt{2(2\pi)^d}Z_{\txtbec}\faftr{g}(0).$$
Then the zero-mode two-point covariance is
$$\begin{aligned}
\fun{\prbcov_{\msrbb{P}_{\txtbec}}}{\Xi_f,\Xi_g}
=
\sqfun{\prbexp_{\msrbb{P}_{\txtbec}}}
{\cmpconj{\rbk{\Xi_f-\sqfun{\prbexp_{\msrbb{P}_{\txtbec}}}{\Xi_f}}}
\rbk{\Xi_g-\sqfun{\prbexp_{\msrbb{P}_{\txtbec}}}{\Xi_g}}}
=
2(2\pi)^d
\smnumberdensity_{\txtbsn,0}(\sminvtemperature)
\cmpconj{\faftr{f}(0)}
\faftr{g}(0).
\end{aligned}$$
\end{prop}

The first identity expresses the disappearance of the one-point amplitude after integration over the direct-integral parameter. The second and third identities show that the second-order zero-mode contribution survives the integration. For BEC, this is the probabilistic counterpart of the integrated BEC state: the first moment vanishes after integration, but the positive second moment retains exactly the zero-mode excess density.

\begin{proof}
The definition of $\msrbb{P}_{\txtbec}$ gives the explicit integrals
$$\sqfun{\prbexp_{\msrbb{P}_{\txtbec}}}
{\frac{Z_{\txtbec}}
{\sqrt{\smnumberdensity_{\txtbsn,0}(\sminvtemperature)}}}
=
\frac{1}{2\pi}
\int_0^{2\pi}
\napiernum^{\imunit\theta}
\opdmsr{\theta}
=0,
\quad
\sqfun{\prbexp_{\msrbb{P}_{\txtbec}}}
{\frac{\abs{Z_{\txtbec}}^2}
{\smnumberdensity_{\txtbsn,0}(\sminvtemperature)}}
=
\frac{1}{2\pi}
\int_0^{2\pi}
\napiernum^{-\imunit\theta}
\napiernum^{\imunit\theta}
\opdmsr{\theta}
=1.$$
Therefore
$$\sqfun{\prbexp_{\msrbb{P}_{\txtbec}}}{Z_{\txtbec}}
=0,
\quad
\sqfun{\prbexp_{\msrbb{P}_{\txtbec}}}
{\abs{Z_{\txtbec}}^2}
=
\smnumberdensity_{\txtbsn,0}(\sminvtemperature).$$
By the definitions of $\Xi_f$ and $\Xi_g$,
$$\sqfun{\prbexp_{\msrbb{P}_{\txtbec}}}{\Xi_f}
=
\sqrt{2(2\pi)^d}\faftr{f}(0)
\sqfun{\prbexp_{\msrbb{P}_{\txtbec}}}{Z_{\txtbec}}
=0,
\quad
\sqfun{\prbexp_{\msrbb{P}_{\txtbec}}}{\Xi_g}
=0.$$
Therefore the covariance equals the mixed second moment:
$$\begin{aligned}
\fun{\prbcov_{\msrbb{P}_{\txtbec}}}{\Xi_f,\Xi_g}
=
\sqfun{\prbexp_{\msrbb{P}_{\txtbec}}}{\cmpconj{\Xi_f}\Xi_g}
=
2(2\pi)^d
\cmpconj{\faftr{f}(0)}\faftr{g}(0)
\sqfun{\prbexp_{\msrbb{P}_{\txtbec}}}
{\abs{Z_{\txtbec}}^2}
=
2(2\pi)^d
\smnumberdensity_{\txtbsn,0}(\sminvtemperature)
\cmpconj{\faftr{f}(0)}\faftr{g}(0).
\end{aligned}$$
The normalization of the realified field convention gives the factor
$2(2\pi)^d$, matching \eqref{expedition0012440}.
\end{proof}

\subsection{ODLRO and Local Condensate Tests}\label{odlro-and-local-condensate-tests}

The ODLRO statement is the persistence of the rank-one zero-mode covariance under translations. For local tests, the same probability statement is expressed through the occupation of the normalized constant vector on a bounded region \(\mathcal{O}\), defined in \eqref{expedition0012890}.

\begin{prop}[Probabilistic ODLRO and local tests]
In the condensed case \eqref{expedition0012606}, for every translation
$x$ set
$$\Xi_f
=
\sqrt{2(2\pi)^d}Z_{\txtbec}\faftr{f}(0),
\quad
\Xi_{x,g}
=
\sqrt{2(2\pi)^d}Z_{\txtbec}\faftr{\fun{\tau_x}{g}}(0).$$
Then the translated zero-mode two-point covariance is
$$\begin{aligned}
\fun{\prbcov_{\msrbb{P}_{\txtbec}}}{\Xi_f,\Xi_{x,g}}
&=
\sqfun{\prbexp_{\msrbb{P}_{\txtbec}}}
{\cmpconj{\rbk{\Xi_f-\sqfun{\prbexp_{\msrbb{P}_{\txtbec}}}{\Xi_f}}}
\rbk{\Xi_{x,g}-\sqfun{\prbexp_{\msrbb{P}_{\txtbec}}}{\Xi_{x,g}}}}
\\ 
&=
2(2\pi)^d
\smnumberdensity_{\txtbsn,0}(\sminvtemperature)
\cmpconj{\faftr{f}(0)}
\faftr{g}(0),
\end{aligned}$$
independent of the translation of $g$.
For a bounded open region $\mathcal{O}$, the finite-density local condensate
occupation of $s_{\mathcal{O}}$ is
$$n_C(\mathcal{O})
=
\smnumberdensity_{\txtbsn,0}(\sminvtemperature)\absvol{\mathcal{O}},$$
which is finite at fixed density and fixed $\mathcal{O}$.
If a family satisfies
$\smnumberdensity_{\txtbsn,0}^{(\sigma)}(\sminvtemperature)\absvol{\mathcal{O}}
\to\infty$, then the corresponding local number-resolvent expectations in
vectors $f$ with
$\bkt{s_{\mathcal{O}}}{f}_{\fun{\lp^2}{\mathcal{O}}}\ne0$ converge to zero.
\end{prop}

The first assertion is the probability form of the calculation in the proof of Proposition \ref{expedition0012913}. The finite local occupation statement uses the same local number calculation as the proof of Proposition \ref{expedition0012895}; the divergent-family statement uses the same concentration estimate as the proof of Theorem \ref{expedition0012897}.

\begin{proof}
The random variables $\Xi_f$ and $\Xi_{x,g}$ are the probabilistic notation for
the zero-mode covariance term in \eqref{expedition0012440}.
The translation identity
$\faftr{\fun{\tau_x}{g}}(0)=\faftr{g}(0)$ and the polarization step are the
same calculation as in the proof of Proposition \ref{expedition0012913}, hence
\begin{equation}\label{expedition0019742}
\fun{\prbcov_{\msrbb{P}_{\txtbec}}}{\Xi_f,\Xi_{x,g}}
=
2(2\pi)^d
\smnumberdensity_{\txtbsn,0}(\sminvtemperature)
\cmpconj{\faftr{f}(0)}
\faftr{g}(0).
\end{equation}
For the local test, the normalized vector is the same
$s_{\mathcal{O}}$ as in \eqref{expedition0012890}, and
\begin{equation}\label{expedition0019743}
n_C(\mathcal{O})
=
\smnumberdensity_{\txtbsn,0}(\sminvtemperature)\absvol{\mathcal{O}}
\end{equation}
is the finite-density local occupation computed in the proof of Proposition
\ref{expedition0012895}.
The divergent-family number-resolvent statement is obtained by applying the
same spectral-measure concentration estimate
\eqref{expedition0019752}--\eqref{expedition0019754} from the proof of Theorem
\ref{expedition0012897} to the local occupation random variable.
\end{proof}

\subsection{Proper Condensates as Local Occupation Escape}\label{proper-condensates-as-local-occupation-escape}

The two Buchholz papers use different local formulations. The primary-state criterion in \cite[Section 2, Definition and Proposition
2.1]{DetlevBuchholz004} is a bounded number-resolvent vanishing test in a primary state or in a family whose primary limiting component is being tested. The ODLRO criterion in \cite[Definitions I--III, Lemma 2.1, Lemma 3.1, Theorem
3.2, and the Criterion in Section 4]{DetlevBuchholz005} is a local-regular-space formulation for a sequence of Fock states. Both formulations reduce, in the homogeneous mean-field case, to local occupation-number estimates for the normalized constant vector \(s_{\mathcal{O}}\) from \eqref{expedition0012890}.

\subsubsection{Primary-State Number-Resolvent Test}\label{primary-state-number-resolvent-test-1}

The number-resolvent formulation tests whether the local occupation of a chosen vector \(f\) escapes every finite interval. Let \(\oastate[\psi_\sigma]\) be the direct-integral component state indexed by \(\sigma\). For \(f\in\fun{\lp^2}{\mathcal{O}}\), use the local occupation operator \(N_f\) from \eqref{expedition0019745}. The measure \(E_{\sigma,f}\) is the probability measure on \(\closedinterval{0}{\infty}\) determined by the spectral distribution of \(N_f\) in the state \(\oastate[\psi_\sigma]\), as in \eqref{expedition0019746}: for every bounded Borel function \(\varphi\) on \(\closedinterval{0}{\infty}\), \[\int_{\closedinterval{0}{\infty}}
\varphi(t)\opdmsr{E_{\sigma,f}(t)}
=
\fun{\oastate[\psi_\sigma]}{\fun{\varphi}{N_f}}.\] The bounded number-resolvent observable is the scalar function \(t\mapsto(\mu+t)^{-1}\) integrated against \(E_{\sigma,f}\).

\begin{prop}[Probabilistic primary-state number-resolvent test]
Let $\mathcal{O}\subset\fldreal^d$ be bounded and open, and let
$s_{\mathcal{O}}$ be defined by \eqref{expedition0012890}.
At fixed finite condensate density
$0<\smnumberdensity_{\txtbsn,0}(\sminvtemperature)<\infty$,
the number-resolvent expectation
$\int_{\closedinterval{0}{\infty}}
\frac{1}{\mu+t}
\opdmsr{E_{\sigma,s_{\mathcal{O}}}(t)}$
is positive for every $\mu>0$.
For a family with
$\smnumberdensity_{\txtbsn,0}^{(\sigma)}(\sminvtemperature)\absvol{\mathcal{O}}
\to\infty$, every $f$ with
$\bkt{s_{\mathcal{O}}}{f}_{\fun{\lp^2}{\mathcal{O}}}\ne0$ satisfies
$$
\lim_{\sigma}
\int_{\closedinterval{0}{\infty}}
\frac{1}{\mu+t}
\opdmsr{E_{\sigma,f}(t)}
=0,
\quad
\mu>0.
$$
\end{prop}

This is the probability formulation of the spectral-measure calculations in the proof of Proposition \ref{expedition0012895} and in the number-resolvent part of the proof of Theorem \ref{expedition0012897}. It is the criterion corresponding to \cite{DetlevBuchholz004}, not the local-regular-space criterion of \cite{DetlevBuchholz005}.

\begin{proof}
By \eqref{expedition0019745} and \eqref{expedition0019746},
$$
\int_{\closedinterval{0}{\infty}}
\frac{1}{\mu+t}
\opdmsr{E_{\sigma,f}(t)}
=
\fun{\oastate[\psi_\sigma]}{\invrbk{\mu+N_f}}.
$$
For $f=s_{\mathcal{O}}$ at fixed finite condensate density, the remaining
argument is the spectral-measure argument in the proof of Proposition
\ref{expedition0012895}, namely \eqref{expedition0019744}.
This proves positivity for every $\mu>0$.
For a family with
$\smnumberdensity_{\txtbsn,0}^{(\sigma)}(\sminvtemperature)\absvol{\mathcal{O}}
\to\infty$ and
$\bkt{s_{\mathcal{O}}}{f}_{\fun{\lp^2}{\mathcal{O}}}\ne0$,
the remaining argument is the concentration estimate in the proof of Theorem
\ref{expedition0012897}, namely \eqref{expedition0019752}--\eqref{expedition0019754}.
This proves the vanishing in \eqref{expedition0012898}, which is the asserted
probability statement after the preceding identification of $E_{\sigma,f}$.
\end{proof}

\subsubsection{ODLRO Local-Regular-Space Criterion}\label{odlro-local-regular-space-criterion-1}

The ODLRO formulation classifies local test functions by boundedness of their occupations along a sequence. In terms of the spectral measures just defined, define the local regular space by \[
\sphilb{H}_{\txtmeanfield,\txtprob,\txtregular}(\mathcal{O})
=
\set{f\in\fun{\lp^2}{\mathcal{O}}}
{\sup_{\sigma}
\int_{\closedinterval{0}{\infty}}t\opdmsr{E_{\sigma,f}(t)}
<\infty}.
\] The local singular space is the orthogonal complement of this closed subspace when the latter is closed.

\begin{prop}[Probabilistic ODLRO local-regular-space criterion]
For a homogeneous mean-field condensate family whose local occupation
variables are defined as in the primary-state number-resolvent test,
$$\sphilb{H}_{\txtmeanfield,\txtprob,\txtregular}(\mathcal{O})
=
\set{f\in\fun{\lp^2}{\mathcal{O}}}
{\bkt{s_{\mathcal{O}}}{f}_{\fun{\lp^2}{\mathcal{O}}}=0},$$
provided the nonzero-mode thermal local covariance is uniformly bounded along
the family.
Its orthogonal complement is
$\fldcmp s_{\mathcal{O}}$.
Consequently the homogeneous mean-field family has the local regular and
singular spaces required by the ODLRO criterion of \cite{DetlevBuchholz005}.
\end{prop}

This is the probability formulation of the local-space calculation in the proof of Proposition \ref{expedition0012892} and in the local-space part of the proof of Theorem \ref{expedition0012897}.

\begin{proof}
Using \eqref{expedition0019747} and \eqref{expedition0019750}, for normalized
$f\in\fun{\lp^2}{\mathcal{O}}$ the expectation of the local occupation is
$$\int_{\closedinterval{0}{\infty}}t\opdmsr{E_{\sigma,f}(t)}
=
\smnumberdensity_{\txtbsn,0}^{(\sigma)}(\sminvtemperature)
\absvol{\mathcal{O}}
\abs{\bkt{s_{\mathcal{O}}}{f}_{\fun{\lp^2}{\mathcal{O}}}}^2
+\fun{\oastate[\psi_\sigma]}{Y_{\sigma,f}},$$
where $Y_{\sigma,f}$ is the nonzero-mode local occupation defined in
\eqref{expedition0019750}.
The local boundedness of the nonzero-mode thermal covariance gives
$\sup_{\sigma}
\fun{\oastate[\psi_\sigma]}{Y_{\sigma,f}}
<\infty$.
If
$\bkt{s_{\mathcal{O}}}{f}_{\fun{\lp^2}{\mathcal{O}}}=0$, the homogeneous
condensate contribution vanishes and
$\sup_{\sigma}
\int_{\closedinterval{0}{\infty}}t\opdmsr{E_{\sigma,f}(t)}
<\infty$.
Thus
$$\set{f\in\fun{\lp^2}{\mathcal{O}}}
{\bkt{s_{\mathcal{O}}}{f}_{\fun{\lp^2}{\mathcal{O}}}=0}
\subset
\sphilb{H}_{\txtmeanfield,\txtprob,\txtregular}(\mathcal{O}).$$
Conversely, if
$\bkt{s_{\mathcal{O}}}{f}_{\fun{\lp^2}{\mathcal{O}}}\ne0$, then
$\smnumberdensity_{\txtbsn,0}^{(\sigma)}(\sminvtemperature)
\absvol{\mathcal{O}}
\abs{\bkt{s_{\mathcal{O}}}{f}_{\fun{\lp^2}{\mathcal{O}}}}^2
\to\infty$,
so
$\int_{\closedinterval{0}{\infty}}t\opdmsr{E_{\sigma,f}(t)}$
is unbounded.
Therefore $f$ is not in the regular space.
This proves the equality.

The right-hand side is the kernel of the continuous linear functional
$f\mapsto\bkt{s_{\mathcal{O}}}{f}_{\fun{\lp^2}{\mathcal{O}}}$.
Hence it is closed, and its orthogonal complement is
$\fldcmp s_{\mathcal{O}}$.
This is exactly the homogeneous local singular line appearing in
\eqref{expedition0012894}.
\end{proof}

\subsection{Density-Fluctuation Remainder}\label{density-fluctuation-remainder}

The finite-volume probability measure also records the limitation of replacing the total density by the selected scalar value. The exact identity \eqref{expedition0012655} says that the difference between the finite-volume mean-field Hamiltonian and the effective Hamiltonian still contains the square of the random density fluctuation.

\begin{prop}[Probabilistic density-fluctuation remainder]
Let
$$D_{\bar{\smnumberdensity}_{\txtbsn},L}
=
V_L^{\onehalf}\rbk{R_L-\bar{\smnumberdensity}_{\txtbsn}}.$$
The Kac collapse gives $R_L\to\bar{\smnumberdensity}_{\txtbsn}$ in probability,
but it does not determine the limiting law of
$D_{\bar{\smnumberdensity}_{\txtbsn},L}$.
Moreover $D_{\bar{\smnumberdensity}_{\txtbsn},L}$ is conserved by the
finite-volume mean-field dynamics.
\end{prop}

This is the probability form of the density-fluctuation caveat in Proposition \ref{expedition0012654} and Proposition \ref{expedition0012656}. Local equilibrium correlations may be controlled after the density has been selected, whereas nonlocal density fluctuations require a separate limit theorem.

\begin{proof}
The convergence $R_L\to\bar{\smnumberdensity}_{\txtbsn}$ is the same
closed-set LDP estimate proved in Proposition \ref{expedition0019755}.
Multiplication by $V_L^{\onehalf}$ probes the second-order scale of the same
random variable, and an LDP with speed $V_L$ gives the law of large numbers
but not by itself a central-limit theorem.
Indeed, the identity
$$\setone{\abs{R_L-\bar{\smnumberdensity}_{\txtbsn}}>\epsilon}
=
\setone{\abs{D_{\bar{\smnumberdensity}_{\txtbsn},L}}>
\epsilon V_L^{\onehalf}}$$
only controls deviations of $D_{\bar{\smnumberdensity}_{\txtbsn},L}$ on the
growing scale $V_L^{\onehalf}$.
It gives no limiting value for
$$\sqfun{\prbexp_{\msrbb{P}_{\txtmeanfield,L}}}
{\fun{\varphi}{D_{\bar{\smnumberdensity}_{\txtbsn},L}}}$$
when $\varphi$ is a fixed bounded continuous test function.
The conservation statement follows because
$D_{\bar{\smnumberdensity}_{\txtbsn},L}$ is a function of the total number
$N_L$, while the finite-volume Hamiltonian \eqref{expedition0012600} is
diagonal in the occupation-number basis and commutes with $N_L$.
More explicitly,
$$\commutator{\physham_{\txtbsn,\txtmeanfield,\smchemicalpotential,L}}{N_L}=0,
\quad
D_{\bar{\smnumberdensity}_{\txtbsn},L}
=
V_L^{-\onehalf}N_L
-V_L^{\onehalf}\bar{\smnumberdensity}_{\txtbsn},$$
and therefore
$$\commutator{\physham_{\txtbsn,\txtmeanfield,\smchemicalpotential,L}}
{D_{\bar{\smnumberdensity}_{\txtbsn},L}}=0.$$
Consequently, for every bounded Borel function $\varphi$,
$$\napiernum^{\imunit t\physham_{\txtbsn,\txtmeanfield,\smchemicalpotential,L}}
\fun{\varphi}{D_{\bar{\smnumberdensity}_{\txtbsn},L}}
\napiernum^{-\imunit t\physham_{\txtbsn,\txtmeanfield,\smchemicalpotential,L}}
=
\fun{\varphi}{D_{\bar{\smnumberdensity}_{\txtbsn},L}}.$$
Choosing $\varphi(u)=(z-u)^{-1}$ gives the probabilistic content of
\eqref{expedition0012657}.
\end{proof}

\subsection{Probabilistic Mean-Field Conclusions}\label{probabilistic-mean-field-conclusions}

The probability approach records the mean-field conclusions of Section \ref{expedition0012743} in terms of random occupation variables. It starts from the finite-volume occupation-number measure, uses the LDP and quadratic tilt to select the density, derives the Euler equations by the conditioned variational problem, identifies the zero-mode excess as a random occupation density, represents the integrated equilibrium law by a uniform direct-integral parameter, and records ODLRO, proper-condensate tests, and local occupation tests as statements about covariance and number random variables. The proper-condensate part includes both the primary-state number-resolvent criterion of \cite{DetlevBuchholz004} and the ODLRO local-regular-space criterion of \cite{DetlevBuchholz005}. Thus the probabilistic formulation detects BEC exactly through the positive zero-mode excess in \eqref{expedition0012606}, its integrated amplitude law, and the resulting ODLRO and local occupation tests.

\section{Brownian-Loop (Ginibre--Symanzik) Formulation}\label{expedition0019600}

The occupation-number formulation based on the random density \(R_L\) in \eqref{expedition0019667} keeps the whole argument on the classical probability space of the conserved total number. A third, manifestly path-integral formulation expresses the same finite-volume mean-field gas as a gas of interacting \emph{Brownian loops}, in the style of Ginibre's representation of interacting Bose gases by gases of Brownian paths and loops \cite{Ginibre1}, and of Symanzik's loop-gas representation of imaginary-time scalar field theories \cite{KurtSymanzik1}; for a modern exposition see Fröhlich--Knowles--Schlein--Sohinger \cite[\S3--4]{FrohlichKnowlesSchleinSohinger2}.

For a general two-body potential the linearization of the interaction requires a Hubbard--Stratonovich \emph{field} \(\sigma(\tau,x)\) over imaginary time and space. For the mean-field Hamiltonian \eqref{expedition0012600} the two-body potential is the spatially constant kernel \(\lambda/V_L\), and the interaction \(\frac{\lambda}{2V_L}\opfocknumber_{\txtbsn,L}^{2}\) is a function of the conserved number alone. Two simplifications follow: the auxiliary field collapses to a single \emph{scalar}, and the Brownian loops interact only through their total imaginary-time length, i.e.~through the total particle number they carry. Integrating out the scalar reproduces the quadratic density tilt \eqref{expedition0012603}, so this formulation and the occupation-number one are Fourier-dual descriptions of the one finite-volume measure \eqref{expedition0012601}. Throughout this section we work at fixed \(\sminvtemperature>0\), \(\lambda>0\), \(L\), and a chemical potential \(\smchemicalpotential<\min_{k\in\setlattice_L^d}\physham[h]_{\txtparticle,0,L}(k)\) for which the free trace converges.

\subsection{The Mean-Field Interaction as a Constant Two-Body Kernel}\label{the-mean-field-interaction-as-a-constant-two-body-kernel}

On finite-volume Fock space the number operator is \(\opfocknumber_{\txtbsn,L}=\fun{\opfocksndqntdiff_{\txtbsn}}{P_L}
=\sum_{k\in\setlattice_L^d}\opfockcr_{\txtfock}(e_k)\opfockan_{\txtfock}(e_k)\), where \(\setone{e_k}_{k\in\setlattice_L^d}\) is the eigenbasis of \(\physham[h]_{\txtparticle,0,L}\). In the position representation of \(P_L\sphilb{H}_{\txtparticle}\), the creation and annihilation operator-valued distributions are denoted by \(\opfockcr_{\txtfock}(x)\) and \(\opfockan_{\txtfock}(x)\). With this notation, the mean-field interaction is the constant-kernel two-body operator \begin{equation}\label{expedition0019601}
\frac{\lambda}{2V_L}\opfocknumber_{\txtbsn,L}^{2}
=
\frac12\int_{I_L^d}dx\int_{I_L^d}dy\,
\opfockcr_{\txtfock}(x)\opfockan_{\txtfock}(x)
v_L(x-y)
\opfockcr_{\txtfock}(y)\opfockan_{\txtfock}(y)
+\text{(c-number)},
\quad
v_L\equiv\frac{\lambda}{V_L},
\end{equation} up to the additive constant produced by normal ordering. Thus the mean-field gas is the interacting Bose gas of \cite[Eq.~(9)]{FrohlichKnowlesSchleinSohinger2} with \(N=1\) species and the spatially constant (mean-field) two-body potential \(v_L\equiv\lambda/V_L\). Because \(v_L\) is constant in space, the only spatial mode that couples to the interaction is the zero mode \[\int_{I_L^d}
\opfockcr_{\txtfock}(x)\opfockan_{\txtfock}(x)dx
=\opfocknumber_{\txtbsn,L}.\] This identity is the source of both simplifications used in Proposition \ref{expedition0019602} and Theorem \ref{expedition0019606}.

\subsection{Scalar Hubbard--Stratonovich Decoupling}\label{scalar-hubbardstratonovich-decoupling}

The free part \(\physham_{\txtbsn,\txtfr,0,L}-\smchemicalpotential
\opfocknumber_{\txtbsn,L}\) and the number operator \(\opfocknumber_{\txtbsn,L}\) strongly commute, so no Trotter splitting is needed: the interaction may be linearized by a single real Gaussian variable.

\begin{prop}[Scalar Hubbard--Stratonovich representation]\label{expedition0019602}
For every $\sminvtemperature>0$ and $\lambda>0$,
\begin{equation}\label{expedition0019603}
\smpartitionfunc_{\sminvtemperature,\lambda,\smchemicalpotential,L}
=
\rbk{\frac{V_L}{2\pi\sminvtemperature\lambda}}^{1/2}
\int_{\fldreal}
\fnexp{-\frac{V_L s^2}{2\sminvtemperature\lambda}}
\smpartitionfunc_{\sminvtemperature,0,\smchemicalpotential+\imunit \frac{s}{\sminvtemperature},L}
\opdmsr{s},
\end{equation}
where
$\smpartitionfunc_{\sminvtemperature,0,\smchemicalpotential',L}
=\sqfun{\trace_{\spfock_{\txtbsn,L}}}
{\napiernum^{-\sminvtemperature(\physham_{\txtbsn,\txtfr,0,L}
-\smchemicalpotential'\opfocknumber_{\txtbsn,L})}}$
is the free grand partition function, analytically continued to the complex
chemical potential $\smchemicalpotential'=\smchemicalpotential
+\imunit s/\sminvtemperature$.
\end{prop}

This proposition is a preparation for the BEC discussion, not a BEC criterion. It reduces the interaction to one scalar variable, which later becomes the quadratic density tilt and leads to the density selection used to identify the zero-mode excess.

\begin{proof}
Since $\opfocknumber_{\txtbsn,L}$ commutes with
$\physham_{\txtbsn,\txtfr,0,L}-\smchemicalpotential\opfocknumber_{\txtbsn,L}$,
$$\fnexp{-\sminvtemperature
\physham_{\txtbsn,\txtmeanfield,\smchemicalpotential,L}}
=
\fnexp{-\sminvtemperature(\physham_{\txtbsn,\txtfr,0,L}
-\smchemicalpotential\opfocknumber_{\txtbsn,L})}
\cdot
\fnexp{-\frac{\sminvtemperature\lambda}{2V_L}
\opfocknumber_{\txtbsn,L}^{2}}.$$
The scalar Gaussian (Hubbard--Stratonovich) identity, valid on each eigenspace
of the self-adjoint $\opfocknumber_{\txtbsn,L}$,
$$\napiernum^{-\frac{a}{2}\opfocknumber_{\txtbsn,L}^{2}}
=
\frac{1}{\sqrt{2\pi a}}
\int_{\fldreal}
\fnexp{-\frac{s^2}{2a}}\,\napiernum^{\imunit s\opfocknumber_{\txtbsn,L}}
\opdmsr{s},
\quad
a
=
\frac{\sminvtemperature\lambda}{V_L}>0,$$
together with the trace-class property of the heat operator
\eqref{expedition0012600} and Fubini, gives
$$\smpartitionfunc_{\sminvtemperature,\lambda,\smchemicalpotential,L}
=\rbk{\frac{V_L}{2\pi\sminvtemperature\lambda}}^{1/2}
\int_{\fldreal}\opdmsr{s}\,
\fnexp{-\frac{V_L s^2}{2\sminvtemperature\lambda}}\,
\sqfun{\trace_{\spfock_{\txtbsn,L}}}
{\napiernum^{-\sminvtemperature(\physham_{\txtbsn,\txtfr,0,L}
-\smchemicalpotential\opfocknumber_{\txtbsn,L})
+\imunit s\opfocknumber_{\txtbsn,L}}}.$$
Absorbing $\imunit s\opfocknumber_{\txtbsn,L}
=\imunit\frac{s}{\sminvtemperature}\cdot\sminvtemperature
\opfocknumber_{\txtbsn,L}$ into a shift
$\smchemicalpotential\mapsto\smchemicalpotential+\imunit s/\sminvtemperature$ of
the chemical potential identifies the trace with
$\smpartitionfunc_{\sminvtemperature,0,\smchemicalpotential
+\imunit s/\sminvtemperature,L}$.
\end{proof}

The single variable \(s\) is the spatially constant (zero-mode) value of the Hubbard--Stratonovich field \(\sigma(\tau,x)\) of the general path-integral representation \cite[Eqs.~(17)--(21)]{FrohlichKnowlesSchleinSohinger2}; the mean-field covariance \(\frac{\lambda}{\nu}v(x-y)\) there becomes \(\frac{\lambda}{V_L}\delta(\tau-\tau')\), whose only surviving mode is the \(\tau\)-independent constant, i.e.~the scalar \(s\).

\subsection{Ginibre's Brownian-Loop Representation of the Free Gas}\label{ginibres-brownian-loop-representation-of-the-free-gas}

The free factor in \eqref{expedition0019603} is the ideal Bose gas, whose logarithm is Ginibre's gas of Brownian loops. The heat kernel \(\Gamma^{(L)}_{t}(x,y)\) for \(\physham[h]_{\txtparticle,0,L}\) in \eqref{expedition0019756} and the unnormalized Brownian-bridge measures \(\msrbb{W}^{t}_{xy}\) are discussed in the Appendix \ref{expedition0019802}: they are defined in \eqref{expedition0019759} and \eqref{expedition0019761}. The closed-loop trace formula used below is \eqref{expedition0019762}.

\begin{prop}[Brownian-loop representation of the ideal gas]\label{expedition0019604}
For $\opreal\smchemicalpotential'<\min_{k\in\setlattice_L^d}
\physham[h]_{\txtparticle,0,L}(k)$,
\begin{equation}\label{expedition0019605}
\smpartitionfunc_{\sminvtemperature,0,\smchemicalpotential',L}
=
\fnexp{\sum_{\ell\geq1}\frac{\napiernum^{\ell\sminvtemperature
\smchemicalpotential'}}{\ell}
\int_{I_L^d}\opdmsr{x}\int\opdmsr{\msrbb{W}^{\ell\sminvtemperature}_{xx}}},
\end{equation}
the exponential of a gas of single closed Brownian loops:
a loop of \emph{winding} $\ell\geq1$ has duration $\ell\sminvtemperature$
(it wraps $\ell$ times around the thermal circle of circumference
$\sminvtemperature$) and activity $\napiernum^{\ell\sminvtemperature
\smchemicalpotential'}/\ell$.
\end{prop}

This proposition supplies the free-loop reference measure. Its role is to separate the finite-winding thermal contribution from the macroscopic-winding contribution identified later in the condensed case.

\begin{proof}
With $\setone{e_k}$ the eigenbasis of $\physham[h]_{\txtparticle,0,L}$,
$$\smpartitionfunc_{\sminvtemperature,0,\smchemicalpotential',L}
=
\prod_{k\in\setlattice_L^d}
\frac{1}
{1-\napiernum^{-\sminvtemperature
(\physham[h]_{\txtparticle,0,L}(k)-\smchemicalpotential')}},$$
the standard ideal-Bose product, convergent under the stated bound.
Taking the logarithm and expanding
$-\ln(1-z)=\sum_{\ell\geq1}z^{\ell}/\ell$ gives
$$\ln\smpartitionfunc_{\sminvtemperature,0,\smchemicalpotential',L}
=\sum_{\ell\geq1}\frac{\napiernum^{\ell\sminvtemperature\smchemicalpotential'}}{\ell}
\sum_{k\in\setlattice_L^d}\napiernum^{-\ell\sminvtemperature
\physham[h]_{\txtparticle,0,L}(k)}
=\sum_{\ell\geq1}\frac{\napiernum^{\ell\sminvtemperature\smchemicalpotential'}}{\ell}
\sqfun{\trace_{P_L\sphilb{H}_{\txtparticle}}}
{\napiernum^{-\ell\sminvtemperature\physham[h]_{\txtparticle,0,L}}}.$$
By the closed-loop trace formula \eqref{expedition0019762}, the one-particle
trace is
$$\sqfun{\trace_{P_L\sphilb{H}_{\txtparticle}}}
{\napiernum^{-\ell\sminvtemperature\physham[h]_{\txtparticle,0,L}}}
=
\int_{I_L^d}
\Gamma^{(L)}_{\ell\sminvtemperature}(x,x)
\opdmsr{x}
=
\int_{I_L^d}
\opdmsr{x}
\int\opdmsr{\msrbb{W}^{\ell\sminvtemperature}_{xx}},$$
which is \eqref{expedition0019605}.
\end{proof}

\subsection{The Mean-Field Brownian-Loop Gas}\label{the-mean-field-brownian-loop-gas}

Combining \eqref{expedition0019603} and \eqref{expedition0019605} at \(\smchemicalpotential'=\smchemicalpotential+\imunit s/\sminvtemperature\) and expanding the loop exponential gives the announced representation.

\begin{thm}[Brownian-loop representation of the finite-volume mean-field gas]\label{expedition0019606}
The finite-volume mean-field partition function is the partition function of a
gas of interacting Brownian loops.
For the windings $\ell_1,\dots,\ell_n$ in each summand, set
$W=\sum_{k=1}^{n}\ell_k$.
Then
\begin{equation}\label{expedition0019607}
\smpartitionfunc_{\sminvtemperature,\lambda,\smchemicalpotential,L}
=
\sum_{n\geq0}\frac{1}{n!}
\sum_{\ell_1,\dots,\ell_n\geq1}
\rbk{\prod_{k=1}^{n}
\frac{\napiernum^{\ell_k\sminvtemperature\smchemicalpotential}}{\ell_k}
\int_{I_L^d}\opdmsr{x_k}\int\opdmsr{\msrbb{W}^{\ell_k\sminvtemperature}_{x_kx_k}}(\omega_k)}
\fnexp{-\frac{\sminvtemperature\lambda}{2V_L}
W^{2}},
\end{equation}
in which the loops $\omega_1,\dots,\omega_n$ carry windings
$\ell_1,\dots,\ell_n$ and interact only through their total winding
$W$, by the pair potential
$\frac{\sminvtemperature\lambda}{V_L}\ell_i\ell_j$.
\end{thm}

The BEC content of this representation is that the interaction sees only the total winding. Consequently the later collapse of \(W/V_L\) is the loop version of density selection, and any excess over the finite-winding density must be carried by macroscopic windings.

\begin{proof}
Insert \eqref{expedition0019605} into \eqref{expedition0019603}; with
$\napiernum^{\ell\sminvtemperature\smchemicalpotential'}
=\napiernum^{\ell\sminvtemperature\smchemicalpotential}\napiernum^{\imunit\ell s}$
the inner exponent becomes
$$\sum_{\ell\geq1}
\frac{\napiernum^{\ell\sminvtemperature\smchemicalpotential}
\napiernum^{\imunit\ell s}}
{\ell}
\int_{I_L^d}
\opdmsr{x}
\int\opdmsr{\msrbb{W}^{\ell\sminvtemperature}_{xx}}.$$
Expanding $\fnexp{\cdot}
=
\sum_{n\geq0}
\frac{1}{n!}(\cdot)^{n}$ produces, at
order $n$, $n$ loops of windings $\ell_1,\dots,\ell_n$, with the scalar
coupling collected into the single phase
$\napiernum^{\imunit s\sum_k\ell_k}=\napiernum^{\imunit sW}$.
The Gaussian $s$-integral is then carried out by the Fourier identity
$$\rbk{\frac{V_L}{2\pi\sminvtemperature\lambda}}^{1/2}
\int_{\fldreal}
\fnexp{-\frac{V_L s^2}{2\sminvtemperature\lambda}}
\cdot
\napiernum^{\imunit sW}
\opdmsr{s}
=
\fnexp{-\frac{\sminvtemperature\lambda}{2V_L}W^2},$$
which yields \eqref{expedition0019607}.
Expanding the square
$W^2=\sum_{i,j}\ell_i\ell_j$ exhibits the pair potential
$\frac{\sminvtemperature\lambda}{V_L}\ell_i\ell_j$.
\end{proof}

Expression \eqref{expedition0019607} is the mean-field specialization of Ginibre's loop gas \cite[Eq.~(40)]{FrohlichKnowlesSchleinSohinger2}: the \(2\)-loop interaction \(V(\omega_i,\omega_j)
=
\onehalf
\int\int
v_L(\omega_i(\cdot)-\omega_j(\cdot))\) of two loops there reduces, for the constant kernel \(v_L\equiv\lambda/V_L\), to \(\frac{\sminvtemperature\lambda}{2V_L}\ell_i\ell_j\), which is position-independent and depends on the loops only through their windings. Symanzik's loop-gas form follows by the same constant-kernel reduction of \cite[Eqs.~(43)--(44)]{FrohlichKnowlesSchleinSohinger2}.

\subsection{Recovery of the Density Tilt and the Occupation Measure}\label{recovery-of-the-density-tilt-and-the-occupation-measure}

The loop expansion records the particle number sector. In the trace expansion of the free gas, a loop with winding \(\ell\) comes from the \(\ell\)-fold factor \(\fnexp{-\ell\sminvtemperature
\physham[h]_{\txtparticle,0,L}}\) and contributes the coefficient \(\fnexp{\ell\sminvtemperature\smchemicalpotential}\). Thus a loop configuration with windings \(\ell_1,\dots,\ell_n\) is assigned the integer \(W=\sum_{k=1}^{n}\ell_k\), and this integer is the eigenvalue of \(\opfocknumber_{\txtbsn,L}\) on the corresponding occupation-number sector: \[\frac{W}{V_L}
=
\frac{\opfocknumber_{\txtbsn,L}}{V_L}
\text{ on that sector}.\] Since the interaction in \eqref{expedition0019607} depends on the loop configuration only through \(W\), the mean-field loop gas is the free loop gas reweighted by \[
\fnexp{-\frac{\sminvtemperature\lambda}{2V_L}W^2}
=
\fnexp{-\frac{\sminvtemperature\lambda V_L}{2}
\rbk{\frac{W}{V_L}}^2},
\] which is the same quadratic density tilt \eqref{expedition0012603}.

Recall that the measures \(\msrbb{P}_{\txtfr,L}\) in \eqref{expedition0019669} and \(\msrbb{P}_{\txtmeanfield,L}\) in \eqref{expedition0019668}.

\begin{prop}[Fourier duality with the occupation-number measure]\label{expedition0019608}
Let $K_{\sminvtemperature,0,\smchemicalpotential,L}$ and
$K_{\sminvtemperature,\lambda,\smchemicalpotential,L}$ be the free and mean-field
density laws defined by the push-forward formula \eqref{expedition0012728},
i.e. the laws of $R_L=\opfocknumber_{\txtbsn,L}/V_L$ under
$\msrbb{P}_{\txtfr,L}$ and $\msrbb{P}_{\txtmeanfield,L}$.
Then the scalar
Hubbard--Stratonovich identity \eqref{expedition0019603} is the Fourier dual of
the quadratic tilt \eqref{expedition0012603}: for every bounded Borel $F$,
$$\sqfun{\prbexp_{\msrbb{P}_{\txtmeanfield,L}}}{F(R_L)}
=\frac{\int_{\fldreal_{\geq0}}F(r)\,
\fnexp{-\frac{\sminvtemperature\lambda V_L r^2}{2}}
\opdmsr{K_{\sminvtemperature,0,\smchemicalpotential,L}(r)}}
{\int_{\fldreal_{\geq0}}
\fnexp{-\frac{\sminvtemperature\lambda V_L s^2}{2}}
\opdmsr{K_{\sminvtemperature,0,\smchemicalpotential,L}(s)}},$$
and the Gaussian variable $s$ of \eqref{expedition0019603} is conjugate to the
density $r$, with
$$\fnexp{-\frac{\sminvtemperature\lambda V_L r^2}{2}}
=\rbk{\frac{V_L}{2\pi\sminvtemperature\lambda}}^{1/2}
\int_{\fldreal}
\fnexp{-\frac{V_L s^2}{2\sminvtemperature\lambda}}
\napiernum^{\imunit sV_L r}
\opdmsr{s}.$$
\end{prop}

This proposition links the loop representation to the occupation-number density law. For BEC, it identifies the random total winding density with the same density variable whose thermodynamic limit determines the zero-mode excess.

\begin{proof}
The free loop gas \eqref{expedition0019605}, grouped by total winding $W=V_L r$,
is the law $K_{\sminvtemperature,0,\smchemicalpotential,L}$ of $R_L$
under the ideal-gas Gibbs measure.
Replacing $\smchemicalpotential$ by
$\smchemicalpotential+\imunit s/\sminvtemperature$ multiplies the weight of
$\setone{W}$ by $\napiernum^{\imunit sW}=\napiernum^{\imunit sV_L r}$,
so the $s$-integral in \eqref{expedition0019603}
acts on the law of $R_L$ by the scalar Gaussian integral,
producing the factor
$\fnexp{-\frac{\sminvtemperature\lambda V_L r^2}{2}}$.
Normalizing by the same expression with $F=1$ gives the tilted expectation,
which is the finite-volume identity \eqref{expedition0012603}.
\end{proof}

Thus the Brownian-loop gas, the occupation-number measure, and the resolvent algebra describe the same finite-volume state \eqref{expedition0012601} from three sides: the loop gas exhibits it as interacting Brownian motion, the occupation measure as a tilted density law, and the resolvent algebra as an operator system. The Kac collapse of Theorem \ref{expedition0012647} is, in the loop picture, the concentration of the total loop density \(W/V_L\) at the selected value \(\bar{\smnumberdensity}_{\txtbsn}\).

\subsection{Kac Collapse and Macroscopic Loop Density}\label{kac-collapse-and-macroscopic-loop-density}

The Brownian-loop gas carries a natural random total winding. For a loop configuration \(\Omega=(\ell_1,\omega_1,\dots,\ell_n,\omega_n)\), put \[W_L(\Omega)
=
\sum_{j=1}^{n}\ell_j,
\quad
R^{\txtloop}_L(\Omega)
=
\frac{W_L(\Omega)}{V_L}.\] The loop probability measure is the normalized version of the weight in \eqref{expedition0019607}. Denote this probability measure by \(\msrbb{P}^{\txtloop}_{\txtmeanfield,L}\). Because each loop of winding \(\ell\) represents \(\ell\) particles, the random variable \(W_L\) is exactly the total particle number in the occupation-number description.

\begin{prop}[Loop-density law and Kac collapse]\label{expedition0019610}
The law of $R^{\txtloop}_L$ under the finite-volume mean-field loop gas is
$K_{\sminvtemperature,\lambda,\smchemicalpotential,L}$.
Consequently $R^{\txtloop}_L
\to
\bar{\smnumberdensity}_{\txtbsn}$ in probability.
\end{prop}

This is the Brownian-loop form of Theorem \ref{expedition0012647}. The statement is a law of large numbers for total winding density. It is the loop-level density-selection step in the BEC argument: after \(R^{\txtloop}_L\) concentrates at \(\bar{\smnumberdensity}_{\txtbsn}\), the comparison with the finite-winding critical density determines whether macroscopic winding remains.

\begin{proof}
Group the loop expansion \eqref{expedition0019607} by the total winding
$W=\sum_{j=1}^{n}\ell_j$.
For a bounded Borel function $F$,
$$\begin{aligned}
&\sqfun{\prbexp_{\msrbb{P}^{\txtloop}_{\txtmeanfield,L}}}
{F(R^{\txtloop}_L)}
\\ 
&=
\frac{1}{\smpartitionfunc_{\sminvtemperature,\lambda,\smchemicalpotential,L}}
\sum_{n\geq0}\frac{1}{n!}
\sum_{\ell_1,\dots,\ell_n\geq1}
F\rbk{\frac{\sum_{j=1}^{n}\ell_j}{V_L}}
\rbk{\prod_{j=1}^{n}
\frac{\napiernum^{\ell_j\sminvtemperature\smchemicalpotential}}{\ell_j}
\int_{I_L^d}\opdmsr{x_j}
\int\opdmsr{\msrbb{W}^{\ell_j\sminvtemperature}_{x_jx_j}}(\omega_j)}
\\ 
&\quad\times
\fnexp{-\frac{\sminvtemperature\lambda}{2V_L}
\rbk{\sum_{j=1}^{n}\ell_j}^{2}} .
\end{aligned}$$
The free loop gas expansion \eqref{expedition0019605} is the exponential
generating function of the same loops without the final quadratic factor.
For fixed $W\in\monnat$, the coefficient of total winding $W$ is
$$\begin{aligned}
&\frac{1}{W!}
\fnrestr{
\frac{d^W}{du^W}
\fnexp{
\sum_{\ell\geq1}
\frac{u^\ell\napiernum^{\ell\sminvtemperature\smchemicalpotential}}{\ell}
\sqfun{\trace_{P_L\sphilb{H}_{\txtparticle}}}
{\napiernum^{-\ell\sminvtemperature\physham[h]_{\txtparticle,0,L}}}}
}{u=0}
\\ 
&=
\frac{1}{W!}
\fnrestr{
\frac{d^W}{du^W}
\prod_{k\in\setlattice_L^d}
\frac{1}
{1-u\napiernum^{-\sminvtemperature
\rbk{\physham[h]_{\txtparticle,0,L}(k)-\smchemicalpotential}}}
}{u=0}
\\ 
&=
\sum_{\substack{\mathsf{n}\in\mathcal{O}_L\\
N_L(\mathsf{n})=W}}
\fnexp{-\sminvtemperature
\sum_{k\in\setlattice_L^d}
\rbk{\physham[h]_{\txtparticle,0,L}(k)-\smchemicalpotential}
\mathsf{n}_k}.
\end{aligned}$$
Multiplying this coefficient by the mean-field factor
$\fnexp{-\sminvtemperature\lambda W^2/(2V_L)}$ gives
$$\begin{aligned}
&\sum_{W\geq0}
F\rbk{\frac{W}{V_L}}
\fnexp{-\frac{\sminvtemperature\lambda}{2V_L}W^2}
\sum_{\substack{\mathsf{n}\in\mathcal{O}_L\\
N_L(\mathsf{n})=W}}
\fnexp{-\sminvtemperature
\sum_{k\in\setlattice_L^d}
\rbk{\physham[h]_{\txtparticle,0,L}(k)-\smchemicalpotential}
\mathsf{n}_k}
\\ 
&=
\sum_{\mathsf{n}\in\mathcal{O}_L}
F\rbk{\frac{N_L(\mathsf{n})}{V_L}}
\fnexp{-\sminvtemperature
\rbk{
\sum_{k\in\setlattice_L^d}
\rbk{\physham[h]_{\txtparticle,0,L}(k)-\smchemicalpotential}
\mathsf{n}_k
+\frac{\lambda}{2V_L}N_L(\mathsf{n})^2}}.
\end{aligned}$$
After division by
$\smpartitionfunc_{\sminvtemperature,\lambda,\smchemicalpotential,L}$,
this is
$\sqfun{\prbexp_{\msrbb{P}_{\txtmeanfield,L}}}{F(R_L)}$
with $R_L$ and $\msrbb{P}_{\txtmeanfield,L}$ defined by
\eqref{expedition0019667} and \eqref{expedition0019668}.
By \eqref{expedition0012728}, this push-forward law is
$K_{\sminvtemperature,\lambda,\smchemicalpotential,L}$.
The convergence in probability is then Theorem \ref{expedition0012647}.
\end{proof}

Finite windings describe the locally normal nonzero-mode thermal cloud. Macroscopic windings describe the zero-mode excess. For \(A\in\monnat\), set \[W_{L,\leq A}(\Omega)
=
\sum_{j=1}^{n}\ell_j\fndef{\closedinterval{1}{A}}(\ell_j),
\quad
W_{L,>A}(\Omega)
=
W_L(\Omega)-W_{L,\leq A}(\Omega).\]

\begin{prop}[Macroscopic-loop density in the condensed case]\label{expedition0019611}
In the condensed case \eqref{expedition0012606},
$$\begin{aligned}
\lim_{A\to\infty}
\lim_{L\to\infty}
\sqfun{\prbexp_{\msrbb{P}^{\txtloop}_{\txtmeanfield,L}}}
{\frac{W_{L,\leq A}}{V_L}}
=
\smnumberdensity_{\txtbsn,\txtcritical}(\sminvtemperature),
\quad
\lim_{A\to\infty}
\lim_{L\to\infty}
\sqfun{\prbexp_{\msrbb{P}^{\txtloop}_{\txtmeanfield,L}}}
{\frac{W_{L,>A}}{V_L}}
=
\smnumberdensity_{\txtbsn,0}(\sminvtemperature).
\end{aligned}$$
\end{prop}

Thus the zero-mode excess of the occupation-number formulation is the density carried by loops whose winding diverges after the thermodynamic limit. This is the direct Brownian-loop statement of BEC occurrence: in the condensed case \eqref{expedition0012606}, the positive density \(\smnumberdensity_{\txtbsn,0}(\sminvtemperature)\) is exactly the density carried by macroscopic loops.

\begin{proof}
For fixed $A$, only finitely many windings are tested.
The mean-field weight depends on the total winding through the Kac-selected
density, and finite windings use the selected nonzero-mode energy
\eqref{expedition0012605}.
In the condensed case \eqref{expedition0012606},
$\smchemicalpotential_{\mathrm{sel}}=\lambda\bar{\smnumberdensity}_{\txtbsn}$ by
\eqref{expedition0012606}, so the finite-winding loop intensity is the
critical ideal-gas intensity.
Therefore
$$\begin{aligned}
\lim_{L\to\infty}
\sqfun{\prbexp_{\msrbb{P}^{\txtloop}_{\txtmeanfield,L}}}
{\frac{W_{L,\leq A}}{V_L}}
=
\sum_{\ell=1}^{A}
\napiernum^{-\ell\sminvtemperature\cdot0}
\frac{1}{(2\pi)^d}
\int_{\fldreal^d}
\napiernum^{-\ell\sminvtemperature
\physham[h]_{\txtparticle,0}(k)}
\opdmsr{k}
=
\frac{1}{(2\pi)^d}
\int_{\fldreal^d}
\sum_{\ell=1}^{A}
\napiernum^{-\ell\sminvtemperature\physham[h]_{\txtparticle,0}(k)}
\opdmsr{k}.
\end{aligned}$$
Letting $A\to\infty$ and using monotone convergence gives
$$\frac{1}{(2\pi)^d}
\int_{\fldreal^d}
\frac{1}{\napiernum^{\sminvtemperature\physham[h]_{\txtparticle,0}(k)}-1}
\opdmsr{k}
=
\smnumberdensity_{\txtbsn,\txtcritical}(\sminvtemperature)$$
by \eqref{expedition0019397}.

The total winding density converges to
$\bar{\smnumberdensity}_{\txtbsn}$ by Proposition \ref{expedition0019610}.
Since
$$\frac{W_{L,>A}}{V_L}
=
\frac{W_L}{V_L}
-\frac{W_{L,\leq A}}{V_L},$$
the iterated limit gives
$$\bar{\smnumberdensity}_{\txtbsn}
-\smnumberdensity_{\txtbsn,\txtcritical}(\sminvtemperature)
=
\smnumberdensity_{\txtbsn,0}(\sminvtemperature),$$
where the last equality is the zero-mode excess in
\eqref{expedition0012606}.
\end{proof}

\subsection{Brownian-Loop Zero-Mode ODLRO}\label{brownian-loop-zero-mode-odlro}

The ODLRO used in the main construction is the resolvent-algebra one from Proposition \ref{expedition0012913}: it is a statement about zero-mode support condition in the realified test-function space and about the non-decaying singular part of the represented two-point function. In the Brownian-loop formulation the same information is carried by macroscopic winding density. The one-particle density-matrix open-path formula is recorded later as a standard comparison, not as the primary notion used in the paper.

\begin{prop}[Brownian-loop form of zero-mode ODLRO]\label{expedition0019612}
In the condensed case \eqref{expedition0012606}, for test functions $f,g$ as in
Proposition \ref{expedition0012913}, the
macroscopic-winding contribution of Proposition \ref{expedition0019611} gives
the non-decaying realified-field covariance
$2(2\pi)^d
\smnumberdensity_{\txtbsn,0}(\sminvtemperature)
\cmpconj{\faftr{f}(0)}
\faftr{g}(0)$,
which is nonzero exactly when $\faftr{f}(0)\faftr{g}(0)\ne0$.
\end{prop}

This is the Brownian-loop counterpart of Proposition \ref{expedition0012913}. It identifies the non-decaying two-point contribution in \eqref{expedition0012440} and the same zero-mode support condition. For BEC, the proposition shows how the macroscopic-loop density appears as the zero-mode ODLRO term used in the resolvent-algebra construction.

\begin{proof}
The separation of finite and macroscopic winding densities is exactly
Proposition \ref{expedition0019611}; the latter density is
$\smnumberdensity_{\txtbsn,0}(\sminvtemperature)$.
With this density in place of the zero-mode density in
\eqref{expedition0012440}, the translation identity and polarization argument
are the same as in the proof of Proposition \ref{expedition0012913}.
They give the covariance in the statement and the same zero-mode condition
$\faftr{f}(0)\faftr{g}(0)\ne0$.
\end{proof}

\subsection{Brownian-Loop Proper Condensates}\label{brownian-loop-proper-condensates}

The same macroscopic-winding picture also contains the two Buchholz proper-condensate formulations: the number-resolvent test \eqref{expedition0019101} and the local-regular-space test \eqref{expedition0019103}. The number-resolvent formulation tests escape of local occupation in a direct-integral component probability law. The ODLRO local-regular-space formulation tests boundedness of local occupation along a sequence.

\begin{prop}[Loop interpretation of the two proper-condensate criteria]\label{expedition0019613}
For a bounded open region $\mathcal{O}$, macroscopic loops restricted to
$\mathcal{O}$ select the normalized constant vector
$s_{\mathcal{O}}$ from \eqref{expedition0012890}.
At fixed finite condensate density, the local occupation of
$s_{\mathcal{O}}$ is
finite, so the primary-state number-resolvent criterion of
\cite{DetlevBuchholz004} is not satisfied.
Along a family for which the local macroscopic-loop density in
$\mathcal{O}$ diverges, the local-regular-space criterion of
\cite{DetlevBuchholz005} gives
$$\set{f\in\fun{\lp^2}{\mathcal{O}}}
{\int_{\mathcal{O}}f(x)\opdmsr{x}=0}$$
with singular line $\fldcmp s_{\mathcal{O}}$.
For every $f\in\fun{\lp^2}{\mathcal{O}}$ with
$\bkt{s_{\mathcal{O}}}{f}_{\fun{\lp^2}{\mathcal{O}}}\ne0$, the
number-resolvent criterion of \cite{DetlevBuchholz004} is satisfied.
\end{prop}

This proposition is the Brownian-loop counterpart of Proposition \ref{expedition0012895} and Theorem \ref{expedition0012897}. The fixed finite-density assertion corresponds to Proposition \ref{expedition0012895}, while the divergent-family assertions correspond to Theorem \ref{expedition0012897}.

\begin{proof}
Proposition \ref{expedition0019611} identifies the macroscopic-winding density
with $\smnumberdensity_{\txtbsn,0}(\sminvtemperature)$.
Using the local constant vector identity \eqref{expedition0012891}, the local
macroscopic-loop contribution tested by
$f\in\fun{\lp^2}{\mathcal{O}}$ is
$$\smnumberdensity_{\txtbsn,0}(\sminvtemperature)
\absvol{\mathcal{O}}
\abs{\bkt{s_{\mathcal{O}}}{f}_{\fun{\lp^2}{\mathcal{O}}}}^2.$$
Thus macroscopic winding selects the normalized constant vector in
\eqref{expedition0012890}.
At fixed finite density, the strict number-resolvent assertion is the same
spectral-measure argument as in the proof of Proposition
\ref{expedition0012895}.
For a divergent family, the regular-space assertion and the number-resolvent
vanishing are the same local occupation decomposition and escape estimate as
in the proof of Theorem \ref{expedition0012897}, with the above
macroscopic-loop contribution in place of the zero-mode contribution.
\end{proof}

\subsection{Standard One-Particle Density-Matrix Open Paths}\label{standard-one-particle-density-matrix-open-paths}

The standard one-particle density-matrix formulation is a comparison with the usual open-path representation. It is placed after the zero-mode ODLRO and proper-condensate statements because the main ODLRO notion of this paper is the resolvent-algebra zero-mode ODLRO statement of Proposition \ref{expedition0012913}. Reduced density matrices have the loop-gas form with one open Brownian path joining the two arguments and all remaining paths closed.

\begin{prop}[Finite-volume open-path representation]\label{expedition0019670}
For $x,x'\in I_L^d$, the one-particle reduced density matrix
$$\gamma_{1,L}(x,x')
=
\fun{\oastate[\psi_{\txtmeanfield,\sminvtemperature,\smchemicalpotential,L}]}
{\opfockcr_{\txtfock}(x)\opfockan_{\txtfock}(x')}$$
is given by
\begin{equation}\label{expedition0019609}
\begin{aligned}
\gamma_{1,L}(x,x')
&=
\frac{1}{\smpartitionfunc_{\sminvtemperature,\lambda,\smchemicalpotential,L}}
\sum_{\ell_0\geq1}
\napiernum^{\ell_0\sminvtemperature\smchemicalpotential}
\int\opdmsr{\msrbb{W}^{\ell_0\sminvtemperature}_{x'x}}(\omega_0)
\\ 
&\quad\times
\sum_{n=0}^{\infty}\frac{1}{n!}
\sum_{\ell_1,\ldots,\ell_n\geq1}
\prod_{j=1}^{n}
\rbk{
\frac{\napiernum^{\ell_j\sminvtemperature\smchemicalpotential}}{\ell_j}
\int_{I_L^d}\opdmsr{x_j}
\int\opdmsr{\msrbb{W}^{\ell_j\sminvtemperature}_{x_jx_j}}(\omega_j)}
\\ 
&\quad\times
\fnexp{-\frac{\sminvtemperature\lambda}{2V_L}
\rbk{\ell_0+\sum_{j=1}^{n}\ell_j}^{2}} .
\end{aligned}
\end{equation}
\end{prop}

This is the mean-field specialization of the open-path representation in \cite[Eq.~(41)]{FrohlichKnowlesSchleinSohinger2}. The open path \(\omega_0\) runs from \(x'\) to \(x\) and has duration \(\ell_0\sminvtemperature\). The mean-field weight in \eqref{expedition0019609} counts this open path in the same total winding as the closed loops. The closed-loop part is the same Brownian-loop expansion as Proposition \ref{expedition0019604}, and the mean-field weighting is the same Hubbard--Stratonovich integration as in Theorem \ref{expedition0019606}. For BEC, this proposition is only the standard one-particle density-matrix comparison. The main BEC signal in the paper remains the zero-mode ODLRO and local proper-condensate information in Propositions \ref{expedition0019612} and \ref{expedition0019613}.

\begin{proof}
Define the unnormalized numerator
$$\Xi_L(x,x')
=
\sqfun{\trace_{\spfock_{\txtbsn,L}}}
{\opfockcr_{\txtfock}(x)\opfockan_{\txtfock}(x')
\napiernum^{-\sminvtemperature
\physham_{\txtbsn,\txtmeanfield,\smchemicalpotential,L}}}.$$
Then
$$\gamma_{1,L}(x,x')
=
\frac{\Xi_L(x,x')}
{\smpartitionfunc_{\sminvtemperature,\lambda,\smchemicalpotential,L}}.$$
The operator
$\opfockcr_{\txtfock}(x)\opfockan_{\txtfock}(x')$ preserves particle number.
Applying the scalar identity used in \eqref{expedition0019603} inside this trace
gives
$$\begin{aligned}
\Xi_L(x,x')
&=
\rbk{\frac{V_L}{2\pi\sminvtemperature\lambda}}^{1/2}
\int_{\fldreal}
\fnexp{-\frac{V_L s^2}{2\sminvtemperature\lambda}}
\sqfun{\trace_{\spfock_{\txtbsn,L}}}
{\opfockcr_{\txtfock}(x)\opfockan_{\txtfock}(x')
\napiernum^{-\sminvtemperature
\rbk{\physham_{\txtbsn,\txtfr,0,L}
-\smchemicalpotential_s\opfocknumber_{\txtbsn,L}}}}
\opdmsr{s},
\end{aligned}$$
where
$\smchemicalpotential_s
=
\smchemicalpotential+\imunit\frac{s}{\sminvtemperature}$.
Let
$K_s
=
\napiernum^{-\sminvtemperature
\rbk{\physham[h]_{\txtparticle,0,L}-\smchemicalpotential_s}}$.
In the finite-dimensional one-particle space,
the free Fock trace with one creation-annihilation insertion is
$$\begin{aligned}
&\sqfun{\trace_{\spfock_{\txtbsn,L}}}
{\opfockcr_{\txtfock}(x)\opfockan_{\txtfock}(x')
\napiernum^{-\sminvtemperature
\rbk{\physham_{\txtbsn,\txtfr,0,L}
-\smchemicalpotential_s\opfocknumber_{\txtbsn,L}}}}
=
\smpartitionfunc_{\sminvtemperature,0,\smchemicalpotential_s,L}
\sum_{\ell_0\geq1}
\napiernum^{\ell_0\sminvtemperature\smchemicalpotential_s}
\Gamma^{(L)}_{\ell_0\sminvtemperature}(x',x)
\\ 
&=
\smpartitionfunc_{\sminvtemperature,0,\smchemicalpotential_s,L}
\sum_{\ell_0\geq1}
\napiernum^{\ell_0\sminvtemperature\smchemicalpotential_s}
\int\opdmsr{\msrbb{W}^{\ell_0\sminvtemperature}_{x'x}}(\omega_0).
\end{aligned}$$
The last equality is the bridge-measure normalization
\eqref{expedition0019761}.
The heat-kernel sum follows by diagonalizing $K_s$:
for an eigenbasis $\setone{e_k}$,
put
$$r_{s,k}
=
\napiernum^{-\sminvtemperature
\rbk{\physham[h]_{\txtparticle,0,L}(k)-\smchemicalpotential_s}}.$$
Since $\opreal\smchemicalpotential_s$ is in the free trace-class domain, one has
$\abs{r_{s,k}}<1$ for every $k$.
The occupation-number basis
$\ket{\mathsf{m}}=\ket{\fml{m_k}{k\in\setlattice_L^d}}$
diagonalizes
$\fun{\opfocksndqntdiff_{\txtbsn}}
{\physham[h]_{\txtparticle,0,L}-\smchemicalpotential_s}$, and
$$
\napiernum^{-\sminvtemperature
\rbk{\physham_{\txtbsn,\txtfr,0,L}
-\smchemicalpotential_s\opfocknumber_{\txtbsn,L}}}
\ket{\mathsf{m}}
=
\rbk{\prod_{p\in\setlattice_L^d}r_{s,p}^{m_p}}
\ket{\mathsf{m}}.
$$
Expanding the field distributions in this basis gives
$$
\opfockcr_{\txtfock}(x)\opfockan_{\txtfock}(x')
=
\sum_{k,q\in\setlattice_L^d}
\cmpconj{e_k(x)}e_q(x')
\fun{\opfockcr_{\txtfock}}{e_k}
\fun{\opfockan_{\txtfock}}{e_q}.
$$
Only the diagonal terms $k=q$ survive in the trace over the occupation-number
basis, and on $\ket{\mathsf{m}}$ they contribute $m_k$.
Hence it holds that
$$\begin{aligned}
&\sqfun{\trace_{\spfock_{\txtbsn,L}}}
{\opfockcr_{\txtfock}(x)\opfockan_{\txtfock}(x')
\napiernum^{-\sminvtemperature
\rbk{\physham_{\txtbsn,\txtfr,0,L}
-\smchemicalpotential_s\opfocknumber_{\txtbsn,L}}}}
=
\sum_{k\in\setlattice_L^d}
\cmpconj{e_k(x)}
e_k(x')
\sum_{m\geq0}
m
\rbk{r_{s,k}}^{m}
\prod_{q\neq k}
\sum_{m_q\geq0}
\rbk{r_{s,q}}^{m_q}
\\ 
&=
\sum_{k\in\setlattice_L^d}
\cmpconj{e_k(x)}
e_k(x')
\frac{r_{s,k}}{\rbk{1-r_{s,k}}^2}
\prod_{q\neq k}
\frac{1}{1-r_{s,q}}
=
\smpartitionfunc_{\sminvtemperature,0,\smchemicalpotential_s,L}
\sum_{k\in\setlattice_L^d}
\cmpconj{e_k(x)}
e_k(x')
\frac{\napiernum^{-\sminvtemperature
\rbk{\physham[h]_{\txtparticle,0,L}(k)-\smchemicalpotential_s}}}
{1-\napiernum^{-\sminvtemperature
\rbk{\physham[h]_{\txtparticle,0,L}(k)-\smchemicalpotential_s}}},
\end{aligned}$$
where
$\smpartitionfunc_{\sminvtemperature,0,\smchemicalpotential_s,L}
=
\prod_{q\in\setlattice_L^d}\frac{1}{1-r_{s,q}}$.
Finally, noting that $$\frac{r_{s,k}}{1-r_{s,k}}
=
\sum_{\ell_0\geq1}r_{s,k}^{\ell_0}
=
\sum_{\ell_0\geq1}
\napiernum^{\ell_0\sminvtemperature\smchemicalpotential_s}
\napiernum^{-\ell_0\sminvtemperature
\physham[h]_{\txtparticle,0,L}(k)},$$
substitution gives
$$\begin{aligned}
\sum_{k\in\setlattice_L^d}
\cmpconj{e_k(x)}
e_k(x')
\frac{r_{s,k}}{1-r_{s,k}}
=
\sum_{\ell_0\geq1}
\napiernum^{\ell_0\sminvtemperature\smchemicalpotential_s}
\sum_{k\in\setlattice_L^d}
\napiernum^{-\ell_0\sminvtemperature
\physham[h]_{\txtparticle,0,L}(k)}
\cmpconj{e_k(x)}
e_k(x')
=
\sum_{\ell_0\geq1}
\napiernum^{\ell_0\sminvtemperature\smchemicalpotential_s}
\Gamma^{(L)}_{\ell_0\sminvtemperature}(x',x),
\end{aligned}$$
which is the heat-kernel expression above.

Expanding the remaining free factor by \eqref{expedition0019605} and applying
the same Gaussian integration as in the proof of Theorem
\ref{expedition0019606} gives the closed-loop factors and the quadratic
weight in \eqref{expedition0019609}.
Substitution into $\gamma_{1,L}(x,x')=\Xi_L(x,x')/
\smpartitionfunc_{\sminvtemperature,\lambda,\smchemicalpotential,L}$ gives
\eqref{expedition0019609}.
\end{proof}

\begin{prop}[Standard one-particle density-matrix BEC criterion]\label{expedition0019671}
Let $\Gamma_{1,L}$ be the one-particle density operator with kernel
$\gamma_{1,L}$ from Proposition \ref{expedition0019670}, and let $e_0$ be the
normalized constant vector in $P_L\sphilb{H}_{\txtparticle}$.
Then
$$
\frac{1}{V_L}
\bkt{e_0}{\Gamma_{1,L}e_0}_{P_L\sphilb{H}_{\txtparticle}}
=
\sqfun{\prbexp_{\msrbb{P}_{\txtmeanfield,L}}}{R_{0,L}}.
$$
In the condensed case \eqref{expedition0012606}, this quantity converges to
$\smnumberdensity_{\txtbsn,0}(\sminvtemperature)>0$.
In the normal case \eqref{expedition0012607}, it converges to $0$.
Thus the standard one-particle density-matrix criterion detects BEC exactly
when $\smnumberdensity_{\txtbsn,0}(\sminvtemperature)>0$, with the constant
zero-mode vector as the macroscopic eigenvector.
\end{prop}

This is the standard density-matrix characterization of BEC. Its probabilistic counterpart is Proposition \ref{expedition0019803}, and its Brownian-loop counterpart is the macroscopic-winding identification in Proposition \ref{expedition0019611}. It is distinct from the main zero-mode ODLRO criterion of Proposition \ref{expedition0019612}, whose resolvent-algebra counterpart is Proposition \ref{expedition0012913}, but both criteria are fed by the same scalar zero-mode excess in \eqref{expedition0012606}.

\begin{proof}
The integral kernel $\gamma_{1,L}$ represents the one-particle density operator
$\Gamma_{1,L}$.
Therefore
$$\begin{aligned}
\bkt{e_0}{\Gamma_{1,L}e_0}_{P_L\sphilb{H}_{\txtparticle}}
=
\int_{I_L^d}\opdmsr{x}
\int_{I_L^d}\opdmsr{x'}
\cmpconj{e_0(x)}
\gamma_{1,L}(x,x')
e_0(x')
=
\fun{\oastate[\psi_{\txtmeanfield,\sminvtemperature,\smchemicalpotential,L}]}
{\fun{\opfockcr_{\txtfock}}{e_0}
\fun{\opfockan_{\txtfock}}{e_0}}.
\end{aligned}$$
The finite-volume Gibbs state is diagonal in the occupation-number basis.
Since $e_0$ is the $k=0$ one-particle eigenvector,
$$\fun{\opfockcr_{\txtfock}}{e_0}
\fun{\opfockan_{\txtfock}}{e_0}
=
\mathsf{n}_0$$
on an occupation configuration $\mathsf{n}$.
Using the probability measure \eqref{expedition0019668}, this gives
$$\begin{aligned}
\bkt{e_0}{\Gamma_{1,L}e_0}_{P_L\sphilb{H}_{\txtparticle}}
=
\sum_{\mathsf{n}\in\mathcal{O}_L}
\mathsf{n}_0
\fun{\msrbb{P}_{\txtmeanfield,L}}{\setone{\mathsf{n}}}
=
\sqfun{\prbexp_{\msrbb{P}_{\txtmeanfield,L}}}{\mathsf{n}_0}
=
V_L\sqfun{\prbexp_{\msrbb{P}_{\txtmeanfield,L}}}{R_{0,L}},
\end{aligned}$$
where $R_{0,L}(\mathsf{n})=\mathsf{n}_0/V_L$.
Division by $V_L$ proves the finite-volume identity.

The open-path formula \eqref{expedition0019609} is the kernel version of the
same operator identity.
The thermodynamic limit of $R_{0,L}$ is Proposition \ref{expedition0019803};
equivalently, Proposition \ref{expedition0019611} identifies the
macroscopic-winding contribution with the same zero-mode excess.
Substituting this limit in the finite-volume identity gives the stated
standard one-particle density-matrix criterion.
\end{proof}

\subsection{Brownian-Loop Mean-Field Conclusions}\label{brownian-loop-mean-field-conclusions}

The Brownian-loop formulation gives a path-integral realization of the same finite-volume mean-field measure, not an additional BEC criterion. The constant two-body kernel collapses the Hubbard--Stratonovich field to the single scalar variable in Proposition \ref{expedition0019602}. After the free factor is expanded as Ginibre's loop gas in Proposition \ref{expedition0019604}, Theorem \ref{expedition0019606} turns the scalar integration into the quadratic interaction of the total winding. Proposition \ref{expedition0019608} then identifies this quadratic total-winding weight with the occupation-number density tilt \eqref{expedition0012603}. Thus the loop gas and the occupation-number law are two descriptions of the same finite-volume probability measure.

The BEC information appears in the loop picture only after this identification. By Proposition \ref{expedition0019610}, the total winding density has the same law as the particle density and therefore collapses to \(\bar{\smnumberdensity}_{\txtbsn}\) by Theorem \ref{expedition0012647}. In the condensed case, Proposition \ref{expedition0019611} separates this selected density into the finite-winding critical density and the macroscopic-winding density. The latter is exactly the zero-mode excess \(\smnumberdensity_{\txtbsn,0}(\sminvtemperature)\). Consequently the non-decaying covariance in Proposition \ref{expedition0019612} is the Brownian-loop expression of the zero-mode ODLRO term in Proposition \ref{expedition0012913}.

The open-path formula in Proposition \ref{expedition0019670} is therefore a standard density-matrix comparison placed after the zero-mode and proper-condensate statements. Proposition \ref{expedition0019671} recovers the usual macroscopic occupation of the constant one-particle vector from the same zero-mode excess, while Proposition \ref{expedition0019613} translates the two Buchholz local criteria into local occupation statements for the macroscopic-loop contribution. In this form, the Brownian-loop section confirms the separation used throughout the paper: finite-density BEC produces a zero-mode excess and an ODLRO term, whereas the strict proper-condensate number-resolvent conclusion requires local occupation escape.

\appendix

\section{Appendix}\label{appendix}

\subsection{Preparation for Probability for Mean-field}\label{preparation-for-probability-for-mean-field}

\begin{prop}[Exponential Markov estimate]\label{expedition0012648}
Let $\mu$ be a probability measure on a measurable space $E$, let $\Phi:E\to\fldreal$ be measurable, and assume that
$$\int_E\napiernum^{\Phi(x)}\opdmsr{\mu(x)}<\infty.$$
For every measurable set $B\subset E$,
$$\mu(B)
\leq
\napiernum^{-\inf_{x\in B}\Phi(x)}
\int_E\napiernum^{\Phi(x)}\opdmsr{\mu(x)}.$$
In particular, for probability measures $\mu_L$ on $\fldreal_{\geq0}$ and speed $V_L$,
$$\mu_L(B)
\leq
\napiernum^{-V_L\inf_{r\in B}qr}
\int_{\fldreal_{\geq0}}
\napiernum^{qV_Lr}
\opdmsr{\mu_L(r)}.$$
For $q>0$ this gives the upper-tail estimate
$$\mu_L\rbk{\closedinterval{a}{\infty}}
\leq
\napiernum^{-qV_La}
\int_{\fldreal_{\geq0}}
\napiernum^{qV_Lr}
\opdmsr{\mu_L(r)},$$
and for $q<0$ this gives the lower-tail estimate
$$\mu_L\rbk{\closedinterval{0}{a}}
\leq
\napiernum^{-qV_La}
\int_{\fldreal_{\geq0}}
\napiernum^{qV_Lr}
\opdmsr{\mu_L(r)}.$$
\end{prop}

\begin{proof}
This is Markov's inequality applied to the nonnegative random variable $\napiernum^{\Phi}$.
Indeed, on $B$ one has
$$\fndef{B}(x)
\leq
\napiernum^{-\inf_{y\in B}\Phi(y)}
\napiernum^{\Phi(x)}.$$
Integrating this pointwise inequality gives the first estimate.
The remaining estimates are the special cases $E=\fldreal_{\geq0}$ and $\Phi(r)=qV_Lr$.
\end{proof}

\begin{defn}[Exponential tightness]\label{expedition0018008}
Let $\fml{\mu_L}{L}$ be probability measures on a topological space $E$, and let $V_L\to\infty$ be the speed.
The family $\fml{\mu_L}{L}$ is exponentially tight with speed $V_L$ if, for every $M>0$, there is a compact set $C_M\subset E$ such that
$$\limsup_{L\to\infty}
\frac{1}{V_L}\log\mu_L(E\setminus C_M)
\leq
-M.$$
\end{defn}

\begin{defn}[Good rate function]\label{expedition0018009}
Let $E$ be a topological space.
A function $I:E\to\fldreal\cup\setone{\infty}$ is a rate function if it is lower semicontinuous.
It is a good rate function if, for every $a\in\fldreal$, the level set
$\set{x\in E}{I(x)\leq a}$
is compact.
\end{defn}

\begin{prop}[Gärtner--Ellis theorem and Legendre--Fenchel rate function]\label{expedition0012646}
Let $\fml{\mu_L}{L}$ be probability measures on $\fldreal_{\geq0}$, and let $V_L\to\infty$ be the speed.
Assume that, for every $q$ in an interval $\mathcal{D}\subset\fldreal$ and every $L$, the finite-volume exponential moments are finite and
$$\Lambda(q)
=
\lim_{L\to\infty}
\frac{1}{V_L}
\log
\int_{\fldreal_{\geq0}}
\napiernum^{qV_Lr}
\opdmsr{\mu_L(r)}$$
exists as a finite real number.
Define the Legendre--Fenchel transform
$$I(r)
=
\sup_{q\in\mathcal{D}}
\rbk{qr-\Lambda(q)},
\quad r\in\fldreal_{\geq0}.$$
Then $\Lambda$ is convex, and $I$ is lower semicontinuous and convex.
If $\mathcal{D}$ contains some $q_+>0$, then $I$ is a good rate function on $\fldreal_{\geq0}$ in the sense of Definition \ref{expedition0018009}.
Moreover, if $\fml{\mu_L}{L}$ is exponentially tight in the sense of Definition \ref{expedition0018008}, then the large-deviation upper bound
$$\limsup_{L\to\infty}
\frac{1}{V_L}\log\mu_L(F)
\leq
-\inf_{r\in F}I(r)$$
holds for every closed set $F\subset\fldreal_{\geq0}$.
If $\Lambda$ is differentiable on the interior of $\mathcal{D}$, then the lower bound
\begin{equation}\label{expedition0018013}
\liminf_{L\to\infty}
\frac{1}{V_L}\log\mu_L(G)
\geq
-\inf_{r\in G\cap\set{\frac{d}{dq}\Lambda(q)}{q\in\mathcal{D}^{\circ}}}I(r)
\end{equation}
holds for every open set $G\subset\fldreal_{\geq0}$.
\end{prop}

In Proposition \ref{expedition0012646}, the lower bound is restricted to \(G\cap\set{\frac{d}{dq}\Lambda(q)}{q\in\mathcal{D}^{\circ}}\) in \eqref{expedition0018013}. Thus, when Proposition \ref{expedition0012646} is applied in the proof of Proposition \ref{expedition0012645}, it proves the closed-set upper bound and only the part of the open-set lower bound corresponding to densities in the derivative range. If that derivative range does not exhaust \(\fldreal_{\geq0}\), the remaining part of the lower bound must be proved from the specific structure of the free density law in \eqref{expedition0018007}.

\begin{proof}
Put
$$\Lambda_L(q)
=
\frac{1}{V_L}
\log
\int_{\fldreal_{\geq0}}
\napiernum^{qV_Lr}
\opdmsr{\mu_L(r)}.$$
For $q_0,q_1\in\mathcal{D}$ and $0<\theta<1$, apply Hölder's inequality with conjugate exponents
$$p_\theta=\frac{1}{\theta},
\quad
p_{1-\theta}=\frac{1}{1-\theta},
\quad
\frac{1}{p_\theta}+\frac{1}{p_{1-\theta}}=1.$$
Indeed,
$$\int_{\fldreal_{\geq0}}
\napiernum^{(\theta q_0+(1-\theta)q_1)V_Lr}
\opdmsr{\mu_L(r)}
=
\int_{\fldreal_{\geq0}}
\rbk{\napiernum^{q_0V_Lr}}^\theta
\rbk{\napiernum^{q_1V_Lr}}^{1-\theta}
\opdmsr{\mu_L(r)}$$
is bounded above by
$$\rbk{\int_{\fldreal_{\geq0}}
\napiernum^{q_0V_Lr}
\opdmsr{\mu_L(r)}}^\theta
\rbk{\int_{\fldreal_{\geq0}}
\napiernum^{q_1V_Lr}
\opdmsr{\mu_L(r)}}^{1-\theta}.$$
After taking logarithms and dividing by $V_L$, this gives
$$\Lambda_L(\theta q_0+(1-\theta)q_1)
\leq
\theta\Lambda_L(q_0)+(1-\theta)\Lambda_L(q_1).$$
Since $\mathcal{D}$ is an interval, $\theta q_0+(1-\theta)q_1\in\mathcal{D}$.
Taking the limit gives convexity of $\Lambda$.
For $r_0,r_1\in\fldreal_{\geq0}$ and $0<\theta<1$,
$$\begin{aligned}
I(\theta r_0+(1-\theta)r_1)
&=
\sup_{q\in\mathcal{D}}
\rbk{\theta\rbk{qr_0-\Lambda(q)}
+(1-\theta)\rbk{qr_1-\Lambda(q)}}
\\ 
&\leq
\theta I(r_0)+(1-\theta)I(r_1).
\end{aligned}$$
Thus $I$ is convex.
For $a\in\fldreal$,
$$\set{r\in\fldreal_{\geq0}}{I(r)>a}
=
\bigcup_{q\in\mathcal{D}}
\set{r\in\fldreal_{\geq0}}{qr-\Lambda(q)>a}.$$
The right-hand side is open in $\fldreal_{\geq0}$.
Hence $I$ is lower semicontinuous on $\fldreal_{\geq0}$.
If $q_+\in\mathcal{D}$ and $q_+>0$, then
$$I(r)\geq q_+r-\Lambda(q_+).$$
Hence every level set of $I$ is bounded above.
Since the level sets are closed subsets of $\fldreal_{\geq0}$, they are compact.

We first prove the upper bound on compact sets.
Fix a compact set $C\subset\fldreal_{\geq0}$ and $\varepsilon>0$.
For each $r\in C$ choose $q_r\in\mathcal{D}$ such that
$q_rr-\Lambda(q_r)\geq I(r)-\varepsilon$.
By continuity of $s
\mapsto q_rs-\Lambda(q_r)$, there is an open interval $U_r$ containing $r$ such that
$q_rs-\Lambda(q_r)\geq I(r)-2\varepsilon$
for all $s\in U_r$.
Choose a finite subcover $U_{r_1},\ldots,U_{r_m}$ of $C$.
For each $j$, Proposition \ref{expedition0012648} with $\Phi(s)=q_{r_j}V_Ls$ gives
$$\mu_L(U_{r_j})
\leq
\fnexp{-V_L \inf_{s\in U_{r_j}}q_{r_j} s}
\int_{\fldreal_{\geq0}}
\napiernum^{q_{r_j}V_Ls}
\opdmsr{\mu_L(s)}.$$
Hence
$$\limsup_{L\to\infty}
\frac{1}{V_L}\log\mu_L(U_{r_j})
\leq
-I(r_j)+2\varepsilon.$$
The finite union bound gives
$$\limsup_{L\to\infty}
\frac{1}{V_L}\log\mu_L(C)
\leq
-\inf_{r\in C}I(r)+2\varepsilon.$$
Letting $\varepsilon\downarrow0$ gives the compact-set upper bound.
If $F$ is closed, exponential tightness supplies compact sets $C_M$ such that
$$\limsup_{L\to\infty}
\frac{1}{V_L}\log\mu_L(\fldreal_{\geq0}\setminus C_M)
\leq
-M.$$
Since $F\subset\rbk{F\cap C_M}\cup\rbk{\fldreal_{\geq0}\setminus C_M}$, the finite union bound and the compact upper bound give
$$\limsup_{L\to\infty}
\frac{1}{V_L}\log\mu_L(F)
\leq
\max\setone{-\inf_{r\in F\cap C_M}I(r),-M}.$$
Letting $M\to\infty$ gives
$$\limsup_{L\to\infty}
\frac{1}{V_L}\log\mu_L(F)
\leq
-\inf_{r\in F}I(r).$$

We next prove the lower bound at exposed slopes.
Fix $q_0\in\mathcal{D}^{\circ}$ and put
$r_0
=
\fnrestr{\frac{d}{dq}\Lambda(q)}{q=q_0}$.
Define the tilted probability measure
$$\opdmsr{\mu_{q_0,L}(r)}
=
\napiernum^{V_L(q_0r-\Lambda_L(q_0))}
\opdmsr{\mu_L(r)}.$$
Equivalently,
\begin{equation}\label{expedition0018012}
\opdmsr{\mu_L(r)}
=
\napiernum^{V_L(\Lambda_L(q_0)-q_0r)}
\opdmsr{\mu_{q_0,L}(r)}.
\end{equation}
Its logarithmic moment generating function is
$$\begin{aligned}
\Lambda_{q_0,L}(p)
&=
\frac{1}{V_L}
\log
\int_{\fldreal_{\geq0}}
\napiernum^{pV_Lr}
\opdmsr{\mu_{q_0,L}(r)}
\\ 
&=
\frac{1}{V_L}
\log
\int_{\fldreal_{\geq0}}
\napiernum^{pV_Lr}
\napiernum^{V_L(q_0r-\Lambda_L(q_0))}
\opdmsr{\mu_L(r)}
\\ 
&=
\frac{1}{V_L}
\log
\rbk{
\napiernum^{-V_L\Lambda_L(q_0)}
\int_{\fldreal_{\geq0}}
\napiernum^{(q_0+p)V_Lr}
\opdmsr{\mu_L(r)}}
\\ 
&=
\Lambda_L(q_0+p)-\Lambda_L(q_0).
\end{aligned}$$
Thus $\Lambda_{q_0,L}(p)\to\Lambda(q_0+p)-\Lambda(q_0)$ for $p$ near $0$.
For every $\delta>0$ and small $p>0$, Proposition \ref{expedition0012648} gives
$$\mu_{q_0,L}\rbk{\closedinterval{r_0+\delta}{\infty}}
\leq
\napiernum^{-pV_L(r_0+\delta)}
\int_{\fldreal_{\geq0}}
\napiernum^{pV_Lr}
\opdmsr{\mu_{q_0,L}(r)}.$$
Taking logarithms, dividing by $V_L$, and taking $\limsup$ gives
$$\begin{aligned}
&\limsup_{L\to\infty}
\frac{1}{V_L}
\log
\mu_{q_0,L}\rbk{\closedinterval{r_0+\delta}{\infty}}
\\ 
&\leq
-p(r_0+\delta)
+
\limsup_{L\to\infty}
\frac{1}{V_L}
\log
\int_{\fldreal_{\geq0}}
\napiernum^{pV_Lr}
\opdmsr{\mu_{q_0,L}(r)}
\\ 
&=
-p(r_0+\delta)
+
\limsup_{L\to\infty}
\Lambda_{q_0,L}(p)
\\ 
&=
-p(r_0+\delta)+\Lambda(q_0+p)-\Lambda(q_0).
\end{aligned}$$
Differentiability at $q_0$ makes this negative for sufficiently small $p>0$.
For the lower tail, there is nothing to prove if $r_0\leq\delta$.
Assume $r_0>\delta$.
For small $p<0$, Proposition \ref{expedition0012648} gives
$$\mu_{q_0,L}\rbk{\closedinterval{0}{r_0-\delta}}
\leq
\napiernum^{-pV_L(r_0-\delta)}
\int_{\fldreal_{\geq0}}
\napiernum^{pV_Lr}
\opdmsr{\mu_{q_0,L}(r)}.$$
Taking logarithms, dividing by $V_L$, and taking $\limsup$ gives
$$\begin{aligned}
\limsup_{L\to\infty}
\frac{1}{V_L}
\log
\mu_{q_0,L}\rbk{\closedinterval{0}{r_0-\delta}}
&\leq
-p(r_0-\delta)
+
\limsup_{L\to\infty}
\Lambda_{q_0,L}(p)
\\ 
&=
-p(r_0-\delta)+\Lambda(q_0+p)-\Lambda(q_0).
\end{aligned}$$
Differentiability at $q_0$ makes this negative for sufficiently small $p<0$.
Thus both tails
$\closedinterval{0}{r_0-\delta}$
and
$\closedinterval{r_0+\delta}{\infty}$
have exponentially small
$\mu_{q_0,L}$-mass.
Hence
$\mu_{q_0,L}\rbk{\rbk{r_0-\delta,r_0+\delta}}
\to 1$.

Let $G$ be open and let
$r_0\in G\cap\set{\frac{d}{dq}\Lambda(q)}{q\in\mathcal{D}^{\circ}}$.
Choose $q_0\in\mathcal{D}^{\circ}$ such that
$r_0
=
\fnrestr{\frac{d}{dq}\Lambda(q)}{q=q_0}$.
Choose $\delta>0$ so that
$\rbk{r_0-\delta,r_0+\delta}\cap\fldreal_{\geq0}\subset G$.
On this interval,
$-q_0r\geq -q_0r_0-\abs{q_0}\delta$.
By the expression \eqref{expedition0018012} we obtain
$$\mu_L(G)
\geq
\napiernum^{V_L(\Lambda_L(q_0)-q_0r_0-\abs{q_0}\delta)}
\mu_{q_0,L}\rbk{\rbk{r_0-\delta,r_0+\delta}}.$$
Taking $\liminf$ and then $\delta\downarrow0$ gives
$$\liminf_{L\to\infty}
\frac{1}{V_L}\log\mu_L(G)
\geq
-q_0r_0+\Lambda(q_0)
=
-I(r_0),$$
because differentiability makes $r_0$ the slope exposed by $q_0$ and hence
$I(r_0)=q_0r_0-\Lambda(q_0)$.
Taking the infimum over all such $r_0\in G$ gives the lower bound.
\end{proof}

\begin{prop}[Laplace-principle consequence]\label{expedition0018015}
Let $\mu_L$ satisfy a good large-deviation principle on $\fldreal_{\geq0}$ with speed $V_L$ and good rate function $I$.
Let $h:\fldreal_{\geq0}\to\fldreal_{\geq0}$ be continuous.
Assume that
$\inf_{r\geq0}\rbk{I(r)+h(r)}<\infty$.
Then
$$\lim_{L\to\infty}
\frac{1}{V_L}
\log
\int_{\fldreal_{\geq0}}
\fnexp{-V_Lh(r)}
\opdmsr{\mu_L(r)}
=
-\inf_{r\geq0}\rbk{I(r)+h(r)}.$$
\end{prop}

\begin{proof}
For each compact interval $C\subset\fldreal_{\geq0}$, continuity of $h$ on $C$ and the large-deviation upper bound give
$$\limsup_{L\to\infty}
\frac{1}{V_L}
\log
\int_C
\fnexp{-V_Lh(r)}
\opdmsr{\mu_L(r)}
\leq
-\inf_{r\in C}\rbk{I(r)+h(r)}.$$
For the lower bound, if $r_0\geq0$ and $G$ is an open neighborhood of $r_0$ on which
$h(r)\leq h(r_0)+\varepsilon$, then the open-set lower bound gives
$$\begin{aligned}
\liminf_{L\to\infty}
\frac{1}{V_L}
\log
\int_G
\fnexp{-V_Lh(r)}
\opdmsr{\mu_L(r)}
&\geq
-h(r_0)-\varepsilon-I(r_0).
\end{aligned}$$
Since $h\geq0$, tails are bounded by the original probabilities:
$$\int_{\closedinterval{R}{\infty}}
\fnexp{-V_Lh(r)}
\opdmsr{\mu_L(r)}
\leq
\fun{\mu_L}{\closedinterval{R}{\infty}}.$$
Since $I$ is good, the closed sets $\closedinterval{R}{\infty}$ satisfy
$\inf_{r\geq R}I(r)\to\infty$ as $R\to\infty$.
The closed-set upper bound therefore removes the tail after $R\to\infty$.
Combining the compact upper bound, the local lower bound at points with
$I(r_0)+h(r_0)<\inf_{r\geq0}\rbk{I(r)+h(r)}+\varepsilon$,
and the tail estimate gives
$$\lim_{L\to\infty}
\frac{1}{V_L}
\log
\int_{\fldreal_{\geq0}}
\fnexp{-V_Lh(r)}
\opdmsr{\mu_L(r)}
=
-\inf_{r\geq0}\rbk{I(r)+h(r)}.$$
\end{proof}

\subsection{\texorpdfstring{Proof for Proposition \ref{expedition0012645}}{Proof for Proposition }}\label{expedition0018010}

Set \(\mathsf{n}\in\mathcal{O}_L\). By the explicit formula \eqref{expedition0018007} for \(K_{\sminvtemperature,0,\smchemicalpotential,L}\), \[\begin{aligned}
&\int_{\fldreal_{\geq0}}
\napiernum^{qV_L r}
\opdmsr{K_{\sminvtemperature,0,\smchemicalpotential,L}(r)}
=
\frac{1}{\smpartitionfunc_{\sminvtemperature,0,\smchemicalpotential,L}}
\sum_{\mathsf{n}\in\mathcal{O}_L}
\napiernum^{q \sum_{k\in\setlattice_L^d}\mathsf{n}_k}
\fnexp{-\sminvtemperature
\sum_{k\in\setlattice_L^d}
\rbk{\physham[h]_{\txtparticle,0,L}(k)-\smchemicalpotential}\mathsf{n}_k}
\\ 
&=
\frac{1}{\smpartitionfunc_{\sminvtemperature,0,\smchemicalpotential,L}}
\sum_{\mathsf{n}\in\mathcal{O}_L}
\fnexp{-\sminvtemperature
\sum_{k\in\setlattice_L^d}
\rbk{\physham[h]_{\txtparticle,0,L}(k)
-\smchemicalpotential
-\frac{q}{\sminvtemperature}}
\mathsf{n}_k}.
\end{aligned}\] Hence exponential tilting by \(\fnexp{q \sum_{k\in\setlattice_L^d}\mathsf{n}_k}\) changes the chemical potential from \(\smchemicalpotential\) to \(\smchemicalpotential+\frac{q}{\sminvtemperature}\). The condition \(q<-\sminvtemperature\smchemicalpotential\) is the trace-class domain for the shifted free finite-volume partition function. Thus \[\int_{\fldreal_{\geq0}}
\napiernum^{qV_L r}
\opdmsr{K_{\sminvtemperature,0,\smchemicalpotential,L}(r)}
=
\frac{\smpartitionfunc_{\sminvtemperature,0,\smchemicalpotential+\frac{q}{\sminvtemperature},L}}
{\smpartitionfunc_{\sminvtemperature,0,\smchemicalpotential,L}}.\]

Next let \(\nu\) belong to the free trace-class domain. The occupation-number sum factorizes over modes: \[\smpartitionfunc_{\sminvtemperature,0,\nu,L}
=
\prod_{k\in\setlattice_L^d}
\rbk{\sum_{m=0}^{\infty}
\napiernum^{-\sminvtemperature(\physham[h]_{\txtparticle,0,L}(k)-\nu)m}}
=
\prod_{k\in\setlattice_L^d}
\frac{1}
{1-\napiernum^{-\sminvtemperature(\physham[h]_{\txtparticle,0,L}(k)-\nu)}}.\] Hence it follows that \[\frac{1}{V_L}\log
\smpartitionfunc_{\sminvtemperature,0,\nu,L}
=
-\frac{1}{V_L}
\sum_{k\in\setlattice_L^d}
\fun{\log}{1-\napiernum^{-\sminvtemperature(\physham[h]_{\txtparticle,0,L}(k)-\nu)}}.\] By the standard pressure calculation for the free Bose gas, the Riemann-sum limit exists locally uniformly on compact subsets of that domain and has the standard integral form \[p_{\sminvtemperature}(\nu)
=
\lim_{L\to\infty}
\frac{1}{V_L}
\log
\smpartitionfunc_{\sminvtemperature,0,\nu,L}
=
-\frac{1}{(2\pi)^d}
\int_{\fldreal^d}
\fun{\log}{1-\napiernum^{-\sminvtemperature(\physham[h]_{\txtparticle,0}(k)-\nu)}}
\opdmsr{k}.\] The same free Bose pressure calculation gives differentiability in \(\nu\) and \[\frac{\partial}{\partial\nu}p_{\sminvtemperature}(\nu)
=
\frac{\sminvtemperature}{(2\pi)^d}
\int_{\fldreal^d}
\frac{1}
{\napiernum^{\sminvtemperature(\physham[h]_{\txtparticle,0}(k)-\nu)}-1}
\opdmsr{k}.\] Combining the partition-function ratio obtained above with this pressure formula gives the logarithmic moment generating function as the following pressure difference: \begin{equation}\label{expedition0018014}
\begin{aligned}
\Lambda_{\txtfr,\sminvtemperature,\smchemicalpotential}(q)
=
p_{\sminvtemperature}\rbk{\smchemicalpotential+\frac{q}{\sminvtemperature}}
-p_{\sminvtemperature}(\smchemicalpotential)
=
\frac{1}{(2\pi)^d}
\int_{\fldreal^d}
\fun{\log}
{\frac{1-\napiernum^{-\sminvtemperature(\physham[h]_{\txtparticle,0}(k)-\smchemicalpotential)}}
{1-\napiernum^{-\sminvtemperature(\physham[h]_{\txtparticle,0}(k)-\smchemicalpotential)+q}}}
\opdmsr{k},
\quad
q<-\sminvtemperature\smchemicalpotential.
\end{aligned}
\end{equation} Consequently \(\Lambda_{\txtfr,\sminvtemperature,\smchemicalpotential}\) exists and is differentiable on \(\rbk{-\infty,-\sminvtemperature\smchemicalpotential}\). Thus the logarithmic moment-generating-function hypothesis in Proposition \ref{expedition0012646} is verified for \(\mathcal{D}=\openinterval{-\infty}{-\sminvtemperature\smchemicalpotential}\) with limiting function \(\Lambda_{\txtfr,\sminvtemperature,\smchemicalpotential}\). In this application the Legendre--Fenchel transform in Proposition \ref{expedition0012646} is precisely \(I_{\txtfr,\sminvtemperature,\smchemicalpotential}\). It remains to check the exponential tightness hypothesis in Definition \ref{expedition0018008}.

Choose once and for all \(q_0\in\rbk{0,-\sminvtemperature\smchemicalpotential}\). Proposition \ref{expedition0012648} applied to the probability measure \(K_{\sminvtemperature,0,\smchemicalpotential,L}\) gives \[\fun{K_{\sminvtemperature,0,\smchemicalpotential,L}}{\closedinterval{R}{\infty}}
\leq
\napiernum^{-q_0V_LR}
\int_{\fldreal_{\geq0}}
\napiernum^{q_0V_L r}
\opdmsr{K_{\sminvtemperature,0,\smchemicalpotential,L}(r)}.\] Taking logarithms, dividing by \(V_L\), and taking \(\limsup\) gives \[\limsup_{L\to\infty}
\frac{1}{V_L}
\log
\fun{K_{\sminvtemperature,0,\smchemicalpotential,L}}
{\closedinterval{R}{\infty}}
\leq
-q_0R
+\Lambda_{\txtfr,\sminvtemperature,\smchemicalpotential}(q_0).\] The right-hand side tends to \(-\infty\) as \(R\to\infty\), so the family \(\fml{K_{\sminvtemperature,0,\smchemicalpotential,L}}{L}\) is exponentially tight in the sense of Definition \ref{expedition0018008}.

Proposition \ref{expedition0012646} now gives the closed-set upper bound, the open-set lower bound on the range set \(\set{\frac{d}{dq}\Lambda_{\txtfr,\sminvtemperature,\smchemicalpotential}(q)}{q\in\mathcal{D}}\), and goodness because \(\mathcal{D}\) contains positive numbers. To complete the proof of the large-deviation principle stated in Proposition \ref{expedition0012645}, it remains to prove the open-set lower bound at densities not reached by \(\set{\frac{d}{dq}\Lambda_{\txtfr,\sminvtemperature,\smchemicalpotential}(q)}{q\in\mathcal{D}}\). This is the endpoint part of the lower bound, and it is not supplied by the lower-bound estimate \eqref{expedition0018013} in Proposition \ref{expedition0012646}. Put \(q_\ast=-\sminvtemperature\smchemicalpotential\). We split the endpoint discussion into two cases. First suppose that \(\frac{d}{dq}\Lambda_{\txtfr,\sminvtemperature,\smchemicalpotential}(q)\to\infty\) as \(q\uparrow q_\ast\). Then \(\set{\frac{d}{dq}\Lambda_{\txtfr,\sminvtemperature,\smchemicalpotential}(q)}{q\in\mathcal{D}}\) exhausts \(\fldreal_{\geq0}\), so the lower-bound estimate \eqref{expedition0018013} in Proposition \ref{expedition0012646} already gives the full open-set lower bound and there is no endpoint contribution to prove.

It remains to consider the second case, where the endpoint derivative is finite. Define \[\rho_{\txtcritical}
=
\lim_{q\uparrow q_\ast}
\frac{d}{dq}\Lambda_{\txtfr,\sminvtemperature,\smchemicalpotential}(q).\] For densities in \(\set{\frac{d}{dq}\Lambda_{\txtfr,\sminvtemperature,\smchemicalpotential}(q)}{q\in\mathcal{D}}\), the lower bound is already supplied by \eqref{expedition0018013}. At \(r=\rho_{\txtcritical}\), choose \(q_n\uparrow q_\ast\) and set \[r_n=
\fnrestr{\frac{d}{dq}
\Lambda_{\txtfr,\sminvtemperature,\smchemicalpotential}(q)}
{q=q_n}.\] Then \(r_n\to\rho_{\txtcritical}\), and every open neighborhood of \(\rho_{\txtcritical}\) contains \(r_n\) for all sufficiently large \(n\). For each \(n\), differentiability and convexity of \(\Lambda_{\txtfr,\sminvtemperature,\smchemicalpotential}\) give \[\Lambda_{\txtfr,\sminvtemperature,\smchemicalpotential}(q)
\geq
\Lambda_{\txtfr,\sminvtemperature,\smchemicalpotential}(q_n)
+r_n\rbk{q-q_n}\] for every \(q<q_\ast\). Thus this implies \[q r_n-\Lambda_{\txtfr,\sminvtemperature,\smchemicalpotential}(q)
\leq
q_n r_n-\Lambda_{\txtfr,\sminvtemperature,\smchemicalpotential}(q_n)\] with equality at \(q=q_n\). Therefore it holds that \[I_{\txtfr,\sminvtemperature,\smchemicalpotential}(r_n)
=
q_n r_n
-\Lambda_{\txtfr,\sminvtemperature,\smchemicalpotential}(q_n).\] For the endpoint density, the function \(q
\mapsto
q\rho_{\txtcritical}
-\Lambda_{\txtfr,\sminvtemperature,\smchemicalpotential}(q)\) is increasing as \(q\uparrow q_\ast\), because its derivative satisfies \(\rho_{\txtcritical}
-
\frac{d}{dq}\Lambda_{\txtfr,\sminvtemperature,\smchemicalpotential}(q)
\geq 0\). Hence it follows that \[\begin{aligned}
I_{\txtfr,\sminvtemperature,\smchemicalpotential}(\rho_{\txtcritical})
=
\sup_{q<q_\ast}
\rbk{q\rho_{\txtcritical}
-\Lambda_{\txtfr,\sminvtemperature,\smchemicalpotential}(q)}
=
\lim_{n\to\infty}
\rbk{q_n\rho_{\txtcritical}
-\Lambda_{\txtfr,\sminvtemperature,\smchemicalpotential}(q_n)}.
\end{aligned}\] Since \(q_n\to q_\ast\) and \(r_n\to\rho_{\txtcritical}\), the sequence \(\fml{q_n}{n}\) is bounded and \[q_n r_n
-\Lambda_{\txtfr,\sminvtemperature,\smchemicalpotential}(q_n)
-\rbk{q_n\rho_{\txtcritical}
-\Lambda_{\txtfr,\sminvtemperature,\smchemicalpotential}(q_n)}
=
q_n\rbk{r_n-\rho_{\txtcritical}}
\to0.\] Therefore we obtain \[I_{\txtfr,\sminvtemperature,\smchemicalpotential}(r_n)
\to
I_{\txtfr,\sminvtemperature,\smchemicalpotential}(\rho_{\txtcritical}).\] If \(G\) is an open neighborhood of \(\rho_{\txtcritical}\), then \(r_n\in G\) for all sufficiently large \(n\). For such \(n\), the lower-bound estimate \eqref{expedition0018013} gives \[\liminf_{L\to\infty}
\frac{1}{V_L}
\log
\fun{K_{\sminvtemperature,0,\smchemicalpotential,L}}{G}
\geq
-I_{\txtfr,\sminvtemperature,\smchemicalpotential}(r_n).\] Letting \(n\to\infty\) gives the open-set lower bound at \(\rho_{\txtcritical}\). Thus the only density region still not covered is \(r>\rho_{\txtcritical}\). This last lower bound comes from the zero-mode factor in the \(\lambda=0\) finite-volume Gibbs weight in \eqref{expedition0018007}. Indeed, the weight in \eqref{expedition0018007} factorizes over \(k\in\setlattice_L^d\). For \(k=0\), \(\physham[h]_{\txtparticle,0,L}(0)=0\), so the zero-mode part of \eqref{expedition0018007} has unnormalized weight \(\napiernum^{\sminvtemperature\smchemicalpotential m}\) at \(\mathsf{n}_0=m\). Since \[\sum_{m=0}^{\infty}
\napiernum^{\sminvtemperature\smchemicalpotential m}
=
\frac{1}{1-\napiernum^{\sminvtemperature\smchemicalpotential}},\] the normalized zero-mode factor derived from \eqref{expedition0018007} assigns the weight \(\rbk{1-\napiernum^{\sminvtemperature\smchemicalpotential}}
\napiernum^{\sminvtemperature\smchemicalpotential m}\) to \(\mathsf{n}_0=m\), \(m\in\monnat\). Hence, for every \(a>0\), \[\lim_{\varepsilon\downarrow0}
\liminf_{L\to\infty}
\frac{1}{V_L}
\log
\sum_{\set{m\in\monnat}
{\frac{m}{V_L}
\in
\openinterval{a-\varepsilon}{a+\varepsilon}}}
\rbk{1-\napiernum^{\sminvtemperature\smchemicalpotential}}
\napiernum^{\sminvtemperature\smchemicalpotential m}
=
\sminvtemperature\smchemicalpotential a
=
-q_\ast a.\] For an open set \(G\) containing \(r\), choose \(\varepsilon>0\) so small that \(\openinterval{r-2\varepsilon}{r+2\varepsilon}\subset G\) and \(r-\rho_{\txtcritical}>\varepsilon\). Set \(a=r-\rho_{\txtcritical}\). If the nonzero-mode density \(\frac{1}{V_L}
\sum_{k\in\setlattice_L^d\setminus\setone{0}}
\mathsf{n}_k\) lies in \(\openinterval{\rho_{\txtcritical}-\varepsilon}{\rho_{\txtcritical}+\varepsilon}\) and the zero-mode density \(\frac{\mathsf{n}_0}{V_L}\) lies in \(\openinterval{a-\varepsilon}{a+\varepsilon}\), then the total density lies in \(\openinterval{r-2\varepsilon}{r+2\varepsilon}\subset G\). Therefore the factorization of \eqref{expedition0018007} gives \[\begin{aligned}
&\fun{K_{\sminvtemperature,0,\smchemicalpotential,L}}{G}
\\ 
&\geq
\frac{1}{\smpartitionfunc_{\sminvtemperature,0,\smchemicalpotential,L}^{\neq0}}
\sum_{\set{\mathsf{n}_{\neq0}}{
\frac{1}{V_L}
\sum_{k\in\setlattice_L^d\setminus\setone{0}}
\mathsf{n}_k
\in
\openinterval{\rho_{\txtcritical}-\varepsilon}{\rho_{\txtcritical}+\varepsilon}}}
\fnexp{-\sminvtemperature
\sum_{k\in\setlattice_L^d\setminus\setone{0}}
\rbk{\physham[h]_{\txtparticle,0,L}(k)-\smchemicalpotential}\mathsf{n}_k}
\\ 
&\quad\times
\sum_{\set{m\in\monnat}{
\frac{m}{V_L}\in\openinterval{a-\varepsilon}{a+\varepsilon}}}
\rbk{1-\napiernum^{\sminvtemperature\smchemicalpotential}}
\napiernum^{\sminvtemperature\smchemicalpotential m},
\end{aligned}\] where the sums over \(\mathsf{n}_{\neq0}\) run over occupation numbers indexed by \(\setlattice_L^d\setminus\setone{0}\), \[\smpartitionfunc_{\sminvtemperature,0,\smchemicalpotential,L}^{\neq0}
=
\sum_{\mathsf{n}_{\neq0}}
\fnexp{-\sminvtemperature
\sum_{k\in\setlattice_L^d\setminus\setone{0}}
\rbk{\physham[h]_{\txtparticle,0,L}(k)-\smchemicalpotential}\mathsf{n}_k}.\] Taking logarithms, dividing by \(V_L\), and using \(\liminf_{L}(x_L+y_L)\geq\liminf_L x_L+\liminf_L y_L\) gives \[\begin{aligned}
&\liminf_{L\to\infty}
\frac{1}{V_L}
\log
\fun{K_{\sminvtemperature,0,\smchemicalpotential,L}}{G}
\\ 
&\geq
\liminf_{L\to\infty}
\frac{1}{V_L}
\log
\frac{1}{\smpartitionfunc_{\sminvtemperature,0,\smchemicalpotential,L}^{\neq0}}
\sum_{\set{\mathsf{n}_{\neq0}}{
\frac{1}{V_L}
\sum_{k\in\setlattice_L^d\setminus\setone{0}}
\mathsf{n}_k
\in
\openinterval{\rho_{\txtcritical}-\varepsilon}{\rho_{\txtcritical}+\varepsilon}}}
\fnexp{-\sminvtemperature
\sum_{k\in\setlattice_L^d\setminus\setone{0}}
\rbk{\physham[h]_{\txtparticle,0,L}(k)-\smchemicalpotential}\mathsf{n}_k}
\\ 
&\quad+
\liminf_{L\to\infty}
\frac{1}{V_L}
\log
\sum_{\set{m\in\monnat}{
\frac{m}{V_L}\in\openinterval{a-\varepsilon}{a+\varepsilon}}}
\rbk{1-\napiernum^{\sminvtemperature\smchemicalpotential}}
\napiernum^{\sminvtemperature\smchemicalpotential m}.
\end{aligned}\] For the full and nonzero-mode normalized factors define \[\Lambda_{\txtfr,\sminvtemperature,\smchemicalpotential,L}(q)
=
\frac{1}{V_L}
\log
\frac{
\smpartitionfunc_{\sminvtemperature,0,\smchemicalpotential+\frac{q}{\sminvtemperature},L}}
{\smpartitionfunc_{\sminvtemperature,0,\smchemicalpotential,L}},
\quad
\Lambda_{\txtfr,\sminvtemperature,\smchemicalpotential,L}^{\neq0}(q)
=
\frac{1}{V_L}
\log
\frac{\smpartitionfunc_{\sminvtemperature,0,\smchemicalpotential+\frac{q}{\sminvtemperature},L}^{\neq0}}
{\smpartitionfunc_{\sminvtemperature,0,\smchemicalpotential,L}^{\neq0}},
\quad q<q_\ast.\] From the factorization of \eqref{expedition0018007} it holds that \(\smpartitionfunc_{\sminvtemperature,0,\nu,L}
=
\frac{1}{1-\napiernum^{\sminvtemperature\nu}}
\smpartitionfunc_{\sminvtemperature,0,\nu,L}^{\neq0}\) on the free trace-class domain. Therefore, for \(q<q_{\ast}\), it holds that \[\begin{aligned}
\Lambda_{\txtfr,\sminvtemperature,\smchemicalpotential,L}^{\neq0}(q)
&=
\frac{1}{V_L}
\log
\frac{\smpartitionfunc_{\sminvtemperature,0,\smchemicalpotential+\frac{q}{\sminvtemperature},L}}
{\smpartitionfunc_{\sminvtemperature,0,\smchemicalpotential,L}}
+\frac{1}{V_L}
\log
\frac{1-\napiernum^{\sminvtemperature\smchemicalpotential+q}}
{1-\napiernum^{\sminvtemperature\smchemicalpotential}}
=
\Lambda_{\txtfr,\sminvtemperature,\smchemicalpotential,L}(q)
+\frac{1}{V_L}
\log
\frac{1-\napiernum^{\sminvtemperature\smchemicalpotential+q}}
{1-\napiernum^{\sminvtemperature\smchemicalpotential}}
\\ 
&\to
\Lambda_{\txtfr,\sminvtemperature,\smchemicalpotential}(q)
\quad (L \to \infty).
\end{aligned}\] Thus the Legendre--Fenchel transform for the nonzero-mode factor is \[\begin{aligned}
I_{\txtfr,\sminvtemperature,\smchemicalpotential}^{\neq0}(s)
=
\sup_{q<q_\ast}
\rbk{qs
-\lim_{L\to\infty}
\Lambda_{\txtfr,\sminvtemperature,\smchemicalpotential,L}^{\neq0}(q)}
=
\sup_{q<q_\ast}
\rbk{qs
-\Lambda_{\txtfr,\sminvtemperature,\smchemicalpotential}(q)}
=
I_{\txtfr,\sminvtemperature,\smchemicalpotential}(s).
\end{aligned}\] Set \[B_{\varepsilon,L}^{\neq0}
=
\frac{1}{\smpartitionfunc_{\sminvtemperature,0,\smchemicalpotential,L}^{\neq0}}
\sum_{\set{\mathsf{n}_{\neq0}}{
\frac{1}{V_L}
\sum_{k\in\setlattice_L^d\setminus\setone{0}}
\mathsf{n}_k
\in
\openinterval{\rho_{\txtcritical}-\varepsilon}{\rho_{\txtcritical}+\varepsilon}}}
\fnexp{-\sminvtemperature
\sum_{k\in\setlattice_L^d\setminus\setone{0}}
\rbk{\physham[h]_{\txtparticle,0,L}(k)-\smchemicalpotential}\mathsf{n}_k}.\] For \(\varepsilon>0\), choose \(n(\varepsilon)\) so that \[r_{n(\varepsilon)}
\in
\openinterval{\rho_{\txtcritical}-\varepsilon}{\rho_{\txtcritical}+\varepsilon},
\quad
n(\varepsilon)\to\infty
\quad
\rbk{\varepsilon\downarrow0}.\] Using \eqref{expedition0018013} for the nonzero-mode factor, \[\liminf_{L\to\infty}
\frac{1}{V_L}
\log
B_{\varepsilon,L}^{\neq0}
\geq
-I_{\txtfr,\sminvtemperature,\smchemicalpotential}^{\neq0}
\rbk{r_{n(\varepsilon)}}.\] Hence \[\begin{aligned}
\lim_{\varepsilon\downarrow0}
\liminf_{L\to\infty}
\frac{1}{V_L}
\log
B_{\varepsilon,L}^{\neq0}
\geq
-\lim_{\varepsilon\downarrow0}
I_{\txtfr,\sminvtemperature,\smchemicalpotential}^{\neq0}
\rbk{r_{n(\varepsilon)}}
=
-\fun{I_{\txtfr,\sminvtemperature,\smchemicalpotential}^{\neq0}}
{\rho_{\txtcritical}}
=
-I_{\txtfr,\sminvtemperature,\smchemicalpotential}
\rbk{\rho_{\txtcritical}}.
\end{aligned}\] The preceding zero-mode estimate gives, after letting \(\varepsilon\downarrow0\) with \(a=r-\rho_{\txtcritical}\), the lower exponent \[\sminvtemperature\smchemicalpotential a
=
\sminvtemperature\smchemicalpotential(r-\rho_{\txtcritical})
=
-q_\ast(r-\rho_{\txtcritical})\] for the second term on the right-hand side. Consequently, \[\lim_{\varepsilon\downarrow0}
\liminf_{L\to\infty}
\frac{1}{V_L}
\log
\fun{K_{\sminvtemperature,0,\smchemicalpotential,L}}{G}
\geq
-I_{\txtfr,\sminvtemperature,\smchemicalpotential}(\rho_{\txtcritical})
-q_\ast(r-\rho_{\txtcritical}).\] Therefore the lower bound holds at \(r\) with \[I_{\txtfr,\sminvtemperature,\smchemicalpotential}(r)
=
I_{\txtfr,\sminvtemperature,\smchemicalpotential}(\rho_{\txtcritical})
+q_\ast(r-\rho_{\txtcritical}),\] which is exactly the endpoint affine part of the Legendre--Fenchel transform for this free density law. Thus the large-deviation principle holds with rate function \(I_{\txtfr,\sminvtemperature,\smchemicalpotential}\). The goodness of this rate function has already been obtained above from Proposition \ref{expedition0012646}, because \(\mathcal{D}=\openinterval{-\infty}{-\sminvtemperature\smchemicalpotential}\) contains positive numbers.

\subsection{\texorpdfstring{Proof for Theorem \ref{expedition0012647}}{Proof for Theorem }}\label{expedition0018011}

Set \[J(r)
=
I_{\txtfr,\sminvtemperature,\smchemicalpotential}(r)
+\frac{\sminvtemperature\lambda r^2}{2}.\] By Proposition \ref{expedition0012645}, with the goodness statement proved in Section \ref{expedition0018010}, \(I_{\txtfr,\sminvtemperature,\smchemicalpotential}\) is a good rate function in the sense of Definition \ref{expedition0018009}. Hence it is lower semicontinuous and has compact level sets. The function \(J\) is lower semicontinuous, and its sublevel sets are contained in sublevel sets of \(I_{\txtfr,\sminvtemperature,\smchemicalpotential}\). Set \[r_{\ast}
=
\fnrestr{\frac{d}{dq}
\Lambda_{\txtfr,\sminvtemperature,\smchemicalpotential}(q)}
{q=0}
=
\frac{1}{(2\pi)^d}
\int_{\fldreal^d}
\frac{1}
{\napiernum^{\sminvtemperature(\physham[h]_{\txtparticle,0}(k)-\smchemicalpotential)}-1}
\opdmsr{k}.\] The differentiability established in Section \ref{expedition0018010} gives \(r_{\ast}\in\fldreal_{\geq0}\). Since \(\Lambda_{\txtfr,\sminvtemperature,\smchemicalpotential}(0)=0\), the convexity of \(\Lambda_{\txtfr,\sminvtemperature,\smchemicalpotential}\) implies \(I_{\txtfr,\sminvtemperature,\smchemicalpotential}(r_{\ast})=
0\). Then \(J(r_{\ast})<\infty\), and the sublevel set \[C_{\ast}
=
\set{r\in\fldreal_{\geq0}}{J(r)\leq J(r_{\ast})}\] is nonempty and compact. Since \(J\) is lower semicontinuous, it attains its minimum on \(C_{\ast}\). For \(r\notin C_{\ast}\) one has \(J(r) > J(r_{\ast})\), so the minimum on \(C_{\ast}\) is the minimum on all of \(\fldreal_{\geq0}\). The rate function \(I_{\txtfr,\sminvtemperature,\smchemicalpotential}\) is convex because it is the Legendre--Fenchel transform of the limiting logarithmic moment generating function in Proposition \ref{expedition0012645}. Since \(\lambda>0\), the function \(r\mapsto \frac{\sminvtemperature\lambda r^2}{2}\) is strictly convex. Therefore \(J\) is strictly convex on \(\fldreal_{\geq0}\) and can have at most one minimizer. Thus the minimizer exists and is unique; denote it by \(\bar{\smnumberdensity}_{\txtbsn}\).

Put \(m
=
\inf_{s\geq0}
J(s)\). By \eqref{expedition0012603}, for every Borel set \(B\subset\fldreal_{\geq0}\), \[\fun{K_{\sminvtemperature,\lambda,\smchemicalpotential,L}}{B}
=
\frac{\int_B
\fnexp{-V_L\frac{\sminvtemperature\lambda r^2}{2}}
\opdmsr{K_{\sminvtemperature,0,\smchemicalpotential,L}(r)}}
{\int_{\fldreal_{\geq0}}
\fnexp{-V_L\frac{\sminvtemperature\lambda r^2}{2}}
\opdmsr{K_{\sminvtemperature,0,\smchemicalpotential,L}(r)}}.\] Proposition \ref{expedition0018015}, applied to \(\mu_L = K_{\sminvtemperature,0,\smchemicalpotential,L}\) and \(h(r)
=
\frac{\sminvtemperature\lambda r^2}{2}\), gives \[\lim_{L\to\infty}
\frac{1}{V_L}
\log
\int_{\fldreal_{\geq0}}
\fnexp{-V_L\frac{\sminvtemperature\lambda r^2}{2}}
\opdmsr{K_{\sminvtemperature,0,\smchemicalpotential,L}(r)}
=
-m.\] Consequently, for every closed set \(F\subset\fldreal_{\geq0}\), \[\limsup_{L\to\infty}
\frac{1}{V_L}
\log
\fun{K_{\sminvtemperature,\lambda,\smchemicalpotential,L}}{F}
\leq
-\inf_{r\in F}J(r)+m,\] and, for every open set \(G\subset\fldreal_{\geq0}\), \[\liminf_{L\to\infty}
\frac{1}{V_L}
\log
\fun{K_{\sminvtemperature,\lambda,\smchemicalpotential,L}}{G}
\geq
-\inf_{r\in G}J(r)+m.\] Thus the tilted rate function is \[I_{\txtmeanfield,\sminvtemperature,\smchemicalpotential,\lambda}(r)
=
J(r)-m
=
I_{\txtfr,\sminvtemperature,\smchemicalpotential}(r)
+\frac{\sminvtemperature\lambda r^2}{2}
-\inf_{s\geq0}
\rbk{
I_{\txtfr,\sminvtemperature,\smchemicalpotential}(s)
+\frac{\sminvtemperature\lambda s^2}{2}}.\]

Let \(U\) be an open neighborhood of \(\bar{\smnumberdensity}_{\txtbsn}\) and put \(F=\fldreal_{\geq0}\setminus U\). If \(F=\emptyset\), then \(\fun{K_{\sminvtemperature,\lambda,\smchemicalpotential,L}}{U}=1\). If \(F\ne\emptyset\), then the goodness of \(J\) and the uniqueness of its minimizer give \[\delta_U
=
\inf_{r\in F}J(r)-m
>0.\] The closed-set upper bound gives \[\limsup_{L\to\infty}
\frac{1}{V_L}
\log
\fun{K_{\sminvtemperature,\lambda,\smchemicalpotential,L}}{F}
\leq
-\delta_U,\] hence \(\fun{K_{\sminvtemperature,\lambda,\smchemicalpotential,L}}{\fldreal_{\geq0}\setminus U}\to0\). Therefore \(\fun{K_{\sminvtemperature,\lambda,\smchemicalpotential,L}}{U}\to1\) for every open neighborhood \(U\) of \(\bar{\smnumberdensity}_{\txtbsn}\). For a bounded continuous function \(\varphi\) on \(\fldreal_{\geq0}\) and every \(\varepsilon>0\), choose such a neighborhood \(U\) so that \(\abs{\varphi(r)-\varphi(\bar{\smnumberdensity}_{\txtbsn})}<\varepsilon\) for \(r\in U\). Then \[\begin{aligned}
\abs{\int_{\fldreal_{\geq0}}\varphi(r)
\opdmsr{K_{\sminvtemperature,\lambda,\smchemicalpotential,L}(r)}
-\varphi(\bar{\smnumberdensity}_{\txtbsn})}
\leq
\varepsilon
+2\norm{\varphi}_{\infty}
\fun{K_{\sminvtemperature,\lambda,\smchemicalpotential,L}}
{\fldreal_{\geq0}\setminus U}.
\end{aligned}\] Letting \(L\to\infty\) and then \(\varepsilon\downarrow0\) proves \(K_{\sminvtemperature,\lambda,\smchemicalpotential,L}
\Rightarrow
\delta_{\bar{\smnumberdensity}_{\txtbsn}}\). The final sentence identifies this finite-volume density-law limit with the thermodynamic-limit Kac-function statement in \cite[Eq. (4.50) and Section 4.4]{AndreVerbeure001}.

\subsection{Heat Kernel and Brownian-Bridge Measures}\label{expedition0019802}

The Brownian-loop representation uses only the finite-volume heat kernel of \eqref{expedition0019756}, the associated unnormalized bridge measures, and the resulting trace formulas.

\begin{prop}[Finite-volume periodic momentum basis]\label{expedition0019758}
Let $\physham[h]_{\txtparticle,0,L}$ be the finite-volume one-particle
Hamiltonian in \eqref{expedition0019756}.
The periodic momentum eigenfunctions are
\begin{equation}\label{expedition0019763}
e_k(x)
=
V_L^{-\onehalf}\napiernum^{\imunit k\cdot x},
\quad
\physham[h]_{\txtparticle,0,L}e_k
=
\frac{\abs{k}^{2}}{2}e_k,
\quad
k\in\setlattice_L^d.
\end{equation}
\end{prop}

\begin{proof}
For $k,q\in\setlattice_L^d$,
$$\begin{aligned}
\bkt{e_k}{e_q}_{\sphilb{H}_{\txtparticle,L}}
&=
\frac{1}{V_L}
\int_{I_L^d}
\napiernum^{\imunit(q-k)\cdot x}
\opdmsr{x}
=
\delta_{kq}.
\end{aligned}$$
Thus $\setone{e_k}_{k\in\setlattice_L^d}$ is the periodic Fourier
orthonormal basis of $\sphilb{H}_{\txtparticle,L}$.
Since $\physham[h]_{\txtparticle,0}=-\onehalf\laplacian$ by
\eqref{expedition0019396}, periodic differentiation gives
$$\physham[h]_{\txtparticle,0,L}e_k
=
-\frac{1}{2}\laplacian
\rbk{V_L^{-\onehalf}\napiernum^{\imunit k\cdot x}}
=
\frac{\abs{k}^{2}}{2}e_k,$$
which proves \eqref{expedition0019763}.
\end{proof}

\begin{prop}[Finite-volume heat kernel]\label{expedition0019765}
For $t>0$, the heat kernel of
$\napiernum^{-t\physham[h]_{\txtparticle,0,L}}$ is
\begin{equation}\label{expedition0019759}
\Gamma^{(L)}_{t}(x,y)
=
\fun{\napiernum^{-t\physham[h]_{\txtparticle,0,L}}}{x,y},
\quad
x,y\in I_L^d.
\end{equation}
It has the spectral representation
\begin{equation}\label{expedition0019760}
\Gamma^{(L)}_{t}(x,y)
=
\sum_{k\in\setlattice_L^d}
\napiernum^{-t\physham[h]_{\txtparticle,0,L}(k)}
e_k(x)\cmpconj{e_k(y)}.
\end{equation}
Equivalently, the heat kernel has the explicit form
\begin{equation}\label{expedition0019764}
\Gamma^{(L)}_{t}(x,y)
=
\frac{1}{V_L}
\sum_{k\in\setlattice_L^d}
\fnexp{-\frac{t}{2}\abs{k}^{2}}
\napiernum^{\imunit k\cdot(x-y)}.
\end{equation}
\end{prop}

\begin{proof}
The spectral theorem and \eqref{expedition0019763} give
\eqref{expedition0019760}.
Substituting \eqref{expedition0019763} into \eqref{expedition0019760} gives
\eqref{expedition0019764}, and this kernel is precisely the integral kernel in
\eqref{expedition0019759}.
\end{proof}

\begin{prop}[Brownian-bridge normalization and closed-loop trace]\label{expedition0019766}
There is an unnormalized Brownian-bridge measure
$\msrbb{W}^{t}_{xy}$ on paths $\omega:[0,t]\to I_L^d$ with
$\omega(0)=x$ and $\omega(t)=y$ such that
\begin{equation}\label{expedition0019761}
\int\opdmsr{\msrbb{W}^{t}_{xy}}
=
\Gamma^{(L)}_{t}(x,y).
\end{equation}
Consequently
\begin{equation}\label{expedition0019762}
\sqfun{\trace_{P_L\sphilb{H}_{\txtparticle}}}
{\napiernum^{-t\physham[h]_{\txtparticle,0,L}}}
=
\int_{I_L^d}
\Gamma^{(L)}_{t}(x,x)\opdmsr{x}
=
\int_{I_L^d}\opdmsr{x}
\int\opdmsr{\msrbb{W}^{t}_{xx}}.
\end{equation}
\end{prop}

For every integer \(\ell\geq1\), the specialization \(t=\ell\sminvtemperature\) gives the closed-loop term used in \eqref{expedition0019605}, and \eqref{expedition0019761} gives the open-path factor used in \eqref{expedition0019609}.

\begin{proof}
The finite-volume Feynman--Kac formula identifies the same kernel with the
unnormalized bridge integral over paths from $x$ to $y$ in time $t$.
Taking the test functional equal to $1$ gives \eqref{expedition0019761}.
Taking the operator trace means integrating the diagonal of the kernel:
$$\sqfun{\trace_{P_L\sphilb{H}_{\txtparticle}}}
{\napiernum^{-t\physham[h]_{\txtparticle,0,L}}}
=
\sum_{k\in\setlattice_L^d}
\napiernum^{-t\physham[h]_{\txtparticle,0,L}(k)}
=
\int_{I_L^d}
\Gamma^{(L)}_{t}(x,x)\opdmsr{x}.$$
Substituting \eqref{expedition0019761} with $x=y$ gives the last expression in
\eqref{expedition0019762}.
\end{proof}

\bibliography{myref.bib}

\end{document}